%% file: main.tex
\documentclass{article}
\providecommand{\keywords}[1]{\textbf{Keywords:} #1}
\usepackage[utf8]{inputenc}
\input{GrandMacros.tex}

\input{Macro.tex}

\usepackage{amsthm}
\usepackage{xcolor}
\newtheorem{remark}{Remark}
\newtheorem{proposition}{Proposition}
\usepackage{graphicx,verbatim,array,multicol, subfigure, color, lscape, mathrsfs, dsfont}
\usepackage{psfrag, amsmath, amsfonts, epsfig, fancybox, setspace, soul, amsthm, longtable, threeparttable, threeparttablex, makecell, xcolor, pdflscape, booktabs, natbib, amssymb}
\usepackage[normalem]{ulem}
\usepackage{url}
\usepackage{cancel}
\usepackage[utf8]{inputenc}
\newcommand{\indep}{\perp \!\!\! \perp}

\def\mubar{\bar{\mu}}

\newcommand{\ellbar}{\overline{\ell}}

\def\bX{\mathbf{X}}
\def\bx{\mathbf{x}}

\def\E{\mathbb{E}}
\def\P{\mathbb{P}}
\def\PP{\mathbb{P}}
\def\V{\mathbb{V}}
\def\I{\mathbb{I}}
\def\eif{\text{EIF}}

\newtheorem{assumption}{Assumption}
\newtheorem{definition}{Definition}
\newtheorem{corollary}{Corollary}
\newtheorem{theorem}{Theorem}
\newtheorem{intassumption}{Assumption}
\numberwithin{intassumption}{assumption}

\usepackage{authblk}
\usepackage{setspace}
\begin{document}

\title{Bounding causal effects with an unknown mixture of informative and non-informative missingness}

\author[1]{Max Rubinstein\thanks{Corresponding author: mrubinstein@rand.org}}
\author[2]{Denis Agniel}
\author[3]{Larry Han}
\author[4]{Marcela Horvitz-Lennon}
\author[5]{Sharon-Lise Normand}

\affil[1]{RAND Corporation, Pittsburgh, USA}
\affil[1]{RAND Corporation, Santa Monica, USA}
\affil[3]{Northeastern University, Boston, USA}
\affil[4]{RAND Corporation, Boston, USA}
\affil[5]{Harvard Medical School, Boston, USA}

\maketitle

\abstract{In experimental and observational data settings, researchers often have limited knowledge of the reasons for missing outcomes. To address this uncertainty, we propose bounds on causal effects for missing outcomes, accommodating the scenario where missingness is an unobserved mixture of informative and non-informative components. Within this mixed missingness framework, we explore several assumptions to derive bounds on causal effects, including bounds expressed as a function of user-specified sensitivity parameters. We develop influence-function based estimators of these bounds to enable flexible, non-parametric, and machine learning based estimation, achieving root-$n$ convergence rates and asymptotic normality under relatively mild conditions. We further consider the identification and estimation of bounds for other causal quantities that remain meaningful when informative missingness reflects a competing outcome, such as death. We conduct simulation studies and illustrate our methodology with a study on the causal effect of antipsychotic drugs on diabetes risk using a health insurance dataset.}




\maketitle

\section{Introduction}

Causal inference requires that the intervention precede the outcome, imposing a necessary temporal ordering. However, in practice, some outcomes may be unobservable because units drop out of the study or become inaccessible to the researcher for other reasons during the period between the intervention and the outcome measurement. This issue can arise in virtually any type of study, even in randomized controlled trials. When outcomes are missing, identification assumptions are necessary to relate causal estimands to observable data quantities. A commonly adopted assumption is the ``missing at random'' (MAR) mechanism, under which outcome missingness is independent of the value of the outcome itself, conditional on all other observed data. Thus, any association between missingness and the outcome is fully explained by observable characteristics, enabling point identification of causal effects (see, e.g., \cite{williamson2012doubly, chaudhuri2016gmm, metten2022inverse}). If instead missingness is a function of the outcomes themselves, or unobserved variables that determine both missingness and the outcomes, then the outcomes are said to be ``missing not at random'' (MNAR), and commonly targeted causal effects, such as the average treatment effect, are not in general identifiable.

Unfortunately, testing whether MAR holds versus MNAR is not possible, and MAR is often implausible in practice. For example, in randomized trials examining disease progression or mortality, unobserved deterioration in health could directly lead to both study dropout and an adverse outcome. These issues are even more pronounced in observational studies utilizing electronic health records or insurance claims data, which typically feature higher rates of missingness. In this setting, outcomes such as diagnoses or procedures might not be documented if patients seek care elsewhere or forgo treatment altogether. Additionally, the reasons for missingness in such settings are often unknown yet plausibly linked to critical downstream outcomes, such as loss of insurance coverage. Furthermore, when missingness constitutes a competing event, an outcome that prevents the primary outcome from occurring, MAR is violated, and traditional causal estimands such as the average treatment effect may not be interpretable \cite{stensrud2022separable}. 

In contrast to the substantial literature addressing potential violations of the ``ignorability" assumption, which states that treatment assignment is independent of potential outcomes conditional on observed covariates, there has been relatively limited exploration of MAR violations in causal inference. While ignorability holds by design in a randomized controlled trial, it generally does not hold in an observational study. As a result, a substantial body of literature has proposed sensitivity analyses that do not generally point identify causal effects, but instead provide bounds in cases where ignorability may not hold \cite{liu2013introduction}. However, less work has extended these methods to address violations of MAR when outcomes are missing. 

This paper addresses this gap by extending some of these methods, particularly those proposed by \cite{luedtke2015statistics}, to accommodate scenarios in which MAR may be violated but ignorability still holds. We propose a novel sensitivity framework where missingness consists of an unknown mixture of informative (non-ignorable) and non-informative (ignorable) components. Our proposed framework yields general bounds for the average treatment effect under MAR violations. To enhance practical utility, we derive narrower bounds by specifying sensitivity parameters -- such as the proportion of informative missingness -- that link observable and unobservable data quantities. All bounds may be inverted or evaluated across a range of sensitivity parameters to determine ``tipping point'' values of sensitivity parameters at which observed associations vanish, which may be preferred in situations where sensitivity parameters are difficult to specify a priori. We propose doubly-robust estimators \cite{bang2005doubly} that incorporate flexible non-parametric and machine learning techniques \cite{kennedy2022semiparametric, chernozhukov2018double} to estimate these bounds. Lastly, beyond bounds on the ATE, we provide bounds for alternative estimands more suitable in scenarios involving competing events, such as mortality, where violations of MAR are particularly pronounced. One such estimand that remains well-defined in the presence of competing events is the separable direct effect \cite{stensrud2022separable}. We implement the proposed methods in the \texttt{marbounds} package for R, available at \url{https://github.com/mrubinst757/marbounds}. 

The rest of the paper is organized as follows. In Section \ref{sec:setting} we formalize our setting and state key assumptions. In Section \ref{sec:bounds} we propose several bounds on the average treatment effect and discuss point identification. In Section \ref{sec:estimation} we formalize our proposed estimators. In Section \ref{sec:alternative-effects} we consider alternative causal estimands that can be more relevant in the presence of certain competing events. In Section \ref{sec:simulations} we provide a simulation study verifying the expected performance of our estimators. We then apply these methods to study how missingness may affect estimates of diabetes risk due to antipsychotic use in commercial claims data in Section \ref{sec:application}. Finally, we provide additional discussion in Section \ref{sec:discussion}. 

\subsection{Related work}
The challenge of drawing valid inferences from data in the presence of missing outcomes has been extensively studied across various statistical settings. Missing outcomes routinely occur in longitudinal and time-to-event settings due to loss to follow up and other missingness. For example, survival analyses often make MAR assumptions (usually referred to as assuming non-informative censoring, or CAR) to point identify various functions of survival distributions (see, e.g., \cite{nicolaie2015vertical, do2017analysis, park2022semiparametric}), and sensitivity analyses that relax the MAR assumption have been proposed \cite{lipkovich2016sensitivity}. However, a limitation of much of this literature is the absence of explicit causal interpretations of commonly estimated parameters; for instance, \cite{young2020causal} demonstrated that hazard ratios generally lack clear causal interpretations. Importantly, our work does not directly address time-to-event outcomes but instead focuses on settings with outcomes measured at discrete time points or single measurements. Nevertheless, the problem of how to handle missing data, particularly when the missingness is due to a type of outcome that precludes the primary outcome from being measured (in the survival setting, a competing event), is directly relevant to this work.

While a subset of the survival analysis literature explicitly addresses causal effects (see, e.g., \cite{robins1986new,frangakis2002principal,young2020causal}), the primary focus has been on point identification and estimation \cite{young2020causal, stensrud2022separable}. In contrast, our focus is specifically on deriving non-parametric bounds and conducting sensitivity analyses to address uncertainties due to informative missingness. Existing approaches proposing bounds and sensitivity analyses have largely relied on parametric models \cite{andrea2001methods}. Additionally, although distinctions between administrative censoring (typically assumed MAR) and competing risks (potentially informative censoring events) are recognized in the literature (see, e.g., \cite{mitra2020analysis, zhou2023case, han2021multiple, do2017analysis, moreno2013regression, lau2009competing, nicolaie2015vertical, lau2011parametric, varadhan2010evaluating}), they generally assume these missingness mechanisms can be explicitly distinguished. In our setting, the researcher does not know which missing data are informative, complicating the analysis.

Some prior work has considered similar problems at a single time point. Similar to the survival literature, missing-at-random assumptions are common in this work to point identify causal effects (see, e.g., \cite{chang2023covariate,zhao2024adjust}). By contrast, \cite{mattei2014identification} define causal effects within principal strata and derive bounds for these effects under non-ignorable missingness, though their approach relies on the availability of a binary instrumental variable separate from treatment assignment. \cite{jiang2022multiply} introduce multiply-robust estimators based on a ``principal ignorability'' assumption, enabling flexible machine learning methods for estimation; however, this assumption is strong and may not hold in many practical scenarios. \cite{scharfstein2003incorporating} propose Bayesian sensitivity analyses that summarize uncertainty in causal effects from non-ignorable missing outcomes but require explicit parametric assumptions about the missingness mechanism. 

Other researchers have jointly modeled missing outcomes and non-compliance in randomized experiments, leveraging randomization as an instrumental variable (\cite{frangakis1999addressing, baker2000, o2005likelihood, zhouandli2006}). Non-parametric bounding techniques have also been developed for average treatment effects with binary outcomes in randomized contexts (\cite{horowitz2000nonparametric, taylor2006, gabriel2023nonparametric}). We extend this non-parametric bounding approach to a broader context encompassing both randomized and observational studies. Our proposed bounds generalize previously derived results, including those in \cite{horowitz2000nonparametric} and equation (1) of \cite{gabriel2023nonparametric}. We illustrate that without additional assumptions or constraints, these general bounds can often be overly wide, limiting their practical utility in scenarios with substantial missingness. Consequently, we provide strategies for incorporating subject matter expertise to yield more informative bounds.

Finally, our approach shares some broad conceptual similarities with the sensitivity analysis framework of \cite{bonvini2022sensitivity}. This paper proposes a sensitivity analysis for no unmeasured confounding assumption by positing that some set proportion of the observed data are confounded. Specifically, they define a latent indicator $S$ such that when $S = 1$, no unmeasured confounding holds, and when $S = 0$ it does not. However, our proposed approach is not simply an extension of this framework. For example, to extend \cite{bonvini2022sensitivity} to our setting, we would instead define a latent indicator $S$ that defines a subset of units where MAR holds and a corresponding subset where it does not. However, our sensitivity analysis instead makes assumptions about the proportion of informatively missing observations among the missing observations -- not the proportion of units where MAR holds. To be clear, even for observations without informatively missing outcomes, MAR does not necessarily hold. Extending the method from \cite{bonvini2022sensitivity} to our setting ultimately requires substantively different starting assumptions, although doing so would be an interesting avenue for future research.

\section{Mixture missingness setting}\label{sec:setting}

\subsection{Data, notation, and definitions}

We consider a dataset $\Osc = \{\bO_i\}_{i=1}^n$, comprising $n$ independent observations $\bO_i \stackrel{indep.}\sim (\bX_i, A_i, C_i, (1-C_i)Y_i)$, where $\bX_i$ denotes a vector of covariates, $A_i \in \{0,1\}$ is a binary treatment indicator, and $Y_i$ denotes some binary or bounded continuous outcome of interest. $Y_i$ may not be observed for all individuals due to a missingness indicator $C_i$, which occurs before $Y_i$ and after $\bX_i$ and $A_i$. Consequently, we observe $Y_i$ when $C_i = 0$.\footnote{In contrast to this formulation, it is more conventional to use an indicator $R$ to denote missingness where $R = 0$ indicates missingness. We choose this alternative formulation for conceptual clarity: because our framework conceives of distinct missingness mechanisms $U_{I}$, and $U_{NI}$, it is easier to conceive of these variables as indicating the presence of these mechanisms rather than their absence. Moreover, defining $C = \I(U_I = 0, U_{NI} = 0)$ would then switch the meaning of the indicator (now $C = 1$ means its observed) which we believe may introduce confusion.} Finally, we posit the existence of potential outcomes $\bO_i(a)$, denoting the counterfactual set of random variables had unit $i$ been treated with $A_i = a$. 

We assume that missingness is a mixture of two mutually exclusive types\footnote{Technically, our results can be shown to hold without mutual exclusivity. However, without this requirement, the case where $U_I = 1$ and $U_{NI} = 1$ is difficult to conceptualize -- either the missingness was informative or it was not. Therefore, we invoke mutual exclusivity for conceptual clarity.}, defined by unobserved indicators $U_{NI}$ and $U_I$. The observed missingness indicator $C$ is a composite of these variables given by
 $C = 1 - (1 - U_{NI})(1 - U_I) = \I\{U_{NI} = 1 \text{ or } U_I = 1\}$.
$U_{NI} \in \{0,1\}$ represents \textit{non-informative} or administrative missingness, while $U_I \in \{0,1\}$ represents \textit{informative missingness}, which can be thought of as an intermediate outcome or a mediator of the effect of $A$ on $Y$. Missingness by $U_{NI}$ is non-informative when assumption (\ref{asmpt:non-informative})\footnote{We thank an anonymous reviewer for observing that our initial submission used a stronger non-informative missingness assumption than was necessary for our later results to hold.} holds:

\begin{assumption}[Non-informative missingness by $U_{NI}$]\label{asmpt:non-informative}
\begin{equation*}
  Y \indep U_{NI} \mid \bX, A, U_{I} = 0
\end{equation*}
\end{assumption}

\noindent For example, in a clinical trial evaluating the effect of an antipsychotic drug on symptoms of psychosis and safety events such as increases in glycosylated hemoglobin (A1c), which may herald development of type 2 diabetes (for simplicity, hereafter we refer to the outcome as diabetes risk), participants might relocate and lose contact with trial administrators. Such administrative missingness can reasonably be assumed to be independent of diabetes risk, conditional on observed covariates, and thus is non-informative. By contrast, we place no restrictions on the dependence between $Y$ and $U_I$ beyond their temporal ordering. This permits informative missingness, where $U_I$ may not be conditionally independent of $Y$. In other words,
    $Y \cancel\indep U_I | \bX, A$.    
Consider again the antipsychotic trial example. In some cases, participants may leave the trial before their outcomes are recorded due to medication side effects. Moreover, these side effects may also affect an individual's diabetes risk. In such cases, missingness is not conditionally independent of the outcome $Y$ (and is therefore \textit{informative} of the outcome). 

Our primary challenge is to establish causal relationships in the presence of unobserved informative missingness. We define our primary target of interest as the Average Treatment Effect (ATE).

\begin{definition}[Average Treatment Effect (ATE)] $ \Psi_0 = \E\{Y(1) - Y(0)\}.$
\end{definition}
\noindent The ATE represents the average difference in outcomes in a world where everyone received treatment versus one where no one received treatment. This effect averages over the covariate distribution of all individuals, including those who experience informative missingness. However, even in a randomized trial, this quantity is not generally identifiable due to the presence of informative missingness. Nevertheless, it remains the quantity most relevant to assessing the causal effect of a treatment. 

In some settings, the ATE may not correspond to a meaningful quantity. For example, consider again the case where $Y$ is type 2 diabetes, but now let $U_I$ denote all-cause mortality. For cases where the patient would have died under treatment $A = a$ ($U_I(a) = 1$) before measurement of the outcome, interpreting the potential outcome (after death) $Y(a)$ is conceptually challenging. In these situations, we define alternative treatment effects that remain well-defined (see Section \ref{sec:alternative-effects}).

We introduce the following notation: let $e(\bx) = \P(A = 1 | \bX = \bx)$ denote the probability of treatment (or propensity score), where we at times use $e_1(\bX)$ to denote $e(\bX)$ and $e_0(\bX)$ to denote $1-e(\bX)$; $\pi_a(\bx) = \P(C = 1 | \bX = \bx, A = a)$ the probability of missingness in treatment arm $A = a$; and $\pi^*_a(\bx) = \P(U_I = 1 | \bX = \bx, A = a)$ the probability of informative missingness in treatment arm $A = a$. Additionally, let $\mu_a(\bx) = \E(Y \mid \bX = \bx, A = a, C = 0)$ denote the conditional mean function in treatment arm $A = a$ among individuals with observed outcomes and $\mu_a^\star(\bx) = \E(Y \mid \bX = \bx, A = a, U_I = 1)$ the corresponding mean function among individuals informatively censored. We assume without loss of generality that $0 < \mu_a(\cdot), \mu_a^\star(\cdot) \leq 1$. This may be ensured by rescaling the bounded outcome or may occur by default with a binary outcome. Throughout, for any function $f$ of data $\bO$, $\P_n f$ denotes the empirical average $\frac{1}{n}\sum_{i=1}^n f(\bO_i)$, $\V(\bO) = \E\{(\bO - \E(\bO))(\bO - \E(\bO))\trans\}$, $\I(\cdot)$ is the indicator function, and we let $\|f(\bO)\|^2 = \int f(\bo)^2d\P(\bo)$ denote the squared $L_2(\P)$ norm. 

Finally, we define a naive target of inference as
\begin{align}\label{observed-ate}
    \Psitilde = \E\left\{\psitilde(\bX)\right\} = \E\left\{\mu_1(\bX) - \mu_0(\bX)\right\}.
\end{align}

\noindent This quantity represents the difference in expected outcomes between treatment arms among uncensored individuals, averaged over the distribution of covariates $\bX$. It is the limit of a naive estimate of the ATE in a conditionally randomized setting where all missingness is assumed to be non-informative. That is, if $\P(U_I = 1) = 0$, then under additional assumptions (specified in Section~\ref{assumps}), we have $\Psitilde = \Psi_0$. However, in the presence of informative missingness, we generally have $\Psitilde \neq \Psi_0$.

\subsection{Assumptions}\label{assumps}
We first assume that $\bX$ is sufficient to control for confounding of the primary outcome as well as the missingness indicators.
\begin{assumption}[No unmeasured confounding]\label{asmpt:confounding}
    $\{Y(a), U_{NI}(a), U_{I}(a)\} \indep A \mid \bX, \qquad a = 0, 1$.
\end{assumption}
This assumption may hold by design in a randomized controlled trial or may be assumed in an observational study. We also assume consistency,

\begin{assumption}[Consistency]\label{asmpt:consistency}
$\{Y, U_{NI}, U_{I}\} = \{Y(1), U_{NI}(1), U_{I}(1)\}A + \{Y(0), U_{NI}(0), U_{I}(0)\}(1-A)$.
\end{assumption}

\noindent Consistency implies that the observed data is equal to the counterfactual data under treatment assignment $A = a$, which will hold when there is no interference between individuals (that is, an individual's potential outcomes do not depend on another individual's treatment assignment).
Finally, we assume positivity of treatment assignment and missingness,

\begin{assumption}[Positivity]\label{asmpt:positivity}
    $\P\left\{ \epsilon < e(\bX) < 1 - \epsilon\right\} = 1$ for some $\epsilon > 0$.
\end{assumption}

\noindent Assumption \ref{asmpt:positivity} states that all individuals have some non-zero probability of being assigned to the treatment or control group. Assumptions \ref{asmpt:confounding}-\ref{asmpt:positivity} are standard in the causal inference literature and suffice to identify the ATE when all outcomes are observed. However, because not all outcomes are observed in our setting, we require a condition similar to Assumption \ref{asmpt:positivity} for the missingness mechanism.

\begin{assumption}[Missingness positivity]\label{asmpt:cens-positivity}
    $\P\left\{\pi_a(\bX) < 1 - \epsilon\right\} = 1$ for some $\epsilon > 0$.
\end{assumption}

\noindent Assumption \ref{asmpt:cens-positivity} states that, in all strata defined by $\bX$, outcomes have a positive probability of being observed. This assumption ensures sufficient overlap between the missing and non-missing observations to identify the observed data functionals that appear below. These assumptions allow us to equate the ATE to a functional that does not involve counterfactual quantities.

\begin{proposition}[Expression for $\Psi_0$]\label{prop:0}
Under Assumptions \ref{asmpt:non-informative}-\ref{asmpt:cens-positivity},
\begin{align*}
\Psi_0 = \E[\psi_0(\bX)] 
= \tilde{\Psi} + \E[\pi_1^\star(\bX)\{\mu_1^\star(\bX) - \mu_1(\bX)\} - \pi_0^\star(\bX)\{\mu_0^\star(\bX) - \mu_0(\bX)\}].
\end{align*}
\end{proposition}

\noindent The proof is given in Appendix \ref{app:proofs}. However, this quantity is not identifiable in the data: given observed data $\bO$, we do not observe $U_{NI}$, only the composite indicator $C$, making the quantities $\pi_a^\star(\bX)$ and $\mu_a^\star(\bX)$ inestimable without further assumptions. Nevertheless, as we will next show, it is possible to bound $\Psi_0$ under no additional assumptions, resulting in the general bounds we present in Section \ref{ssec:gbounds}. Moreover, in section \ref{narrower-bounds}, we show that we may narrow these bounds by invoking additional assumptions that relate the unobserved quantities $\mu_a^\star$ and $\pi_a^\star$ to the observed quantities $\mu_a$ and $\pi_a$, possibly as a function of sensitivity parameters. Moreover, all bounds we will present are sharp under the stated motivating assumptions. In section \ref{sec:sensitivity}, we propose sensitivity analyses that find which sensitivity parameters could explain away the naive $\Psitilde$, which may be useful when it is difficult to pre-specify values of the sensitivity parameters.

\section{Bounding the average treatment effect}\label{sec:bounds}

\subsection{General bounds}\label{ssec:gbounds}

Under Assumptions \ref{asmpt:non-informative}-\ref{asmpt:cens-positivity}, the ATE $\Psi_0$ is not identifiable due to the presence of informative missingness by $U_{I}$. We can, however, still bound the ATE using only quantities from the observed data under no additional assumptions. We state this result in Proposition \ref{prop:1}.

\begin{proposition}[General bounds on ATE]\label{prop:1}
Under Assumptions \ref{asmpt:non-informative}-\ref{asmpt:cens-positivity}, $\ell_0 \le \Psi_0 \le u_0$, where
\begin{align}
\begin{aligned}\label{assumption-free-bounds}
    &\ell_0 = \tilde{\Psi} - \E\left[\pi_1(\bX)\mu_1(\bX) + \pi_0(\bX)\{1-\mu_0(\bX)\}\right]\\
    &u_0 = \tilde{\Psi} + \E\left[\pi_1(\bX)\{1 - \mu_1(\bX)\} + \pi_0(\bX)\mu_0(\bX)\right].
    \end{aligned}
\end{align}    
\end{proposition}

\noindent Proposition \ref{prop:1} follows by the law of iterated expectations and the fact that $Y$ is binary (or bounded and rescaled). This result has also been noted, in, for example, \cite{mattei2014identification, gabriel2023nonparametric}. These bounds are sharp because they can be attained in the case where $\pi_a^\star(\bx) = \pi_a(\bx)$ and $\mu^\star_a(\bx) = a$ for all $\bx$ (for the upper bound) or $\mu^\star_a(\bx) = 1 - a$ for all $\bx$ (for the lower bound). Proofs for this and all subsequent propositions are available in Appendix \ref{app:proofs}. 

Consider a simple example with a single binary covariate $X$ with $\P(X = 1) = 0.7$, where the outcome risk is twice as high in the informatively missing versus non-informatively missing groups for either treatment arm ($\mu^*_a(x) = 2\mu_a(x)$), the treatment doubles the risk of the outcome (so that $\mu_1(x) = 2\mu_0(x)$), and $\mu_0(x) = 0.1 + 0.05x$. Assume there is a moderate amount of informative missingness with $\pi_0^\star(0) = 0.07, \pi_0^\star(1) = 0.12, \pi_1^\star(0) = 0.14$, and $\pi_1^\star(1) = 0.19$. For simplicity, assume that all the non-informative missingness has half of these same probabilities (implying that $\pi_a^\star(x) = \frac{2}{3}\pi_a(x)$). Treatment is randomly assigned ($e(x) = 0.5, x = 0, 1$). This implies that approximately 21\% of the outcomes are missing, with two-thirds of the missingness informative. 

In this scenario, the true ATE is $\Psi_0 \approx 0.17$. However, the naive estimand -- which assumes that all missingness is non-informative -- is $\Psitilde \approx 0.135$, approximately 20\% less. Plugging into \eqref{assumption-free-bounds}, we can bound the ATE in the range $[\ell_0, u_0] \approx [-0.07, 0.35]$. This bound is somewhat wide, and it leaves open the possibility that the true causal effect is negative, even though the naive estimate is sizeable and positive. We will next outline additional assumptions to construct narrower bounds.

\subsection{Narrower bounds}\label{narrower-bounds}

We first consider Assumptions \ref{asmpt:known-missingness}-\ref{asmpt:bounded-missingness}, which parameterize the proportion of informative missingness.

\stepcounter{assumption}

\setcounter{assumption}{6}

\begin{intassumption}[Known proportion informative missingness]\label{asmpt:known-missingness}
    $\P\left\{\frac{\pi_a^\star(\bX)}{\pi_a(\bX)} = \delta_a \right\} = 1 \text{ for some } \delta_a \in [0, 1]$.
\end{intassumption}

\begin{intassumption}[Bounded proportion informative missingness]\label{asmpt:bounded-missingness}
    $\P\left\{\delta_{a\ell} \leq \frac{\pi_a^\star(\bX)}{\pi_a(\bX)} \leq \delta_{au} \right\} = 1 \text{ for some } \delta_{a\ell}, \delta_{au} \in [0, 1]$.
\end{intassumption}

Assumption \ref{asmpt:known-missingness} states that the proportion of informative missingness in each treatment arm is some known constant that does not vary across covariate strata. One could in principle relax this assumption to allow $\delta_{a}$ to vary with $\bX$; however, in practice this function would likely be prohibitively challenging to parameterize or work with. Assumption \ref{asmpt:bounded-missingness} proposes a more feasible relaxation, stating instead that this proportion is bounded above and below across all covariate strata.\footnote{To bound $\Psi_0$ we only require assumptions on $\delta_{au}$; however, $\delta_{a\ell}$ may be useful for bounding $\Psi_1$, discussed in Section \ref{sec:alternative-effects}}. These assumptions immediately allow us to narrow the general bounds in Proposition \ref{prop:1}. 

\begin{corollary}[General bounds on ATE, bounded or known proportion informative missingness]\label{corr:0}
Under Assumptions \ref{asmpt:non-informative}-\ref{asmpt:cens-positivity}, and either \ref{asmpt:known-missingness} or \ref{asmpt:bounded-missingness}, $\ell_{0}^\star \le \Psi_0 \le u_{0}^\star$, where
\begin{align}\label{assumption-free-bounds-delta}
    &\ell_{0}^\star = \tilde{\Psi} - \E\left[\delta_{1u}\pi_1(\bX)\mu_1(\bX) + \delta_{0u}\pi_0(\bX)\{1-\mu_0(\bX)\}\right]\\
    &u_{0}^\star = \tilde{\Psi} + \E\left[\delta_{1u}\pi_1(\bX)\{1 - \mu_1(\bX)\} + \delta_{0u}\pi_0(\bX)\mu_0(\bX)\right].
\end{align}
\end{corollary}

We may also invoke assumptions that relate $\mu_a^\star$ to $\mu_a$ to further narrow these bounds. First, consider the following monotonicity constraints, which result in the bounds given in Proposition \ref{prop:2}.

\stepcounter{assumption}

\setcounter{assumption}{7}
\begin{intassumption}[Monotonicity, positive]\label{asmpt:monotonicity-pos}
    $\P\{\mu_a^\star(\bX) \geq \mu_a(\bX)\} = 1, \qquad a = 0, 1$.
\end{intassumption}

\begin{intassumption}[Monotonicity, negative]\label{asmpt:monotonicity-neg}
    $\P\{\mu_a^\star(\bX) \leq \mu_a(\bX)\} = 1, \qquad a = 0, 1$.
\end{intassumption}

\begin{proposition}[Bounds on ATE under monotonicity]\label{prop:2}
Under Assumptions \ref{asmpt:non-informative}-\ref{asmpt:cens-positivity}, \ref{asmpt:monotonicity-pos}, and either \ref{asmpt:known-missingness} or \ref{asmpt:bounded-missingness}, $\ell^{\text{pos}}_0 \leq \Psi_0 \leq u^{\text{pos}}_0$, where
\begin{align*}
     \ell^{\text{pos}}_0 &= \tilde{\Psi} - \E\left[\delta_{0u}\pi_0(\bX)\{1-\mu_0(\bX)\}\right], \\
    u^{\text{pos}}_0 &= \tilde{\Psi} + \E\left[\delta_{1u}\pi_1(\bX)\{1 - \mu_1(\bX)\}\right].
\end{align*}
Alternatively, under Assumptions \ref{asmpt:non-informative}-\ref{asmpt:cens-positivity}, \ref{asmpt:monotonicity-neg}, either \ref{asmpt:known-missingness} or \ref{asmpt:bounded-missingness}, we may obtain the bounds $\ell^{\text{neg}}_0 \leq \Psi_0 \leq u^{\text{neg}}_0$, where
\begin{align*}
\ell^{\text{neg}}_0 &= \tilde{\Psi} - \E\left[\delta_{1u}\pi_1(\bX)\mu_1(\bX)\right], \\
u^{\text{neg}}_0 &= \tilde{\Psi} + \E\left[\delta_{0u}\pi_0(\bX)\mu_0(\bX)\right].
\end{align*}
\end{proposition}

\begin{remark}
Assumption \ref{asmpt:bounded-missingness} is again guaranteed to hold for $\delta_{au} = 1$. Assumption \ref{asmpt:bounded-missingness} may thus be discarded if $\delta_{au}$ is replaced by one in the expressions in Proposition \ref{prop:2}.
\end{remark}

Consider again the example in Section \ref{ssec:gbounds}, where we obtained the general bounds $[-0.07, 0.35]$. Under positive monotonicity alone, which holds in our example, we can obtain the slightly more narrow bounds $[-0.00, 0.33]$ (via Proposition \ref{prop:2}). If we are further willing to assume that the proportion of informative missingness in either treatment arm is at most 80\% ($\delta_{1u} = \delta_{0u} = 0.8$), we can further refine the bounds to be approximately $[0.03, 0.29]$. 

Typically, the informative missingness that most threatens identification is monotonic -- e.g., sicker patients are less likely to have complete data -- but we do not assume that monotonicity assumptions will be realistic for all applications. For example, patients may be less likely to follow up with healthcare providers if they are feeling healthier (and thus not in immediate need of care) as well as if they are too sick to seek care. As an alternative, we may consider assumption \ref{asmpt:bounded-risk} below.\footnote{Implicit in assumption \ref{asmpt:bounded-risk} and \ref{asmpt:known-risk} is that $\mu_a(\bx) > 0$ for all $a, \bx$. Alternatively, one could re-state this assumption on the support of $\bx$ where $\mu_a(\bx) > 0$, and define $\mu_a^\star(\bx) = 0$ where $\mu_a(\bx) = 0$. It would then be straightforward as below to derive analogous expressions and estimators for the implied bounds, though we do not explore this further in this paper.}

\begin{intassumption}[Bounded outcome risk]\label{asmpt:bounded-risk}
    $\P\left\{\tau_a\inv \leq \frac{\mu_a^\star(\bX)}{\mu_a(\bX)} \leq \tau_a \right\} = 1, a = 0, 1 \text{ for some } \tau_a \ge 1$.
\end{intassumption}

This assumption restricts how much informative missingness may increase/decrease the risk of the primary outcome, resulting in the following bound:

\begin{proposition}[Bounds on ATE, bounded outcome risk]\label{prop:3}
Under Assumptions \ref{asmpt:non-informative}-\ref{asmpt:cens-positivity}, either \ref{asmpt:bounded-missingness} or \ref{asmpt:known-missingness}, and \ref{asmpt:bounded-risk}, $\widetilde{\ell}_0 \le \Psi_0 \le \widetilde{u}_0$, where

\begin{align*}
\widetilde{\ell}_0 &= \tilde{\Psi} + \E\left[\delta_{1u}(\tau_1\inv - 1)\pi_1(\bX)\mu_1(\bX) - \delta_{0u}(\tau_0 - 1)\pi_0(\bX)\mu_0(\bX) \right], \\
\widetilde{u}_0 &= \tilde{\Psi} + \E\left[\delta_{1u}(\tau_1 - 1)\pi_1(\bX)\mu_1(\bX) - \delta_{0u}(\tau_0\inv - 1)\pi_0(\bX)\mu_0(\bX) \right].
\end{align*}
\end{proposition}

\begin{remark}\label{rmk:prop3}
This bound is sharp when $\tau_a \mu_a(\bx) \le 1$ for $a = 0, 1$ and all $\bx$. However, this may not hold for any given choice of $\tau_a$. This will be a greater concern when the risk of the outcome is high (near one), leading to an overly conservative bound. One may instead start from the assumption that $\frac{\mu_a^\star(\bx)}{\mu_a(\bx)}\le\min(1/\mu_a(\bx), \tau_a)$ and integrate over the resulting expression to obtain a sharp bound. We give the explicit formulation of the resulting bounds, along with doubly-robust style estimators of these quantities, in Appendix \ref{app:addtlres}. \end{remark}

Continuing with our example, assuming that the risk of the outcome in observations with missing outcomes is no more than three times and no less than one-third of the risk of the outcome in the non-censored observations ($\tau = 3$), then even if we allow all missingness to be informative ($\delta_{au} = 1$), we obtain the bounds $[0.04, 0.29]$. If we further restrict the proportion of informative missingness to be no more than 80\%, we obtain the slightly narrower bounds $[0.06, 0.26]$. 

\subsection{Choosing sensitivity parameters}

Some of our proposed sensitivity analyses require specifying values for unknown parameters, such as the proportion of informative missingness $\delta_{au}$ and the outcome risk ratio $\tau_a$. Depending on the context of the application, we recommend four broad strategies to choose these values: (1) imposing the weakest assumptions possible; (2) using subject-area expertise; (3) using data-informed approaches; and (4) conducting an analysis across a range or grid of parameters (which, to be clear, is not mutually exclusive with the first three strategies). We discuss these strategies in more detail below.

First, we consider minimal assumption analyses. Researchers who feel unable to reason about the values of these parameters may prefer to make as few assumptions as possible. In this case, one may set $\delta_{au}=1$ (and, where relevant, $\delta_{a\ell}=0$). This corresponds to allowing the maximum possible proportion of informative missingness and therefore imposes no restriction on the missingness mechanism. Bounds obtained under these values represent the most conservative conclusions that can be drawn under the assumed structural conditions. Similarly, if we wish to make no assumptions on $\tau_a$, researchers may use the general bounds provided in Proposition \ref{prop:1}. Naturally, these bounds will typically be the widest.

Second, domain knowledge can meaningfully constrain the sensitivity parameters for some applications. For example, investigators may know that some missing outcomes are due to administrative censoring (e.g., the study ending before outcomes are observed), which would suggest that only a fraction of missingness is plausibly informative. Such reasoning can justify values of $\delta_{au}<1$. Similarly, subject-matter experts may have intuition about plausible outcome risk differences under informative missingness. For example, investigators might believe that the outcome risk among informatively missing individuals could plausibly be no more than twice as large (or half as large) as among observed individuals, corresponding to $\tau_a=2$.

Third, we may use data to help calibrate the sensitivity parameters. For example, in some cases an auxiliary dataset may be available that contains information on the informative missingness indicator $U_I$, allowing one to estimate quantities such as $\frac{\P(U_I=1)}{\P(C=1)}$, or even $\pi_a^\star(\bX)$ directly, depending on what variables are observed. These quantities can provide empirical guidance about plausible ranges for $\delta_{au}$. To inform $\tau_a$, even with no auxiliary dataset, one may compare the predictive strength of observed covariates with respect to missingness. Suppose there exist covariates $(\bX,U_X)$ such that
\[
Y \perp U_I \mid \bX,U_X.
\]
Although this assumption rules out direct causal effects of $U_I$ on $Y$, it allows the observed covariates to serve as a proxy for the confounding strength of informative missingness. One can then compute
\[
\Gamma_a =
\frac{\E[Y\mid A=a,C=1]}
     {\E[Y\mid A=a,C=0]} =
\frac{\E[\E[Y\mid A=a,C=0,\bX]\mid A=a,C=1]}
     {\E[\E[Y\mid A=a,C=0,\bX]\mid A = a, C = 0]},
\]
which summarizes the ability of the entire observed covariate set to explain outcome differences between censored and uncensored individuals within treatment arm $A=a$. This quantity can serve as a heuristic benchmark when specifying values of $\tau_a$. While these proposals outline some basic ideas, developing principled methods for calibrating these sensitivity parameters to observed data would be a useful area for future work.

Finally, regardless of whether researchers wish to make minimal assumptions, use subject-area expertise, or a data-driven approach, they may also wish to evaluate the bounds over a range or a grid of sensitivity parameter values. For example, subject-area expertise may only inform the range of plausible $\delta_{au}$; similarly, a calibration analysis may suggest a value for $\P(U_I = 1)/\P(C = 1)$, which can restrict the plausible range of $\pi_a^\star(\bx)$, but does not inform $\delta_{au}$ directly. In these cases one may evaluate the bounds across a range or grid, and plot how the bounds vary over $(\delta_{au},\tau_a)$ to provide a transparent view of how conclusions depend on the assumed parameters. This approach is particularly useful when analysts cannot commit to a single parameter choice but wish to understand how robust their conclusions are across a range of plausible assumptions.

\subsection{Point identification and tipping point sensitivity analyses}\label{sec:sensitivity}

The strategies outlined above all require specifying values or ranges for the sensitivity parameters. An alternative approach is to instead ask what magnitude of informative missingness would be required to explain away the observed association. This type of analysis, often called a ``tipping-point'' analysis, starts from the naive estimate $\tilde{\Psi}$ and determines the values of the sensitivity parameters for which the causal effect could change sign. Each of the previously proposed bounds may be ``inverted'' to obtain such a sensitivity analysis. That is, for fixed values of all but one sensitivity parameter, one may set the expression equal to zero and solve for the value of the parameter that would satisfy the equation. However, we emphasize that the grid searches proposed above also often reveal tipping points without directly solving for these values, and should not be viewed as incompatible or mutually exclusive with these tipping point analyses. 

Nevertheless, consider again the running example above to illustrate. In place of the bounds based on Proposition \ref{prop:2}, we may rearrange terms in Proposition \ref{prop:2} to determine what value of $\delta_{0u}$ or $\delta_{1u}$ are required for $\ell_0^{\text{pos}} > 0$ or $u_0^{\text{pos}} < 0$. In this example, we can guarantee that $\ell_0^\text{pos} \geq 0$ if $\delta_{0u} \leq \Psitilde/\E\left[\pi_0(\bX)(1-\mu_0(\bX))\right] \approx 1.00$ (which is always true because $\delta_{0u} \leq 1$ by definition). This type of analysis may be appealing in settings where researchers struggle to specify particular values of the sensitivity parameters but yet have a sense of which sensitivity parameters would be implausible. 

However, in cases where we also wish to make more substantive restrictions on the outcome risk ratio, the sensitivity analysis proposed in Proposition \ref{prop:3} may be harder to invert in an interpretable way, in part due to the numbers of parameters. In order to ease communication of such analyses to collaborators and non-statisticians, one may wish to reduce the complexity of this analysis. We therefore also propose a simplified analysis based primarily on a single sensitivity parameter $\tau$ that indicates how much informative missingness changes the risk. Specifically, we first show how $\Psi_0$ may be identified by assuming that the proportion of informative missingness ratio and outcome risk ratio are constant and known quantities. Assumption \ref{asmpt:known-risk} and Proposition \ref{prop:4} formalize this result.

\begin{intassumption}[Known outcome risk ratio]\label{asmpt:known-risk}
    $\P\left\{\frac{\mu_a^\star(\bX)}{\mu_a(\bX)} = \tau \right\} = 1, a = 0, 1 \text{ for some } \tau > 0$.
\end{intassumption}

\begin{proposition}[Point identification of the ATE]\label{prop:4}
Under Assumptions \ref{asmpt:non-informative}-\ref{asmpt:cens-positivity}, \ref{asmpt:known-missingness}, and \ref{asmpt:known-risk},
\begin{align*}
\Psi_0 = \tilde{\Psi} + \E\left[(\tau - 1)\left\{\delta_1\pi_1(\bX)\mu_1(\bX) - \delta_0\pi_0(\bX)\mu_0(\bX)\right\}\right]. 
\end{align*}    
\end{proposition}

\begin{remark}
This expression is also valid for any continuous outcome under the stated assumptions.
\end{remark}

We can then use Proposition \ref{prop:4} to examine what combination of $(\tau, \delta_0, \delta_1)$ would be required to explain away the observed naive association $\Psitilde$, making succinct statements based only on $\tau$ when we let $\delta_0, \delta_1$ take extreme values. For example, assume, as in our running example, that $\Psitilde > 0$. Then, setting $\delta_0=1$ and $\delta_1 = 0$ (the most adversarial configuration of informative missingness), we may conclude that $\tau$ must be at least as large as $1 + \Psitilde/\E[\mu_0(\bX)\pi_0(\bX)]$ to possibly explain away the observed association. Correspondingly, if $\tau < 1 + \Psitilde/\E[\mu_0(\bX)\pi_0(\bX)]$, then $\Psi_0 > 0$ under the assumptions of Proposition \ref{prop:4}. 

In our running example, this analysis allows us to conclude that unless informative missingness increases risk by at least a factor of $\tau = 7.12$, there is no way to explain the observed naive association of $\Psitilde = 0.135$. And, even if $\tau$ could reasonably be as large as that, the proportion of informative missingness in the control group $\delta_0$ has to be at least as large as $\frac{\delta_1\E[\mu_1(\bX)\pi_1(\bX)] + \Psitilde/(\tau - 1)}{\E[\mu_0(\bX)\pi_0(\bX)]}$ when $\tau \ne 1$. If we take $\tau = 10$, then we require $\delta_0 > 0.68 + 3.29\delta_1$, which restricts the amount of informative missingness in the treated group to be very small and in the control group to be nearly universal. We plot the regions implied by Proposition \ref{prop:4} in Figure \ref{fig:regions}. For a given value of $\tau$, corresponding to one line in Figure \ref{fig:regions}, $(\delta_0, \delta_1)$ must lie in the region to the lower right of the line if informative missingness could possibly explain the observed association.

\begin{figure}
 \begin{center}
     \includegraphics[scale=0.8]{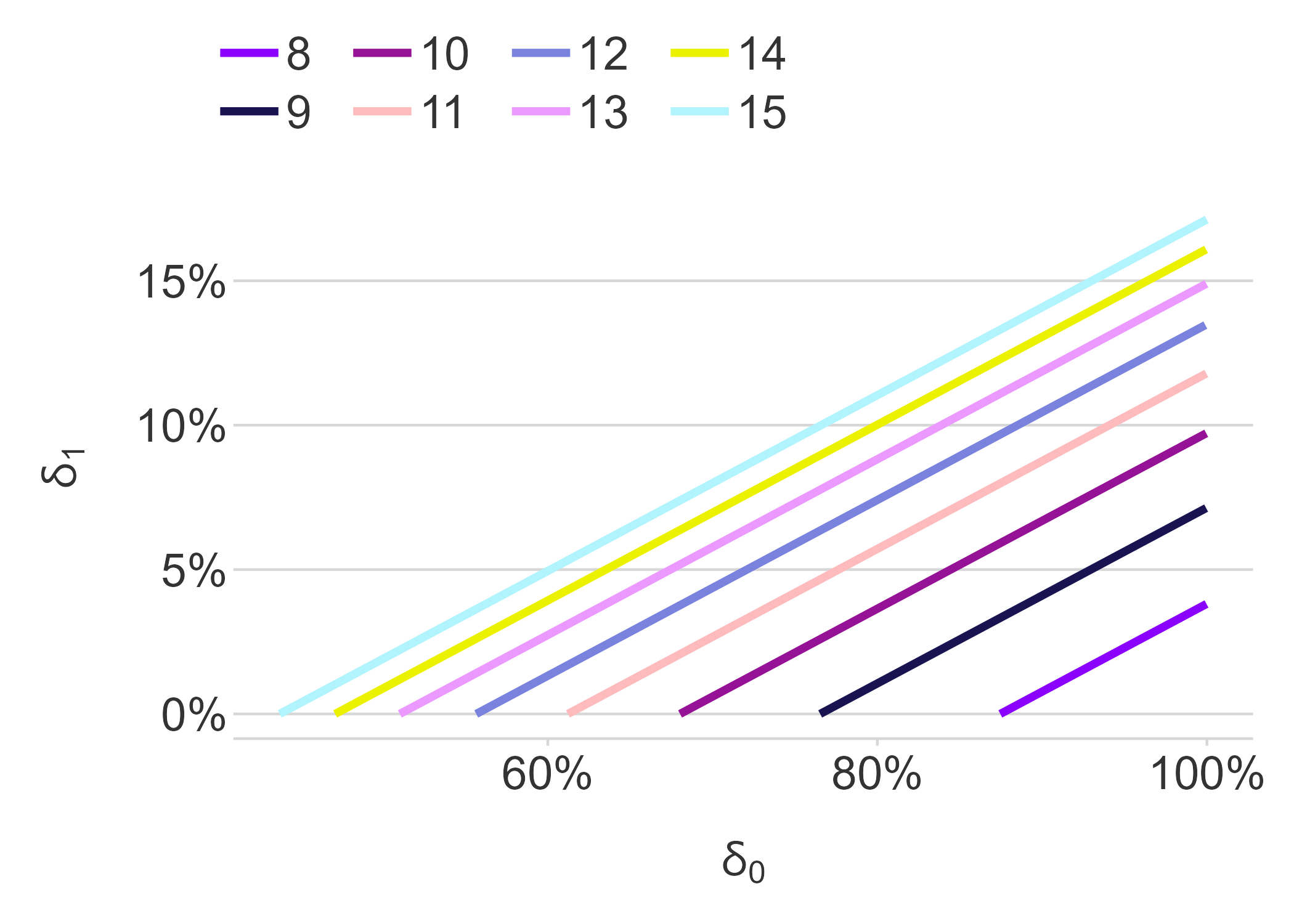}
 \end{center}
 \caption{Sensitivity analysis parameters corresponding to $\Psi_0 = 0$ under the assumptions of Proposition \ref{prop:4}. Each line corresponds to a different value of $\tau$, and each point on a line indicates a set of sensitivity parameters ($\tau, \delta_0, \delta_1$) corresponding to $\Psi_0 = 0$. Any values of $(\delta_0, \delta_1)$ to the lower right of the $\tau-$line are sensitivity parameters that can explain the observed association $\Psitilde$.}
 \label{fig:regions}
\end{figure}

We again construct this sensitivity analysis to be as simple as possible. To this end, this sensitivity analysis relies on Assumption \ref{asmpt:known-risk} which states that the increased risk due to informative missingness is the same in the treated and control groups. This is not required, and one may instead assume two different risk ratios $\tau_0, \tau_1$. This additional flexibility comes at the expense of a sensitivity analysis that is harder to interpret, and we therefore prefer the former approach. Furthermore, we again stress that all of the bounds in preceding sections may be ``inverted'' to make similar statements about what type of informative missingness might explain away the observed association under their stated assumptions. However, the analysis proposed in this section may be relatively easier to communicate than others.

\section{Estimation and inference}\label{sec:estimation}

In practice, since the nuisance functions $(e, \pi_0, \pi_1, \mu_0, \mu_1)$ are unknown, we cannot compute the bounds exactly. Instead, we must estimate the bounds as functions of these quantities using observed data. To do this, we propose using an influence-function-based estimation approach. Estimators based on influence functions typically have a so-called ``doubly-robust'' property, enabling them to achieve parametric rates of convergence while still allowing for the use of non-parametric and machine learning estimation methods for the nuisance functions (e.g. the propensity score $e$, the missingness probabilities $\pi_a$, and the outcomes models $\mu_a$ -- collectively denoted by $\eta = \{e(\cdot), \pi_0(\cdot), \pi_1(\cdot), \mu_0(\cdot), \mu_1(\cdot)\}$). This property occurs because the errors of influence-function-based estimators are products of the errors of nuisance function estimates (call them $\etahat = \{\ehat(\cdot), \pihat_0(\cdot), \pihat_1(\cdot), \muhat_0(\cdot), \muhat_1(\cdot)\}$). Consequently, the overall estimator can achieve a faster rate of convergence even when the individual components of $\etahat$ converge at slower rates (see \cite{kennedy2022semiparametric} for more details). As a result, we require minimal assumptions on the data distribution to estimate the bounds. 

All proposed bounds (and the point identification) can be expressed as linear combinations of elements of the vector $\btheta = \mathbb{E}\{\zeta(\bX)\},$ where $\zeta(\bx) = \{\mu_0(\bx), \mu_1(\bx), \pi_0(\bx), \pi_1(\bx),  \mu_0(\bx)\pi_0(\bx), \mu_1(\bx)\pi_1(\bx)\}$.  Because the influence function of a linear combination of functionals is a linear combination of the same influence functions \cite{kennedy2022semiparametric}, it suffices to state the influence function of the components of $\btheta$. All bounding quantities can be estimated based on the following influence functions: 

\begin{proposition}\label{prop:5}
The uncentered influence functions for the functionals $\E\{\mu_a(\bX)\}$, $\E\{\pi_a(\bX)\}$, and $\E\{\mu_a(\bX)\pi_{a}(\bX)\}$ are given by $\varphi_{1, a}(\bO; \eta), \varphi_{2, a}(\bO; \eta)$, $\varphi_{3, a}(\bO; \eta)$, respectively, where
\begin{align*}
\varphi_{1, a}(\bO; \eta) &= \frac{\I\{C = 0, A = a\}}{(1 - \pi_a(\bX))e_a(\bX)}\{Y - \mu_a(\bX)\} + \mu_a(\bX), \\
\varphi_{2, a}(\bO; \eta) &= \frac{\I\{A = a\}}{e_a(\bX)}\{C - \pi_a(\bX)\} + \pi_a(\bX), \\
\varphi_{3, a}(\bO; \eta) &= \varphi_{1, a}(\bO; \eta)\pi_{a}(\bX) + \varphi_{2,a}(\bO; \eta)\mu_a(\bX) - \mu_a(\bX)\pi_{a}(\bX).
\end{align*}
\end{proposition}

The derivations of these influence functions follow naturally from previously established results (see, e.g., \cite{hahn1998role}, \cite{luedtke2015statistics}); for completeness, we include all proofs in Appendix \ref{app:proofs}. To construct estimators, we plug-in unknown nuisance functions with their estimates and compute the empirical average of these influence functions. For example, we may estimate $\E\{\mu_a(\bX)\}$ by $\P_n\{\varphi_{1, a}(\bO; \etahat)\} = n\inv \sum_{i=1}^n\varphi_{1, a}(\bO_i; \etahat)$.

Suppose the target of inference is $\theta_\bc = \bc\trans\btheta$, a linear combination of $\btheta$ (and where the elements of $\bc$ are finite). Collect the influence functions into the vector $\mathbf{\varphi}(\bO, \eta) = \left\{\varphi_{1,0}(\bO, \eta), \varphi_{1,1}(\bO, \eta),
\varphi_{2,0}(\bO, \eta), \varphi_{2,1}(\bO, \eta),
\varphi_{3,0}(\bO, \eta), \varphi_{3,1}(\bO, \eta)\right\}.$  All the quantities discussed above, including the bounds, can be written in this form. Corollary \ref{corr:1} in Appendix \ref{app:addtlres} provides explicit expressions for all previously considered quantities. We now describe the asymptotic behavior of the estimators of these quantities.

\begin{theorem}\label{thm:1}
    Let $\thetahat_\bc = \mathbb{P}_n\{\bc\trans\varphi(\bO; \hat{\eta})\}$. Suppose that $\hat{\eta}$ is estimated on an independent sample of size $n$, and assume that $\|\bc\trans\varphi(\bO; \hat{\eta}) - \bc\trans\varphi(\bO; \eta)\| = o_p(1)$. Then, we have
    \begin{align*}
    \thetahat_\bc - \theta_\bc &= \mathbb{P}_n[\bc\trans\varphi(\bO; \eta) - \theta_\bc] + o_p(n^{-1/2}) + \mathcal{O}_p(R_2(\thetahat_\bc)).
    \end{align*}
    where
    \begin{align*}
        |R_2(\thetahat_\bc)| &\lesssim \sum_{a=0,1}\left(\|\mu_a(\bX) - \hat{\mu}_a(\bX)\|\left\{\|\pi_a(\bX) - \hat{\pi}_a(\bX)\| + \|e(\bX) - \hat{e}(\bX)\|\right\} \right. \\ &\left.+ \|e(\bX) - \hat{e}(\bX)\|\|\pi_a(\bX) - \hat{\pi}_a(\bX)\|\right). 
    \end{align*}
    Therefore, if $R_2(\thetahat_\bc) = o_p(n^{-1/2})$, it follows by the Central Limit Theorem that
     \begin{align*}
        \sqrt{n}\bigl(\thetahat_\bc - \theta_\bc\bigr) \stackrel{d}\to \mathcal{N}\bigl(0, \bc\trans \V\{\varphi(\bO; \eta)\}\bc\bigr).
     \end{align*}
\end{theorem}
Theorem \ref{thm:1} states that any linear combination of the estimated influence functions of $\btheta$ has a limiting normal distribution with asymptotic variance given by the variance of the corresponding linear combination of the true influence functions. This result holds whenever the nuisance estimators converge sufficiently fast so that the second-order remainder term $R_2(\thetahat_\bc)$ is $o_p(n^{-1/2})$. A sufficient condition is that the nuisance estimators converge faster than $n^{-1/4}$ in $L_2(\P)$ norm. Many off-the-shelf machine learning algorithms can achieve these rates under minimal assumptions on the data-generating mechanism \cite{kennedy2022semiparametric}. An asymptotically valid $(1-\alpha)\times 100\%$ confidence
interval may then be constructed as
\begin{align}\label{eqn:ci}
\thetahat_\bc \pm z_{1-\alpha/2}\sqrt{\bc^\top \widehat{\V}\{\varphi(\bO; \hat\eta)\}\bc/n},
\end{align}
where $z_q$ is the $q$-th quantile of the standard normal distribution.

To illustrate, the point identification from Proposition \ref{prop:4} corresponds to choosing $\bc_1 = (-1, 1, 0, 0, (1-\tau)\delta_0, (\tau-1)\delta_1)^\top$ so that $\theta_\bc = \Psi_0$ under Assumptions \ref{asmpt:non-informative}-\ref{asmpt:cens-positivity}, \ref{asmpt:known-missingness}, and \ref{asmpt:known-risk}. The influence-function based estimator for this quantity is then, 
\begin{align*}
    \Psihat_0 = \thetahat_{\bc_1} = \mathbb{P}_n\{\bc_1\trans\varphi(\bO; \hat{\eta})\},
\end{align*} 
with the corresponding confidence interval given by \eqref{eqn:ci}, substituting $\bc_1$ for $\bc$.

\begin{remark}
    The bound on the second-order remainder term $|R_2(\thetahat_\bc)|$ given in Theorem \ref{thm:1} may involve more error products than are strictly required for specific bounds. The precise second-order remainder terms for individual estimators of each bound are given in Corollary \ref{corr:2} in Appendix \ref{app:addtlres}.
\end{remark}

\begin{remark}
    Assuming $\etahat$ comes from an independent sample facilitates the statement of Theorem \ref{thm:1} but is not required, provided that cross-fitting is used to separate the estimation of nuisance functions from their evaluation in the influence function (see \cite{chernozhukov2018double}). 
\end{remark}

\begin{remark}\label{rmk:uniform}
    In cases where analysts wish to evaluate the bounds across a grid or range of sensitivity parameters, one may construct uniform confidence bands (that is, confidence intervals with simultaneous coverage guarantees for all parameters in some parameter space) at level $1-\alpha$ using the multiplier bootstrap under minimal additional conditions. We formalize this in Theorem \ref{thm:4} in Appendix \ref{app:multiplier}. 
\end{remark}

\begin{remark}
    As noted previously, all bounds and identifying expressions may be inverted to obtain tipping points, that is, values of the sensitivity parameters that set the bound or identifying expression equal to zero. In principle, such tipping points may be treated as functionals of the observed data distribution and estimated directly using influence-function based methods. In many cases, these quantities are ratios of linear functionals of $\btheta$ so that their efficient influence functions follow from standard results \cite{kennedy2022semiparametric}. We do not pursue such direct estimators here. Instead, because many analysts will evaluate the bounds over a grid of sensitivity parameters, we recommend inferring tipping points from the sensitivity curve itself, namely as the parameter value at which the lower or upper bound crosses zero. Combined with the uniform confidence bands described in Remark \ref{rmk:uniform}, this yields a simultaneous uncertainty quantification for the entire sensitivity curve and therefore an uncertainty region for the tipping point.
\end{remark}

\section{Alternative target estimands}\label{sec:alternative-effects}

The ATE may not always be a relevant causal estimand, particularly in scenarios where the informative missingness event $U_{I}$ can be thought to \textit{prevent} the outcome $Y$ from occurring -- for example, when $U_I$ reflects death. In this section, we therefore discuss two alternative estimands.

We first define a composite outcome $Z = \I\{Y = 1 \text{ or } U_I = 1\}$, which indicates whether either $Y$ or $U_I$ has occurred. This composite outcome avoids complications related to defining $Y(a)$ after $U_I(a)$. Notice that $Z(a) = 1$ whenever $U_I(a) = 1$, regardless of the value or existence of $Y(a)$. This treatment effect (the ``composite ATE'', hereafter) is
    $\Psi_1 = \E\{Z(1) - Z(0)\}$.
The composite ATE is appropriate when both $U_I$ and $Y$ are undesirable outcomes (e.g., $U_I$ is death and $Y$ is diabetes onset) or desirable outcomes (e.g., $U_I$ is study drop-out due to treatment success and $Y$ is a marker of positive health status) that the treatment may uniformly negatively or positively affect.

In some cases, however, this composite outcome may not be of interest. Building on ideas from the mediation literature \cite{robins2022interventionist}, \cite{stensrud2022separable} proposed a \textit{separable direct effect} (SDE) that essentially isolates the effect of the treatment on $Y$. Specifically, the SDE conceives of the treatment indicator $A$ as being comprised of two separate treatments, $A_y$ and $A_d$, where $A_y$ only affects the outcome and $A_d$ only affects $U_I$. That is, we can conceive of a hypothetical trial where $A_y$ and $A_d$ are jointly randomized and recover causal contrasts involving the potential outcomes $Y(a_y, a_d)$ for different realizations of $a_y, a_d$.  Here we are interested in the part of $A$ that affects $Y$ (i.e., $A_y$), while holding the part of $A$ that affects $U_I$ constant at $A_d = 0$. We therefore first assume that $Y(a) = Y(a_y, a_d)$, and that $U_I(a) = U_I(a_y, a_d)$, and define the SDE as
    $\Psi_2 = \E\{Y(1, 0) - Y(0, 0)\}$.
Again, this estimand represents the causal effect via the pathways from $A$ that affect $Y$ alone ($A \boldsymbol{\rightarrow} A_y \rightarrow Y$, where the bold arrow indicates a deterministic relationship), leaving the pathways that affect $U_I$ unchanged ($A\boldsymbol{\rightarrow} A_d \to U_I \to Y$). Neither $\Psi_1$ or $\Psi_2$ are identified under the basic assumptions (Assumptions \ref{asmpt:non-informative}-\ref{asmpt:cens-positivity}). However, these assumptions are sufficient to rewrite $\Psi_1$ similarly to how $\Psi_0$ was rewritten in Proposition \ref{prop:0}.

\begin{proposition}[Expression for $\Psi_1$]\label{prop:6}
Under Assumptions \ref{asmpt:non-informative}-\ref{asmpt:cens-positivity},

\begin{align*}
\Psi_1 = \tilde{\Psi} + \E[\pi_1^\star(\bX)(1 - \mu_1(\bX)) - \pi_0^\star(\bX)(1 - \mu_0(\bX))].
\end{align*}
\end{proposition}

\noindent To bound $\Psi_2$, we further require the so-called ``dismissible component conditions'' \cite{stensrud2022separable}. 

\begin{assumption}[Dismissible component conditions]\label{asmpt:dcc}
For all $\bx$,
\begin{align*}
\P\left\{Y(1, 1) = 1 \mid U_I(1, 1) = 0, \bx\right\} &= 
\P\left\{Y(1, 0) = 1 \mid U_I(1, 0) = 0, \bx\right\}, \\
\P\left\{U_I(1, 0) = 1 \mid \bx\right\} &= \P\left\{U_I(0, 0) = 1 \mid \bx\right\}.
\end{align*}
\end{assumption}

\noindent This assumption is required because we do not observe the hypothetical trial that jointly randomizes $(A_y, A_d)$. Instead, our observed data contain only the composite treatment indicator $A$, so that $A_y = A_d$ for all units. This assumption also intuitively formalizes the idea that $A_d$ does not affect $Y(a_y, a_d)$ except via its effect on $U_I$, and similarly, that $A_y$ does not affect $U_I$. We refer to \cite{stensrud2022separable} for more details.

\begin{proposition}[Expression for $\Psi_2$]\label{prop:7}
Under Assumptions \ref{asmpt:non-informative}-\ref{asmpt:cens-positivity} and Assumption \ref{asmpt:dcc},
$\Psi_2 = \E[Y(1, 0) - Y(0, 0)] = \Psitilde - \E[\pi_0^\star(\bX)(\mu_1(\bX) - \mu_0(\bX))]$.
\end{proposition}

\noindent Unlike $\Psi_0$, these estimands do not depend on the unobserved quantity $\mu_a^\star(\bX)$, and therefore Assumptions \ref{asmpt:monotonicity-neg}-\ref{asmpt:bounded-risk} are not required to bound these quantities. Instead, $\Psi_1$ only depends on the unobserved quantities $\pi^\star_a(\bX)$ for $a = 0, 1$, and $\Psi_2$ only on $\pi^\star_0(\bX)$. In the next section, we provide bounds for these alternative estimands analogous to the bounds in Propositions \ref{prop:3} and \ref{prop:4}.

\subsection{Bounds}

Analogous to results for the ATE above, we first discuss bounds for $\Psi_1$ and $\Psi_2$ under Assumption \ref{asmpt:bounded-missingness}, which limits the amount of informative missingness. Bounds under minimal assumptions are obtained as a special case when we set $\delta_{a\ell} = 0, \delta_{au} = 1$. We also provide an expression for the point identified parameter under Assumption \ref{asmpt:known-missingness}.

\begin{proposition}\label{prop:8}
Under Assumptions \ref{asmpt:non-informative}-\ref{asmpt:cens-positivity} and \ref{asmpt:bounded-missingness}, $\ell_1 \leq \Psi_1 \leq u_1;\quad \ell_2 \leq \Psi_2 \leq u_2$ for
\begin{align*}
&\ell_1 = \tilde{\Psi} + \E\left[\delta_{\ell1}\pi_1(\bX)(1 - \mu_1(\bX)) - \delta_{u0}\pi_0(\bX)(1-\mu_0(\bX))\right], \\
&u_1 = \tilde{\Psi} + \E\left[\delta_{u1}\pi_1(\bX)(1 - \mu_1(\bX)) - \delta_{\ell0}\pi_0(\bX)(1 - \mu_0(\bX))\right], \\
&\ell_2 = \tilde{\Psi} - \E\left[\delta_{u0}\pi_0(\bX)(\mu_1(\bX) - \mu_0(\bX))\I\{\mu_1(\bX) > \mu_0(\bX)\}\right], \\
&u_2 = \tilde{\Psi} - \E\left[\delta_{u0}\pi_0(\bX)(\mu_1(\bX) - \mu_0(\bX))\I\{\mu_1(\bX) \le \mu_0(\bX)\}\right].
\end{align*}
\end{proposition}

To illustrate, consider again the running example introduced in Section \ref{sec:bounds}. In this case, $\Psi_1 \approx 0.17$ and $\Psi_2 \approx 0.12$. Setting $\delta_{1u} = \delta_{0u} = 1$ and $\delta_{1\ell} = \delta_{0\ell} = 0$ (which hold by definition), we derive the bounds $[-0.00, 0.33]$ for $\Psi_1$ and $[0.11, 0.135]$ for $\Psi_2$. These can be further refined using Assumption $\ref{asmpt:bounded-missingness}$.

Finally, as with the ATE, we can point identify both the composite ATE and the SDE under Assumption \ref{asmpt:known-missingness}.

\begin{proposition}\label{prop:9}
Under Assumptions \ref{asmpt:non-informative}-\ref{asmpt:cens-positivity} and \ref{asmpt:known-missingness},
\begin{align*}
   \Psi_1 &= \tilde{\Psi} + \E\left[\delta_1\pi_1(\bX)\{1-\mu_1(\bX)\} - \delta_0\pi_0(\bX)\{1-\mu_0(\bX)\}\right], \\ 
   \Psi_2 &= \tilde{\Psi} - \E\left[\delta_0\pi_0(\bX)\{\mu_1(\bX) - \mu_0(\bX)\} \right].
\end{align*}
\end{proposition}

\noindent As noted previously, Assumption \ref{asmpt:known-missingness} is quite strong, so it is most useful as part of sensitivity analyses, as discussed in Section \ref{sec:sensitivity}.  Moreover, for $\Psi_1$, once we are simply obtaining point estimates under a range of $\delta_1, \delta_0$, this is essentially equivalent to obtaining bounds $\ell_1,u_1$ under different choices of $\delta_{au},\delta_{a\ell}$.

\subsection{Estimation}

Estimation of quantities for these estimands using influence-function based estimators follows the same general strategy outlined in Section \ref{sec:estimation}. Except for the bounds on $\Psi_2$, all bounds are functions of previously described functionals $\btheta$, so the influence-function based estimator of the bounds can be obtained using results from Proposition \ref{prop:5} and Theorem \ref{thm:1}. These estimators are fully described in the Appendix, and Corollary \ref{corr:2} provides an expression for the second-order remainder terms that must be $o_p(n^{-1/2})$ for the estimators to achieve root-$n$ consistency and asymptotic normality.

The remaining challenge is estimating $\ell_2$ and $u_2$. These functionals are not pathwise-differentiable due to the presence of the indicator functions ($\I\{\mu_1(\bX) \ge \mu_0(\bX)\}$ and $\I\{\mu_1(\bX) < \mu_0(\bX)\}$) \cite{branson2023causal} and therefore have no influence functions. One option is to replace  the indicator functions with smooth approximations, such as the centered normal cumulative distribution function $\Phi_{\varepsilon}$, where $\varepsilon$ parameterizes the standard deviation. In other words, instead of targeting
$\ell_2 = \tilde{\Psi} - \E[\pi_0(\bX)(\mu_1(\bX) - \mu_0(\bX))\I\{\mu_1(\bX) > \mu_0(\bX)\}]$, we may instead target
$\ell_{2, \varepsilon} = \tilde{\Psi} - \E[\pi_0(\bX)(\mu_1(\bX) - \mu_0(\bX))\Phi_{\varepsilon}(\mu_1(\bX) - \mu_0(\bX))]$. While these estimands are not identical, the difference can be made arbitrarily small by decreasing the parameter $\varepsilon$.\footnote{However, as we decrease $\varepsilon$, we may increase the variance of the corresponding estimator \cite{yang2018asymptotic, branson2023causal}, leading to a bias-variance tradeoff.} Proposition \ref{prop:10} gives the influence function for $u_{2, \varepsilon}$.

\begin{proposition}\label{prop:10}
For a fixed $\varepsilon > 0$, the uncentered influence function for $\ell_{2, \varepsilon}$ is 

\begin{align}
\nonumber\varphi_{\ell_2}(\bO; \eta) &= \varphi_{1,1}(\bO) - \varphi_{1,0}(\bO) \\
&-\left(\frac{\I\{C = 0, A = 1\}}{(1-\pi_1(\bX))e_1(\bX)} - \frac{\I(C = 0, A = 0)}{(1 - \pi_0(\bX))e_0(\bX)}\right)\{Y - \mu_A(\bX)\}\pi_0(\bX)\Phi_{\varepsilon}(\mu_1(\bX) - \mu_0(\bX)) \\
\nonumber&- \frac{\I\{A = 0\}}{e_0(\bX)}\{C - \pi_0(\bX)\}[\mu_1(\bX) - \mu_0(\bX)]\Phi_{\varepsilon}(\mu_1(\bX) - \mu_0(\bX)) \\
\nonumber&- [\mu_1(\bX) - \mu_0(\bX)]\pi_0(\bX)\Phi_{\varepsilon}(\mu_1(\bX) - \mu_0(\bX)) \\
&- [\mu_1(\bX) - \mu_0(\bX)]\pi_0(\bX)\phi_{\varepsilon}(\mu_1(\bX) - \mu_0(\bX))\{\varphi_{1, 1}(\bO) - \varphi_{1, 0}(\bO) - \mu_1(\bX) + \mu_0(\bX)\}.\label{eqn:final}
\end{align}
\end{proposition}
where $\phi_\varepsilon$ is the probability density function corresponding to $\Phi_\varepsilon$.

\begin{remark}
The influence function for $u_{2, \varepsilon}$ is similar but reverses the signs of $\mu_1(\bX) - \mu_0(\bX)$ in the functions $\Phi_{\varepsilon}$ and $\phi_{\varepsilon}$, as well as in the term $\varphi_{1, 1}(\bO) - \varphi_{1, 0}(\bO) - \mu_1(\bX) + \mu_0(\bX)$. The full expression is available in Appendix \ref{app:addtlres}.
\end{remark}

To obtain an estimator of $\ell_{2, \epsilon}$, we plug in estimates of the nuisance components and compute the empirical average of these quantities. Theorem \ref{thm:2} in Appendix \ref{app:addtlres} establishes the conditions under which this estimator is root-$n$ consistent and asymptotically normal.

As an alternative option, one can estimate $\ell_2$ ($u_2$) by eliminating the final line \ref{eqn:final}, and replacing $\Phi_\epsilon$ with $\I(\mu_1(\bX) > \mu_0(\bX))$ ($\I(\mu_1(\bX) \le \mu_0(\bX))$). This estimator will also be root-n consistent and asymptotically normal under similar conditions to those given Theorem \ref{thm:2}. We formalize this result in Theorem \ref{thm:3} in Appendix \ref{app:addtlres}, invoking an additional ``margin condition'' that essentially requires the density of $\mu_1(\bX)-\mu_0(\bX)$ to be bounded. This condition is likely reasonable for many applications, though would be violated if, for example, the sharp null held with respect to the true conditional treatment effects, and there were no differential dropout, so that $\mu_1(\bx)=\mu_0(\bx)$ for all $\bx$.  
 
\section{Simulations}\label{sec:simulations}

We assess the theoretical properties of our proposed estimators through a simulation study. We begin by generating a population of size $N = 2,000,000$. For each individual $i=1,...,N$, we draw a single covariate independently $X_i \sim \mathcal U[-3, 3]$ and generate $A_i, C_i, Y_i$ according to the nuisance functions depicted below in Figure \ref{fig:0}. As illustrated, the propensity score model $e(x)$, the outcome risk $\mu_a(x)$, and the missingness probabilities $\pi_a(x)$ are all increasing functions of $x$. A complete description of the data-generating process is provided in Appendix \ref{app:dgp}. 
\begin{figure}
 \begin{center}
     \includegraphics[scale=0.6]{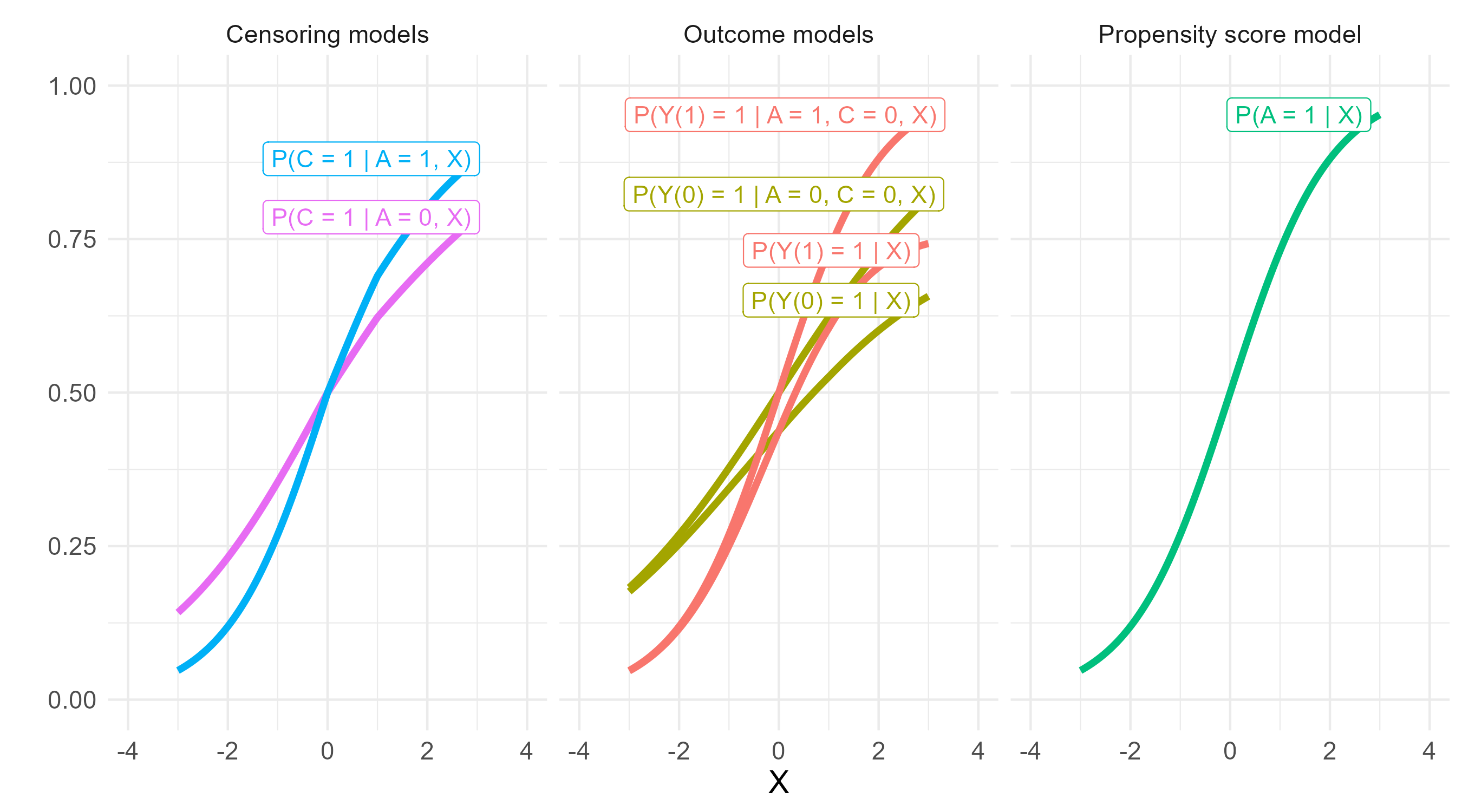}
 \end{center}
 \caption{Nuisance components used in simulation including missingness ($\pi_a(x)$), outcome ($\mu_a(x), \mu^*_a(x)$), and propensity score ($e_1(x)$) models.}\label{fig:0}
\end{figure}

We focus on estimating the following six quantities: $\Psitilde$ as defined in \eqref{observed-ate}, $\Omega_1 = \E[\pi_1(X)]$,  $\Omega_2 = \E[\pi_0(X)]$, $\Omega_3 = \E[\pi_1(X)\mu_1(X)]$, $\Omega_4 = \E[\pi_0(X)\mu_0(X)]$, and $\Omega_5 = \E[\pi_0(X)(\mu_1(X) - \mu_0(X))\Phi_\epsilon(\mu_1(X) - \mu_0(X))]$. Since the proposed bounds are almost entirely linear functions of these quantities -- and the proposed estimators are also linear in the estimated uncentered influence functions of these same quantities -- simulation results verifying the expected performance of each component quantity suffices to show the expected performance of any of the proposed estimates of the bounds. 

Next, we take 1,000 random samples of size $n = 1,000$. Following the simulation framework of \cite{kennedy2023towards}, for each nuisance parameter $\eta_j$, we simulate an estimated function $\hat{\eta}_j$ by introducing stochastic perturbations:
$\hat{\eta}_j = \text{expit}(\text{logit}(\eta_j) + \mathcal{N}(C_{1j}n^{-\alpha}, C_{2j}n^{-2\alpha}))$
for constants $C_{1j}$ and $C_{2j}$ that are specific to each nuisance function and parameter to ensure that we generate estimates that satisfy $\|\hat{\eta} - \eta\| = \mathcal{O}_p(n^{-\alpha})$ in the $L_2(\P)$ norm. Using the simulated estimates $\hat{\eta}$, we construct the estimated influence functions and take the empirical average and variance to construct confidence intervals.

We consider two scenarios for the convergence rate parameter $\alpha$: $0.3$ and $0.1$. In the first case, the convergence rate satisfies the conditions of Theorem \ref{thm:1}, and therefore the estimators are theoretically fast enough to achieve root-$n$ consistency and asymptotic normality. However, the second case -- $\alpha = 0.1$ -- does not satisfy the rate conditions of Theorem \ref{thm:1}, and therefore, we should not expect the estimator to perform well. If we believed that this was the convergence rate for an optimal non-parametric estimator, we would not be able to obtain a root-n convergent estimator without correctly specifying the nuisance component models.

The first two headers in Table \ref{tab:1} summarize these results, evaluating the bias, root mean squared error (RMSE), and coverage rates for each value of $\alpha$ across the 1,000 random samples. As expected, when $\alpha=0.3$, the estimators are approximately unbiased and the confidence intervals achieve the nominal $95\%$ coverage. In contrast, when $\alpha=0.1$, the estimators display increased bias and reduced coverage probabilities, aligning with our theoretical expectations.

To evaluate the practical performance of our estimators, we conduct a second simulation study where we estimate the nuisance functions $\hat{\eta}_j$ directly using sample splitting and SuperLearner to stack a generalized linear model and random forest for samples of size $n = 2,000$. Given the relatively simple data-generating processes, we expect the nuisance estimators to achieve convergence rates greater than $n^{-1/4}$, yielding similar results to the $\alpha = 0.3$ setting above. The final column in Table \ref{tab:1} displays these results and confirms this expectation. The estimators demonstrate minimal bias and confidence intervals with coverage rates close to the nominal $95\%$ level.

\begin{table}
\caption{Simulation results of six estimands for different convergence rates of nuisance parameters \label{tab:1}}
\centering
\begin{tabular}[t]{cccccccccc}
\toprule
\multicolumn{1}{c}{ } & \multicolumn{3}{c}{$\alpha = 0.3$} & \multicolumn{3}{c}{$\alpha = 0.1$} & \multicolumn{3}{c}{SuperLearner} \\
\cmidrule(l{3pt}r{3pt}){2-4} \cmidrule(l{3pt}r{3pt}){5-7} \cmidrule(l{3pt}r{3pt}){8-10}
Estimand & Bias & RMSE & Coverage & Bias & RMSE & Coverage & Bias & RMSE & Coverage\\
\midrule
$\tilde\Psi$ & 0.00 & 0.07 & 0.94 & 0.06 & 0.14 & 0.92 & 0.00 & 0.06 & 0.94\\
$\Omega_1$ & 0.00 & 0.02 & 0.95 & 0.03 & 0.03 & 0.67 & 0.00 & 0.02 & 0.95\\
$\Omega_2$ & 0.01 & 0.04 & 0.93 & 0.24 & 0.26 & 0.29 & 0.00 & 0.02 & 0.96\\
$\Omega_3$ & 0.00 & 0.02 & 0.94 & 0.05 & 0.06 & 0.62 & 0.00 & 0.01 & 0.96\\
$\Omega_4$ & 0.00 & 0.04 & 0.95 & 0.10 & 0.14 & 0.59 & 0.00 & 0.05 & 0.94\\
$\Omega_5$ & 0.00 & 0.05 & 0.93 & 0.00 & 0.07 & 0.97 & 0.00 & 0.05 & 0.92\\
\bottomrule
\end{tabular}
\end{table}

\section{Application}\label{sec:application}
Antipsychotic drugs (APs) are a mainstay in the treatment of schizophrenia. However, some APs are associated with safety risks, including A1c disregulation, type 2 diabetes, and other cardiometabolic morbidity. Randomized studies comparing APs on safety events that, like type 2 diabetes, take months or years to develop, are virtually non-existent; similarly, there is very limited randomized evidence on the effectiveness of APs among elderly individuals with schizophrenia, a population who may plausibly respond differently to these drugs relative to younger patients. Consequently,  clinicians must rely on evidence from observational studies to guide their clinical decision-making. Health insurance administrative databases, especially when paired with state-of-the-art causal inference methods, offer the promise of helping inform the relative effectiveness and safety in usual care populations. However, the potential for informative missingness is particularly acute in many of these databases because reasons for missingness are often not ascertained or communicated to the researcher.

We examined data on $n = 11,618$ elderly adults (age greater than 66) insured by Medicare Advantage diagnosed with schizophrenia who were relatively new users of APs identified in Inovalon Insights LLC's Payer Sourced Claims Data (2017-2020). For more details on inclusion criteria in a similar population, see \cite{agniel2024revisiting}.

We studied the causal effect of each of five common APs compared with taking any other AP in our data: aripiprazole ($n = 975$), haloperidol ($n = 1,024$), olanzapine ($n = 1,798$), quetiapine ($n = 3,702$), and risperidone ($n = 3,036$). People were assigned to a particular AP based on their first qualifying fill in the database. A total of 1,083 patients took other APs. We measured two outcomes that may play an important role in clinicians' prescribing decision-making: diagnosis of type 2 diabetes (which informs AP safety) and inpatient hospitalization for any psychiatric condition (which informs AP effectiveness), both measured at 12 months after the initial AP fill. Missingness was common ($> 20\%$) for both outcomes in all AP groups.

Estimation of all nuisance components was undertaken using random forests, and cross-fitting with two folds to separate evaluation of the influence function from nuisance function estimation. To account for additional randomness due to sample splitting, we repeated sample-splitting and estimation 11 times, using different seeds each time. We took estimates to be the median of the 11 estimates, and we computed standard errors using the approach in \cite{chernozhukov2018double}, which accounts for additional variability due to different seeds across the 11 estimates.

In our safety analyses, we find that three of the APs have lower risk of diabetes than the other APs in the naive analysis assuming all missingness is non-informative. We naively estimate that aripiprazole lowers diabetes risk by 4.9 percentage points (pp), haloperidol by 11.8pp, and quetiapine by 4.2pp, compared with treatment with all other APs. The CIs for haloperidol and quetiapine include only sizable negative effects, suggesting that the naive analysis points to strong evidence of risk reduction, while the CI for aripiprazole includes both large negative effects and some small positive effects, suggesting weaker evidence of risk reduction. Results for olanzapine and risperidone do not provide strong evidence in favor of risk reduction or increase. Similarly, in our effectiveness analyses examining risk of any psychiatric inpatient hospitalization, haloperidol and quetiapine show strong evidence of risk reduction (-2.1pp and -1.3pp, respectively), while we find no strong evidence in either direction for the other APs. See Table \ref{real-data:naive-estimates} for all estimates and CIs.

\begin{table}
\caption{Naive ATE estimates (percentage points) and confidence intervals for five APs on 12-month risk of diabetes and any psychiatric inpatient hospitalization.}\label{real-data:naive-estimates}
\centering
\begin{tabular}{l|l|l}
\hline
\textbf{Antipsychotic} & \textbf{$\Psitilde$ [CI], diabetes} & \textbf{$\Psitilde$ [CI], hospitalization}\\
\hline
Aripiprazole & -4.90 [-10.69, 0.89] & -0.60 [-1.66, 0.46]\\
\hline
Haloperidol & -11.76 [-17.28, -6.25] & -2.07 [-3.91, -0.23]\\
\hline
Olanzapine & 1.12 [-1.36, 3.60] & -0.15 [-0.94, 0.64]\\
\hline
Quetiapine & -4.19 [-6.39, -1.99] & -1.32 [-1.95, -0.70]\\
\hline
Risperidone & 1.37 [-0.76, 3.50] & 0.05 [-0.55, 0.66]\\
\hline
\end{tabular}
\end{table}

While this naive analysis suggests that haloperidol and quetiapine (and possibly aripiprazole) may reduce diabetes risk and risk of inpatient hospitalization, this is only under the assumption of non-informative missingness. The most general bounds from \eqref{assumption-free-bounds} are given in Table \ref{real-data:general-bounds} and do not rule out very large effects in either direction for all APs. 

\begin{table}

\caption{Estimates of general bounds on ATE (percentage points) and confidence intervals for five APs on 12-month diabetes risk and any psychiatric inpatient hospitalization.}\label{real-data:general-bounds}
\centering
\begin{tabular}[t]{c|c|c|c|c}
\hline
Antipsychotic & $\ell_0$ [CI], diabetes & $u_0$ [CI], diabetes & $\ell_0$ [CI], hospitalization& $u_0$ [CI], hospitalization\\
\hline
Aripiprazole & -28.73 [-34.57,-22.88] & 24.43 [21.58,27.28] & -42.00 [-43.22,-40.78] & 35.83 [30.41,41.24]\\
\hline
Haloperidol & -36.24 [-41.36,-31.12] & 23.54 [20.71,26.38] & -42.52 [-44.24,-40.81] & 43.15 [39.44,46.86]\\
\hline
Olanzapine & -27.34 [-30.21,-24.47] & 31.45 [29.68,33.22] & -41.79 [-42.87,-40.71] & 40.78 [38.44,43.13]\\
\hline
Quetiapine & -32.14 [-34.28,-30.00] & 26.33 [24.80,27.86] & -41.75 [-42.84,-40.66] & 43.84 [42.21,45.48]\\
\hline
Risperidone & -27.98 [-30.02,-25.94] & 28.87 [27.39,30.35] & -41.78 [-42.81,-40.75] & 41.44 [39.86,43.03]\\
\hline
\end{tabular}
\end{table}

This most general analysis allows informative missingness to either increase or decrease the risk of the outcomes. In this elderly population with schizophrenia, we are most concerned about missingness that is due to causes that might increase their risk of these adverse outcomes. We may encode this using Assumption \ref{asmpt:monotonicity-pos}, and we can use the bounds from Proposition \ref{prop:2} to provide more narrow bounds. Based on this proposition, we have a negative upper bound for the ATE ($u_0^{\text{pos}} < 0$) as long as the amount of informative missingness in the treated group ($\delta_{1u}$) is not too high. For example, $u_0^{\text{pos}} < 0$ for haloperidol's effect on diabetes as long as no more than 40\% of the missingness is informative. The estimates for aripiprazole's and quetiapine's effects on diabetes are slightly less robust, requiring no more than 21\% and 17\% informative missingness, respectively. The psychiatric inpatient hospitalization (effectiveness) analyses, on the other hand, are not robust to informative missingness at all. Even the largest magnitude effect on inpatient hospitalizations requires no more than 5\% informative missingness. 

Finally, we can examine what type of increased risk due to informative missingness would be required to explain away the naive estimate $\Psitilde$. Under Assumption \ref{asmpt:known-risk}, we require $\tau > -\Psitilde/\E[\mu_1(\bX)\pi_1(\bX)] + 1$, if $\Psi_0$ has any possibility of being non-negative. For the ATE of haloperidol on diabetes, we require informative missingness to increase diabetes risk by a factor of 10.7 to explain the naive ATE estimate. Similarly, aripiprazole requires $\tau > 4.0$ -- a still sizable increase -- while quetiapine is less robust requiring a more modest increased risk of diabetes of a factor of $\tau > 1.7$.\footnote{These numbers are estimated; for simplicity, we do not quantify uncertainty around these estimates.}

In short, unless the severity of informative missingness is quite high and assuming that there is no unmeasured confounding, we find the strongest evidence that haloperidol decreases the risk of diabetes among elderly adults with schizophrenia relative to all other drugs we study, a result that is consistent with findings from \cite{poulos2023antipsychotics}. 

\section{Discussion}\label{sec:discussion}

We have proposed a novel framework for bounding treatment effects and conducting sensitivity analyses in the presence of a mixture of informative and non-informative missingness. Just as most observational studies likely suffer from some amount of unobserved confounding, most studies with missingness likely suffer from some amount of informative missingness. To address this, we have developed assumption-lean bounds that are broadly applicable in virtually any study with missingness and a bounded outcome. We also provide a variety of additional assumptions that researchers may use to narrow these bounds or conduct sensitivity analyses in more general settings. Bounding approaches may be preferred when realistic values of sensitivity parameters can reasonably be specified, while ``tipping point''-style sensitivity analyses may be preferred when researchers may only know broadly what values of sensitivity parameters are implausible. In many applications, researchers may instead wish to estimate the bounds over a range of sensitivity parameters, which can both identify a tipping point (or tipping region) and provide a broader picture of how sensitive the bounds are to the assumed sensitivity parameter values.

A limitation of our proposed approach is that the estimation is not guaranteed to be constrained to the natural bounds of a bounded outcome. For example, in sensitivity analyses where one might wish to evaluate and reason about quantities like $\mathbb{E}[\pi_1(\bX)\mu_1(\bX)]$, these quantities could be negative even if the outcome is binary. A straightforward adaptation of our current approach, which leverages a one-step estimator based on the influence function, would be to implement targeted minimum loss-based estimation (TMLE) \cite{van2011targeted}. TMLE could enforce the necessary constraints, ensuring that the estimators remain within plausible bounds.

We have also limited our work to the setting with a single time point at which missingness occurs. However, many studies face scenarios where missingness occurs at multiple time points, and outcomes may be observed at different intervals. Extending our approach to a longitudinal or survival context is conceptually straightforward but would introduce significant notational complexity and require additional sensitivity parameters. To make such tools practical, further assumptions linking informative missingness to changes in risk over time would be necessary. The development of a comprehensive system for handling these more complex settings is an important area for future research.

We also only consider the case where outcome missingness is due to entirely unknown reasons. While this scenario will hold for many applied settings, such as the analysis of claims datasets, other datasets may contain information about specific competing events, such as death, or other reasons for outcome missingness generally. In the extreme case where $U_I$ is fully observed, $\pi_a^\star$ would be identifiable, and $\Psi_1$ and $\Psi_2$ would be point identified. Moreover, we could obtain narrower bounds on $\Psi_0$ under fewer assumptions, though ultimately we would still require making assumptions on the relationship between $\mu_a^\star$ and $\mu_a$ to bound or identify $\Psi_0$. Unfortunately, in most real-world settings, even when some reasons for outcome missingness are known, the remaining missing outcomes will be due to some unknown combination of informative and non-informative components. Extending our framework to this setting would be an interesting avenue for future research.

Finally, we have mentioned repeatedly how we have drawn inspiration for this work from the extensive and growing literature on bounding causal effects in the presence of unobserved confounding. We can extend our approach to account for \textit{both} informative missingness and unobserved confounding by adding assumptions about the relationship between the expected potential outcome for individuals where the outcome was observed versus unobserved. We provide a high-level sketch of this approach using the ATE and the composite ATE as the target of inference in Appendix \ref{app:adtl}. Additional work on sensitivity analyses for both of these assumptions simultaneously would be valuable future research.

\noindent\textbf{Acknowledgments:} The authors are indebted to Ben Buzzee, Harvard Medical School, for creation of the antipsychotic dataset. The authors would also like to thank two anonymous reviewers and the Associate Editor for helpful comments, questions, and suggestions that substantially improved the quality of this manuscript. Finally, the authors used ChatGPT for assistance proofreading and Claude Code to write and document code. All mistakes are our own. \hfill \break

\noindent\textbf{Research funding:} This work was funded by Grant R01-MH130213 from the National Institute of Mental Health. \hfill \break

\noindent\textbf{Conflict of interest:} Authors state no conflicts of interest. \hfill \break

\noindent\textbf{Author contributions:} All authors contributed to the formulation and development of this work. \hfill \break

\noindent\textbf{Data availability statement:} Code to replicate the simulations is available at \url{https://github.com/mrubinst757/mixture-missingness}. The data for the Application is not publicly available.

\bibliographystyle{plainnat}
\bibliography{bibbb.bib}

\section{Additional results}\label{app:addtlres}

\subsection{Bounds}

We present the alternative formulation of Assumption \ref{prop:3} that results in sharp bounds on $\Psi_0$, briefly discussed in Remark~\ref{rmk:prop3}.

\begin{assumption}[Bounded outcome risk, alternative]\label{asmpt:bounded-risk2}
\[
    \P\left[\tau_a\inv \leq \frac{\mu_a^\star(\bX)}{\mu_a(\bX)} \leq \min\left\{\frac{1}{\mu_a(\bX)},\tau_a\right\} \right] = 1, \text{ for some } \tau_a \ge 1, \quad a = 0, 1.
\]
\end{assumption}

\noindent Proposition \ref{prop:3a} gives an expression for bounds that may be derived under assumption \ref{asmpt:bounded-risk2}.

\begin{proposition}\label{prop:3a}
Under Assumptions \ref{asmpt:non-informative}-\ref{asmpt:cens-positivity}, either \ref{asmpt:bounded-missingness} or \ref{asmpt:known-missingness}, and \ref{asmpt:bounded-risk2}, $\widetilde{\ell_0}^\star \le \Psi_0 \le \widetilde{u}^\star$, where
\begin{align*}
    \widetilde \ell_0^\star = \tilde\Psi + \E[\delta_{1u}\pi_1(\bX)\mu_1(\bX)(\tau_1^{-1}-1)] - \E[\delta_{0u}\pi_0(\bX)\min\{(1-\mu_0(\bX)),\mu_0(\bX)(\tau_0-1)\}], \\
    \widetilde u_0^\star = \tilde\Psi + \E[\delta_{1u}\pi_1(\bx)\min\{(1-\mu_1(\bX)),\mu_1(\bX)(\tau_1-1)\}] - \E[\delta_{0u}\pi_0(\bX)\mu_0(\bX)(\tau_0^{-1}-1)].
\end{align*}    
\end{proposition}

\subsection{Estimation}

The corollaries and theorems below provide expressions for influence-function based estimators of all proposed bounds and identifying expressions in the main paper as well as their asymptotic properties. Throughout we let $\PP f = \mathbb{E}[f(\bO) \mid D_0^n]$, where $D_0^n$ is a training data of size $n$. Notice that $\PP(f(\bO)) = \E(f(\bO))$ for a fixed function $f$; however, we also consider estimated functions $\hat{f}$ that are learned on $D_0^n$, so that $\PP\hat{f}$ is fixed and not equal to $\mathbb{E}[\hat{f}]$. We let $a \lesssim b$ denote $a \le C b$ for some universal constant $C$. Recall that $\|f\|$ denotes the $L_2(\P)$ norm, i.e. $\left(\int f(\bo)^2d\P(\bo)\right)^{1/2}$, and we assume throughout that $Y$, $\mu_a(\bX)$, and $\hat\mu_a(\bX)$ are uniformly bounded.

\begin{corollary}\label{corr:1}
For fixed and finite values of the sensitivity parameters, the influence-function based estimators of the proposed bounds on $\Psi_0$ are,

\begin{align*}
\hat{\ell}_0^\star &= \mathbb{P}_n[\hat{\varphi}_{1,1}(\bO) - \hat{\varphi}_{1,0}(\bO) - (\delta_{1u}\hat{\varphi}_{3,1}(\bO) + \delta_{0u}\{\hat{\varphi}_{2,0}(\bO) - \hat{\varphi}_{3,0}(\bO))\}] \\
\hat{u}_0^\star &= \mathbb{P}_n[\hat{\varphi}_{1,1}(\bO) - \hat{\varphi}_{1,0}(\bO) + \delta_{1u}\{\hat{\varphi}_{2,1}(\bO) - \hat{\varphi}_{3,1}(\bO)\} + \delta_{0u}\hat{\varphi}_{3,0}(\bO)] \\
\hat{\ell}_0^{\text{pos}} &= \mathbb{P}_n[\hat{\varphi}_{1,1}(\bO) - \hat{\varphi}_{1,0}(\bO) - \delta_{0u}(\hat{\varphi}_{2,0}(\bO) - \hat{\varphi}_{3,0}(\bO))] \\
\hat{u}_0^{\text{pos}} &= \mathbb{P}_n[\hat{\varphi}_{1,1}(\bO) - \hat{\varphi}_{1,0}(\bO) + \delta_{1u}(\hat{\varphi}_{2,1}(\bO) - \hat{\varphi}_{3,1}(\bO))] \\
\hat{\ell}_0^{\text{neg}} &= \mathbb{P}_n[\hat{\varphi}_{1,1}(\bO) - \hat{\varphi}_{1,0}(\bO) - \delta_{1u}\hat{\varphi}_{3,1}(\bO)] \\
\hat{u}_0^{\text{neg}} &= \mathbb{P}_n[\hat{\varphi}_{1,1}(\bO) - \hat{\varphi}_{1,0}(\bO) + \delta_{0u}\hat{\varphi}_{3,0}(\bO)] \\
\hat{\tilde{\ell}}_0 &= \mathbb{P}_n[\hat{\varphi}_{1,1}(\bO) - \hat{\varphi}_{1,0}(\bO) + \delta_{u1}(\tau_{1}^{-1} - 1)\hat{\varphi}_{3,1}(\bO) - \delta_{u0}(\tau_0 - 1)\hat{\varphi}_{3,0}(\bO)] \\
\hat{\tilde{u}}_0 &= \mathbb{P}_n[\hat{\varphi}_{1,1}(\bO) - \hat{\varphi}_{1,0}(\bO) + \delta_{u1}(\tau_{1} - 1)\hat{\varphi}_{3,1}(\bO) - \delta_{u0}(\tau_0^{-1} - 1)\hat{\varphi}_{3,0}(\bO)] \\
\end{align*}

\noindent (where we use $\hat{\varphi}(\bO)$ to denote $\varphi(\bO; \hat\eta)$). For fixed values of the sensitivity parameters and $\varepsilon$, the influence function based estimators of the proposed bounds on $\Psi_1$ and $\Psi_2$ are,

\begin{align*}
\hat{\ell}_1 &= \mathbb{P}_n[\hat{\varphi}_{1,1}(\bO) - \hat{\varphi}_{1,0}(\bO) + \delta_{l1}(\hat{\varphi}_{2,1}(\bO) - \hat{\varphi}_{3,1}(\bO)) - \delta_{u0}(\hat{\varphi}_{2,0}(\bO) - \hat{\varphi}_{3,0}(\bO))] \\
\hat{u}_1 &= \mathbb{P}_n[\hat{\varphi}_{1,1}(\bO) - \hat{\varphi}_{1,0}(\bO) + \delta_{u1}(\hat{\varphi}_{2,1}(\bO) - \hat{\varphi}_{3,1}(\bO)) - \delta_{l0}(\hat{\varphi}_{2,0}(\bO) - \hat{\varphi}_{3,0}(\bO))]\\
\hat{\ell}_{2,\varepsilon} &= \mathbb{P}_n\left[\hat\varphi_{1,1}(\bO) - \hat\varphi_{1,0}(\bO) \right.\\
&\left.-\left(\frac{\I(C = 0, A = 1)}{(1-\hat{\pi}_1(\bX))\hat{e}_1(\bX)} - \frac{\I(C = 0, A = 0)}{(1 - \hat{\pi}_0(\bX))\hat{e}_0(\bX))}\right)\{Y - \hat{\mu}_A(\bX)\}\hat{\pi}_0(\bX)\Phi_{\varepsilon}(\hat{\mu}_1(\bX) - \hat{\mu}_0(\bX)) \right. \\
& -\left. \frac{\I(A = 0)}{\hat{e}_0(\bX)}\{C - \hat{\pi}_0(\bX)\}[\hat{\mu}_1(\bX) - \hat{\mu}_0(\bX)]\Phi_{\varepsilon}(\hat{\mu}_1(\bX) - \hat{\mu}_0(\bX)) \right. \\
&- \left. [\hat{\mu}_1(\bX) - \hat{\mu}_0(\bX)]\hat{\pi}_0(\bX)\Phi_{\varepsilon}(\hat{\mu}_1(\bX) - \hat{\mu}_0(\bX))[\hat{\varphi}_{1, 1}(\bX) - \hat{\varphi}_{1, 0}(\bX) - \hat{\mu}_1(\bX) + \hat{\mu}_0(\bX)] \right. \\
&- \left. [\hat{\mu}_1(\bX) - \hat{\mu}_0(\bX)]\hat{\pi}_0(\bX)\Phi_{\varepsilon}(\hat{\mu}_1(\bX) - \hat{\mu}_0(\bX))\right] \\
\hat{u}_{2,\varepsilon} &= \mathbb{P}_n\left[\hat\varphi_{1,1}(\bO)-\hat\varphi_{1,0}(\bO) \right.\\
&\left.-\left(\frac{\I(C = 0, A = 1)}{(1-\hat{\pi}_1(\bX))\hat{e}_1(\bX)} - \frac{\I(C = 0, A = 0)}{(1 - \hat{\pi}_0(\bX))\hat{e}_0(\bX))}\right)\{Y - \hat{\mu}_A(\bX)\}\hat{\pi}_0(\bX)\Phi_{\varepsilon}(\hat{\mu}_0(\bX) - \hat{\mu}_1(\bX)) \right. \\
&- \left. \frac{\I(A = 0)}{\hat{e}_0(\bX)}\{C - \hat{\pi}_0(\bX)\}[\hat{\mu}_1(\bX) - \hat{\mu}_0(\bX)]\Phi_{\varepsilon}(\hat{\mu}_0(\bX) - \hat{\mu}_1(\bX)) \right. \\
&- \left. [\hat{\mu}_1(\bX) - \hat{\mu}_0(\bX)]\hat{\pi}_0(\bX)\Phi_{\varepsilon}(\hat{\mu}_0(\bX) - \hat{\mu}_1(\bX))[\hat{\varphi}_{1, 0}(\bO) - \hat{\varphi}_{1, 1}(\bO) - \hat{\mu}_0(\bX) + \hat{\mu}_1(\bX)] \right. \\
&- \left. [\hat{\mu}_1(\bX) - \hat{\mu}_0(\bX)]\hat{\pi}_0(\bX)\Phi_{\varepsilon}(\hat{\mu}_0(\bX) - \hat{\mu}_1(\bX))] \right] 
\end{align*}
Finally, under assumption \ref{asmpt:known-missingness} and \ref{asmpt:known-risk}, the influence-function based estimators of $\Psi_0$, $\Psi_1$, and $\Psi_2$ are,

\begin{align*}
\hat{\Psi}_0 &= \mathbb{P}_n[\left(\hat\varphi_{1,1}(\bO) - \hat\varphi_{1,0}(\bO)\right) + \delta_1(\tau_1 - 1)\hat\varphi_{3, 1}(\bO) - \delta_0(\tau_0 - 1)\hat\varphi_{3, 0}(\bO)] \\
\hat{\Psi}_1 &= \mathbb{P}_n[\hat{\varphi}_{1,1}(\bO) - \hat{\varphi}_{1,0}(\bO) + \delta_1(\hat{\varphi}_{2,1}(\bO) - \hat{\varphi}_{3,1}(\bO)) - \delta_0(\hat{\varphi}_{2,0}(\bO) - \hat{\varphi}_{3,0}(\bO))] \\
\hat{\Psi}_2 &=\mathbb{P}_n[\hat{\varphi}_{1,1}(\bO) - \hat{\varphi}_{1,0}(\bO) - \delta_0(\hat{\varphi}_{3, 01}(\bO) - \hat{\varphi}_{3,0}(\bO))],
\end{align*}

\noindent where

\begin{align*}
\hat{\varphi}_{3,01}(\bO) &= \hat{\varphi}_{1, 1}(\bO)\hat{\pi}_0(\bX) + \hat{\varphi}_{2,0}(\bO)\hat{\mu}_1(\bX) - \hat{\mu}_1(\bX)\hat{\pi}_0(\bX).
\end{align*}
and $\varphi_{3,a}$ is defined as before.
\end{corollary}

Throughout the remainder of the Appendices, when $b$ is one of the previously defined estimands, we take $\varphi_{b}$ to be the influence-function of that estimand, and $\hat\varphi_b$ it's influence-function based estimate. For example, $\varphi_{l_1} = \varphi_{1,1}-\varphi_{1,0} + \delta_{\ell 1}(\varphi_{2,1}-\varphi_{3,1}) - \delta_{u 0}(\varphi_{2,0}-\varphi_{3,0})$, and $\P_n\hat\varphi_{l_1}$ is the influence-function based estimator of $P(\varphi_{l_1}(\bO;\eta)) = l_1$.

\begin{corollary}\label{corr:2}
For fixed and finite values of the sensitivity parameters $\Theta$ and a fixed value of $\varepsilon$, consider the influence-function–based estimators
\[
\hat b(\Theta)=\mathbb P_n\left(\hat\varphi(\bO;\Theta)\right),
\]
of the lower bounds $\ell_0^\star,\ell_0^{\text{pos}},\ell_0^{\text{neg}},\tilde\ell_0,\ell_1,\ell_2$, the corresponding upper bounds, and the point-identified parameters $\Psi_0,\Psi_1,\Psi_2$. 
Assume positivity conditions
\[
\P(\pi_a(\bX)<1-\epsilon)=1,
\qquad
\P(\hat\pi_a(\bX)<1-\epsilon)=1,
\qquad
\P(\epsilon < e(\bX)\le 1-\epsilon)=1,
\qquad
\P(\epsilon < \hat e(\bX)\le 1-\epsilon)=1.
\]
Assume that $\|\hat \varphi(\bO; \Theta) - \varphi(\bO; \Theta)\| = o_p(1)$. Then,

\begin{align*}
\mathbb{P}_n\bigl(\hat{\varphi}(\bO;\Theta) - b(\Theta)\bigr)= \mathbb{P}_n\bigl(\varphi(\bO; \Theta) - b(\Theta)\bigr) + o_p(n^{-1/2}) + \mathcal{O}_p\bigl(R_2(\hat{b}(\Theta))\bigr),
\end{align*}

\noindent so that when $R_2(\hat{b}(\Theta)) = o_p(n^{-1/2})$,

\begin{align*}
\sqrt{n}\{\mathbb{P}_n(\hat{\varphi}(\bO;\Theta) - b(\Theta))\} \stackrel{d}\to \mathcal{N}\bigl(0, \text{Cov}(\varphi(\bO;\Theta))\bigr).
\end{align*}

\noindent Moreover, for each proposed estimator $\hat{b}_j(\Theta)$, the expression $R_2(\hat{b}_j(\Theta))$ is given by,

    \begin{align*}
        |R_2(\hat{\ell}_0^\star) |&\lesssim \sum_{a = 0, 1} \|\mu_a - \hat{\mu}_a\|(\|\pi_a - \hat{\pi}_a\| + \|e_a - \hat{e}_a\|) + \|e_{a} - \hat{e}_{a}\|\|\pi_a - \hat{\pi}_a\| \\
        |R_2(\hat{u}_0^\star)| &\lesssim \sum_{a = 0, 1} \|\mu_a - \hat{\mu}_a\|(\|\pi_a - \hat{\pi}_a\| + \|e_a - \hat{e}_a\|) + \|e_{a} - \hat{e}_{a}\|\|\pi_a - \hat{\pi}_a\| \\
        |R_2(\hat{\ell}_0^{\text{pos}})| &\lesssim \sum_{a = 0, 1}\|\mu_a - \hat{\mu}_a\|(\|\pi_a - \hat{\pi}_a\| + \|e_a - \hat{e}_a\|) + \|e_{0} - \hat{e}_{0}\|\|\pi_0 - \hat{\pi}_0\| \\
        |R_2(\hat{u}_0^{\text{pos}})| &\lesssim \sum_{a = 0, 1}\|\mu_a - \hat{\mu}_a\|(\|\pi_a - \hat{\pi}_a\| + \|e_a - \hat{e}_a\|) + \|e_{1} - \hat{e}_{1}\|\|\pi_1 - \hat{\pi}_1\| \\
        |R_2(\hat{\ell}_0^{\text{neg}})| &\lesssim \sum_{a = 0, 1}\|\mu_a - \hat{\mu}_a\|(\|\pi_a - \hat{\pi}_a\| + \|e_a - \hat{e}_a\|) + \|e_{1} - \hat{e}_{1}\|\|\pi_1 - \hat{\pi}_1\| \\
        |R_2(\hat{u}_0^{\text{neg}})| &\lesssim \sum_{a = 0, 1}\|\mu_a - \hat{\mu}_a\|(\|\pi_a - \hat{\pi}_a\| + \|e_a - \hat{e}_a\|) + \|e_{0} - \hat{e}_{0}\|\|\pi_0 - \hat{\pi}_0\| \\
        |R_2(\hat{\tilde \ell}_0)| &\lesssim \sum_{a = 0, 1} \|\mu_a - \hat{\mu}_a\|(\|\pi_a - \hat{\pi}_a\| + \|e_a - \hat{e}_a\|) + \|e_{a} - \hat{e}_{a}\|\|\pi_a - \hat{\pi}_a\| \\
        |R_2(\hat{\tilde u}_0)| &\lesssim \sum_{a = 0, 1} \|\mu_a - \hat{\mu}_a\|(\|\pi_a - \hat{\pi}_a\| + \|e_a - \hat{e}_a\|) + \|e_{a} - \hat{e}_{a}\|\|\pi_a - \hat{\pi}_a\| \\
        |R_2(\hat{\ell}_1)| &\lesssim \sum_{a = 0, 1} \|\mu_a - \hat{\mu}_a\|(\|\pi_a - \hat{\pi}_a\| + \|e_a - \hat{e}_a\|) + \|e_{a} - \hat{e}_{a}\|\|\pi_a - \hat{\pi}_a\| \\
        |R_2(\hat{u}_1)| &\lesssim \sum_{a = 0, 1} \|\mu_a - \hat{\mu}_a\|(\|\pi_a - \hat{\pi}_a\| + \|e_a - \hat{e}_a\|) + \|e_{a} - \hat{e}_{a}\|\|\pi_a - \hat{\pi}_a\| \\
        |R_2(\hat{\Psi}_0)| &\lesssim \sum_{a = 0, 1} \|\mu_a - \hat{\mu}_a\|(\|\pi_a - \hat{\pi}_a\| + \|e_a - \hat{e}_a\|) + \|e_{a} - \hat{e}_{a}\|\|\pi_a - \hat{\pi}_a\| \\
        |R_2(\hat{\Psi}_1)| &\lesssim \sum_{a = 0, 1} \|\mu_a - \hat{\mu}_a\|(\|\pi_a - \hat{\pi}_a\| + \|e_a - \hat{e}_a\|) + \|e_{a} - \hat{e}_{a}\|\|\pi_a - \hat{\pi}_a\| \\
        |R_2(\hat{\Psi}_2)| &\lesssim \sum_{a = 0, 1} \|\mu_a - \hat{\mu}_a\|(\|\pi_a - \hat{\pi}_a\| + \|e_a - \hat{e}_a\|) + \|e_{0} - \hat{e}_{0}\|\|\pi_a - \hat{\pi}_a\| \\
        &+ \|\mu_1 - \hat{\mu}_1\|\|\pi_0 - \hat{\pi}_0\| 
    \end{align*}
\noindent Therefore, for any vector of estimators $\hat{b}(\Theta)$, $R_2(\hat{b}(\Theta))$ is simply the sum of the relevant quantities above.
\end{corollary}

\setcounter{theorem}{1}

\begin{theorem}\label{thm:2}
Let $\hat{\varphi}_{2,\varepsilon} = [\hat{\varphi}_{l_{2, \varepsilon}}, \hat{\varphi}_{u_{2, \varepsilon}}]$, where, for example, $\hat{\varphi}_{u_{2, \varepsilon}} = \varphi_{u_{2, \varepsilon}}(\bO; \hat{\eta})$, and $\hat{\eta}$ is learned on data $D_0^n$. Additionally, let $b_{2, \varepsilon} = [l_{2, \varepsilon}, u_{2, \varepsilon}]$. Assume that $\|\hat{\varphi}_{2,\varepsilon} - \varphi_{2,\varepsilon}\| = o_p(1)$. Then,

\begin{align*}
\mathbb{P}_n\left(\hat{\varphi}_{2, \varepsilon}(\bO) - b_{2, \varepsilon}\right)  = \mathbb{P}_n\bigl(\varphi_{2, \varepsilon}(\bO) - b_{2, \varepsilon}\bigr) + o_p(n^{-1/2}) + \mathcal{O}_p\bigl(R_2(\hat{b}_{2, \varepsilon})\bigr)
\end{align*}

\noindent for 

\begin{align*}
|R_2(\hat{b}_{2, \varepsilon})| &\lesssim \sum_{a=0,1}\|\mu_a - \hat{\mu}_a\|(\|\pi_a - \hat{\pi}_a\| + \|e_a - \hat{e}_a\| + \|\mu_a - \hat{\mu}_a\|) \\
&+ \|\pi_0 - \hat{\pi}_0\|(\|e_0 - \hat{e}_0\| + \|\mu_1 - \hat{\mu}_1\|) + \|\mu_0 - \hat{\mu}_0\|\|\mu_1 - \hat{\mu}_1\|
\end{align*}

\noindent so that when $R_2(\hat{b}_{2, \varepsilon}) = o_p(n^{-1/2})$,

\begin{align*}
\sqrt{n}\{\mathbb{P}_n\bigl(\hat{\varphi}_{2, \varepsilon}(\bO) - b_{2, \varepsilon}\bigr)\} \stackrel{d}\to \mathcal{N}\bigl(0, \text{Cov}(\varphi_{2, \varepsilon})\bigr).
\end{align*}
\end{theorem}

\setcounter{theorem}{2}

\begin{theorem}\label{thm:3}
Let

\begin{align*}
\varphi_{\ell_2}(\bO; \eta)  &= \varphi_{1,1}(\bO) - \varphi_{1,0}(\bO) \\
&-\left(\frac{\I\{C = 0, A = 1\}}{(1-\pi_1(\bX))e_1(\bX)} - \frac{\I(C = 0, A = 0)}{(1 - \pi_0(\bX))e_0(\bX)}\right)\{Y - \mu_A(\bX)\}\pi_0(\bX)\I(\mu_1(\bX) > \mu_0(\bX)) \\
&- \frac{\I\{A = 0\}}{e_0(\bX)}\{C - \pi_0(\bX)\}[\mu_1(\bX) - \mu_0(\bX)]\I(\mu_1(\bX) > \mu_0(\bX)) \\
&- [\mu_1(\bX) - \mu_0(\bX)]\pi_0(\bX)\I(\mu_1(\bX) > \mu_0(\bX)),\\
\varphi_{u_2}(\bO; \eta)  &= \varphi_{1,1}(\bO) - \varphi_{1,0}(\bO) \\
&-\left(\frac{\I\{C = 0, A = 1\}}{(1-\pi_1(\bX))e_1(\bX)} - \frac{\I(C = 0, A = 0)}{(1 - \pi_0(\bX))e_0(\bX)}\right)\{Y - \mu_A(\bX)\}\pi_0(\bX)\I(\mu_1(\bX) \le \mu_0(\bX)) \\
&- \frac{\I\{A = 0\}}{e_0(\bX)}\{C - \pi_0(\bX)\}[\mu_1(\bX) - \mu_0(\bX)]\I(\mu_1(\bX) \le \mu_0(\bX)) \\
&- [\mu_1(\bX) - \mu_0(\bX)]\pi_0(\bX)\I(\mu_1(\bX) \le \mu_0(\bX)) 
\end{align*}

\noindent Let $\hat{\varphi}_{2} = [\hat{\varphi}_{l_{2}}, \hat{\varphi}_{u_{2}}]$, where, for example, $\hat{\varphi}_{u_{2}} = \varphi_{u_{2}}(\bO; \hat{\eta})$, and $\hat{\eta}$ is learned on data $D_0^n$. Additionally, let $b_{2} = [l_{2}, u_{2}]$. 

Assume that for all $t > 0$, there exists $\alpha > 0$ such that $\P\bigl(|\mu_1(\bX)-\mu_0(\bX)| < t\bigr) \le Ct^\alpha$ for some constant $C$, and finally $\|\hat{\varphi}_2 - \varphi_2\| = o_p(1)$. Then,

\begin{align*}
\mathbb{P}_n\left(\hat{\varphi}_{2}(\bO) - b_{2}\right)  = \mathbb{P}_n\left(\varphi_{2}(\bO) - b_{2}\right) + o_p(n^{-1/2}) + \mathcal{O}_p(R_2(\hat b_2))
\end{align*}

\noindent for 

\begin{align*}
|R_2(\hat{b}_{2})| &\lesssim \sum_{a=0,1}\|\mu_a - \hat{\mu}_a\|(\|\pi_a - \hat{\pi}_a\| + \|e_a - \hat{e}_a\|) + \|\pi_0 - \hat{\pi}_0\|\|e_0 - \hat{e}_0\| + (\|\hat\mu_1-\mu_1\|_{\infty} + \|\hat\mu_0-\mu_0\|_\infty)^{1+\alpha}
\end{align*}

\noindent so that when $R_2(\hat{b}_2) = o_p(n^{-1/2})$,

\begin{align*}
\sqrt{n}\{\mathbb{P}_n(\hat{\varphi}_{2}(\bO) - b_{2})\} \stackrel{d}\to \mathcal{N}\bigl(0, \text{Cov}(\varphi_{2})\bigr).
\end{align*}
\end{theorem}

\begin{remark}
The term $\bigl(\|\hat\mu_1-\mu_1\|_{\infty} + \|\hat\mu_0-\mu_0\|_\infty\bigr)^{1+\alpha}$ comes from the margin condition. This will be satisfied for $\alpha = 1$ when $\mu_1-\mu_0$ has a bounded density, in which case this would require $\|\mu_a-\hat\mu_a\|_\infty = o_p(n^{-1/4})$ for $R_2 = o_p(n^{-1/2})$ (along with the other error products).
\end{remark}

Finally, in Theorem \ref{thm:5} we consider doubly-robust style estimators of the bounds $\tilde u_0^\star, \tilde \ell_0^\star$ given in Proposition \ref{prop:3a}.

\begin{theorem}\label{thm:5}
Let
\begin{align*}
\varphi_{\tilde u_0^\star}(\bO;\eta) &= \varphi_{1,1}(\bO;\eta)-\varphi_{1,0}(\bO;\eta) \\
&\quad + \delta_{1u}\Big[\{\varphi_{2,1}(\bO;\eta)-\varphi_{3,1}(\bO;\eta)\}\I\{\tau_1\mu_1(\bX)>1\}
+(\tau_1-1)\varphi_{3,1}(\bO;\eta)\I\{\tau_1\mu_1(\bX)\le 1\}\Big] \\
&\quad - \delta_{0u}(\tau_0^{-1}-1)\varphi_{3,0}(\bO;\eta),
\end{align*}
and
\begin{align*}
\varphi_{\tilde \ell_0^\star}(\bO;\eta) &= \varphi_{1,1}(\bO;\eta)-\varphi_{1,0}(\bO;\eta)
+ \delta_{1u}(\tau_1^{-1}-1)\varphi_{3,1}(\bO;\eta) \\
&\quad - \delta_{0u}\Big[\{\varphi_{2,0}(\bO;\eta)-\varphi_{3,0}(\bO;\eta)\}\I\{\tau_0\mu_0(\bX)>1\}
+(\tau_0-1)\varphi_{3,0}(\bO;\eta)\I\{\tau_0\mu_0(\bX)\le 1\}\Big].
\end{align*}

\noindent Let $\hat{\varphi}_{\tilde b_0^\star} = [\hat{\varphi}_{\tilde \ell_0^\star},\hat{\varphi}_{\tilde u_0^\star}]$, where, for example, $\hat{\varphi}_{\tilde u_0^\star}=\varphi_{\tilde u_0^\star}(\bO;\hat\eta)$ and $\hat\eta$ is learned on data $D_0^n$, and define $\tilde b_0$ analogously.

Assume that for $a=0,1$ there exist constants $C_a,\alpha_a>0$ such that, for all $t>0$,
\[
\P\left(\left|\tau_a\mu_a(\bX)-1\right|<t\right)\le C_a t^{\alpha_a},
\]
and assume finally that $\|\hat{\varphi}_{\tilde b_0^\star}-\varphi_{\tilde b_0^\star}\|=o_p(1)$. Then
\begin{align*}
\mathbb{P}_n\left[\hat{\varphi}_{\tilde b_0^\star}(\bO)-\tilde b_0^\star\right]
=
\mathbb{P}_n\left[\varphi_{\tilde b_0^\star}(\bO)-\tilde b_0^\star\right]
+o_p(n^{-1/2})+\mathcal{O}_p\bigl(R_2(\hat{\tilde{b_0^\star}})\bigr),
\end{align*}
for
\begin{align*}
|R_2(\hat{\tilde b}_0)|
&\lesssim
\sum_{a=0,1}\|\mu_a-\hat\mu_a\|
\bigl(\|\pi_a-\hat\pi_a\|+\|e_a-\hat e_a\|\bigr)
+\sum_{a=0,1}\|e_a-\hat e_a\|\|\pi_a-\hat\pi_a\| \\
&\quad + \bigl(\|\hat\mu_1-\mu_1\|_\infty\bigr)^{1+\alpha_1}
+ \bigl(\|\hat\mu_0-\mu_0\|_\infty\bigr)^{1+\alpha_0}.
\end{align*}
Therefore, when $R_2(\hat{\tilde b}_0)=o_p(n^{-1/2})$ (requiring, $\alpha_a \ge 1,\quad a=0, 1$),
\begin{align*}
\sqrt{n}\left\{\mathbb{P}_n\left(\hat{\varphi}_{\tilde b_0}(\bO)-\tilde b_0\right)\right\}
\stackrel{d}{\to}
\mathcal{N}\left(0,\text{Cov}(\varphi_{\tilde b_0})\right).
\end{align*}
\end{theorem}

\begin{remark}
    Unlike the smooth bounds in Corollary \ref{corr:2}, the estimators of considered in Theorem \ref{thm:3} and \ref{thm:5} are not doubly robust in the usual sense because of the use of the indicator function. Thus consistency of the relevant outcome regression is required in order for the margin term to vanish. Conditional on this, the remaining nuisance estimation error retains the same product-of-errors structure as in the smooth case. The margin conditions are not overly strong conditions, requiring in effect that the density of $\mu_1(\bX)-\mu_0(\bX)$ is bounded for Theorem \ref{thm:3}, and that the densities of $\mu_a(\bX)$ is bounded for $a = 0, 1$, giving exponents $\alpha = 1$ on the quantity $\|\hat\mu_a-\mu_a\|_\infty^{1+\alpha}$.
    \end{remark}

\begin{remark}
    Instead of the estimator $\P_n\bigl(\hat\varphi_{\tilde b_0^\star}\bigr)$, analogous to the relationship of $\Psi_{2,\varepsilon}$ to $\Psi_2$ and their respective estimators, we may construct an influence-function based estimator $\P_n\bigl(\hat\varphi_{b_0^\star,\varepsilon}\bigr)$ of the quantity $\tilde b_{0,\varepsilon}^\star$, which we define as the same as $\tilde b_{0}^\star$, but replacing $\I$ with $\Phi_\varepsilon$ for some fixed $\varepsilon > 0$. In this case $\varphi_{\tilde b_0^\star,\varepsilon}$ has the same form as $\varphi_{\tilde b_0^\star}$, but with $\Phi_{\varepsilon}$ replacing $\I$, but with the following additional term,

    \begin{align*}
        \delta_{au}\tau_a\varphi_{1,a}(\bO)\left\{\pi_a(\bX)(1-\mu_a(\bX))\phi_{\varepsilon}\bigl(\tau_a\mu_a(\bX) -1\bigr) + (\tau_a-1)\pi_a(\bX)\mu_a(\bX)\phi_\varepsilon\bigl(1-\tau_a\mu_a(\bX)\bigr)\right\},
    \end{align*}
    where we take $a = 1$ and add this quantity to obtain an estimator of $\tilde u_{0,\varepsilon}^\star$; and similarly take $a = 0$ and add this quantity to obtain an estimator of $\tilde \ell_{0,\varepsilon}^\star$. The proofs follow analogously to those for the results above, and we omit them for brevity.
\end{remark}

\subsection{Uniform inference}\label{app:multiplier}

Under only mildly stronger conditions than those required to ensure asymptotic normality of the bound estimators for fixed values of the sensitivity parameters, one may construct confidence bands that hold simultaneously over a continuous set of sensitivity parameters by using the multiplier bootstrap to approximate the distribution of the supremum over the standardized statistics. The following result closely follow previously derived results for uniform inference in sensitivity analysis outlined \cite{kennedy2019nonparametric, mauro2020instrumental}.

\begin{theorem}\label{thm:4}
Assume the regularity conditions of Coroallary \ref{corr:2}, and let $\Theta$ denote any of the previously discussed sensitivity parameters and let $\Pi$ be a compact set of values these parameters may take. Let $\varphi(\bO;\eta,\Theta)$ denote a p-dimensional vector of uncentered efficient influence functions (EIF) of the previously discussed bounds, so that $\varphi_j(\bO;\eta,\Theta)$ are a function of the nuisance functions $\eta$ and sensitivity parameters $\Theta$ for $j = 1,\dots,p$. We then define the set $\Delta := \Pi \times \{1,\dots,p\}$ with elements $\Theta^\star$, so that, for example, for $\Theta^\star = (j,\Theta) \in \Delta$, $\sigma(\Theta^\star):=\sigma_j(\Theta)$.

Let $b(\Theta^\star)$ denote $\E(\varphi(\bO;\eta,\Theta^\star))$, the statistical target of inference, and let $\hat b(\Theta^\star) = \P_n(\varphi(\bO;\hat\eta,\Theta^\star))$. Recall that $R_2(\hat b(\Theta^\star)) = \E(\hat b(\Theta^\star) - b(\Theta^\star))$, where these expressions for bounds on the specific elements $R_2(\hat b_j(\Theta))$ are given in Corollary \ref{corr:2}. Let
\[
\sigma^2(\Theta^\star) =\E\{(\varphi(\bO; \eta,\Theta^\star)-b_j(\Theta^\star))^2\}, \qquad
\hat\sigma^2(\Theta^\star)=\mathbb \P_n\{(\varphi(\bO; \hat\eta, \Theta^\star) - \hat b(\Theta^\star))^2\}.
\]
Assume that

\begin{enumerate}
    \item $\sup_{\Theta^\star\in\Delta}\left|\frac{\hat\sigma(\Theta^\star)}{\sigma(\Theta^\star)}-1\right|=o_p(1)$
    \item $\sup_{\Theta^\star\in\Delta}\| \hat\varphi(\Theta^\star)-\varphi(\Theta^\star)\|_2 = o_p(1)$
    \item $\sup_{\Theta^\star\in\Delta}|R_2(\hat b(\Theta^\star))| = o_p(n^{-1/2}).$
\end{enumerate}
Then
\[
\sup_{\Theta^\star\in\Delta} \left\lvert\sqrt{n}\frac{\hat b(\Theta^\star) -  b(\Theta^\star)}{\hat\sigma(\Theta^\star)} - \sqrt{n}\frac{\P_n(\varphi(\bO;\eta,\Theta^\star)) -b(\Theta^\star)}{\sigma(\Theta^\star)}\right\rvert = o_p(1).
\]
\end{theorem}

\begin{remark}
The conditions in Theorem \ref{thm:4} are mild in our setting. The first imposes only uniform consistency of the variance estimators. The second and third conditions impose no additional substantive assumptions: the bounds and their second-order remainder terms may be written as finite linear combinations of nuisance-dependent quantities with coefficients depending on $\Theta$. Since $\Delta$ is compact, these coefficients are uniformly bounded over $\Theta^\star\in\Delta$. Thus, the pointwise consistency and rate conditions from Corollary \ref{corr:2} extend immediately to uniform versions by taking suprema over $\Theta^\star$.
\end{remark}

\begin{remark}
    From this result, one may then apply Theorem 4 from \cite{kennedy2019nonparametric} to apply the multiplier bootstrap to generate some value $\hat c_\alpha$ such that
    \[
    \P\left(\hat b(\Theta^\star) - \frac{\hat c_\alpha \hat \sigma(\Theta^\star)}{\sqrt{n}} \le b(\Theta^\star) \le \hat b(\Theta^\star) + \frac{\hat c_\alpha \hat \sigma(\Theta^\star)}{\sqrt{n}} \text{, for all } \Theta^\star \in \Delta\right) = 1 - \alpha + o(1).
    \]
\end{remark}

\section{Proofs}\label{app:proofs}

\subsection{Bounds}\label{app:bounds}

\begin{proof}[Proof of Proposition 1: expression for $\Psi_0$]
\begin{align*}
\E[Y(a) \mid \bX = \bx] &= \E[Y(a) \mid \bx, A = a] \\
&= \E[Y \mid \bx, A = a] \\
&= \E[Y \mid \bx, A = a, C = 1]\P(C = 1 \mid A = a, \bx) \\
&+ \E[Y \mid \bx, A = a, C = 0]\P(C = 0 \mid A = a, \bx) \\
&= \{\E[Y \mid \bx, A = a, C = 1, U_I = 1, U_{NI} =0]\P(U_I = 1, U_{NI} =0\mid \bx, C = 1,A = a)\\
&+ \E[Y \mid \bx, A = a, U_{NI} =1, U_I =0]\P(U_{NI} = 1, U_I = 0\mid \bx, C= 1,A = a)\}\pi_a(\bx) \\
&+ \mu_a(\bx)(1-\pi_a(\bx)) \\
&= \left\{\E[Y \mid \bx, A = a, U_I = 1, U_{NI} =0]\left(\frac{\pi_a^\star(\bx)}{\pi_a(\bx)}\right)\right.\\
&\left.+ \E[Y \mid \bx, A = a, U_{NI} =1, U_I =0]\left(1-\frac{\pi_a^\star(\bx)}{\pi_a(\bx)}\right)\right\}\pi_a(\bx) + \mu_a(\bx)(1-\pi_a(\bx)) \\
&= \mu_a^\star(\bx)\pi_a^\star(\bx) + \mu_a(\bx)(\pi_a(\bx)-\pi_a^\star(\bx)) + \mu_a(\bx)(1-\pi_a(\bx)) \\
&= \mu_a(\bx) + \pi^\star_a(\bx)(\mu^\star_a(\bx)-\mu_a(\bx))
\end{align*}

\noindent where the first equality holds by ignorability, the second by consistency, the third by the law of iterated expectations, and the fourth by iterating expectations and the fact that $U_I$ and $U_{NI}$ are mutually exclusive. To obtain the fifth equality, further notice that

\begin{align*}
&\P(U_I = 1, U_{NI} = 0 \mid C = 1, \bx, A = a) \\
&= \P(U_I = 1  \mid C = 1, \bx, A = a) \\
&= \frac{\P(U_I = 1, C = 1 \mid \bx, A = a)}{\P(C = 1 \mid \bx, A = a)} \\
&= \frac{\P(U_I = 1 \mid \bx, A = a)}{\P(C = 1 \mid \bx, A = a)}
\end{align*}

\noindent where the first equality holds due to mutual exclusivity, the second by definition, and the third by the fact that $U_I = 1 \implies C = 1$. By mutual exclusivity we can also see that $\P(U_I = 0, U_{NI} = 1 \mid C = 1, \bx, A = a) = 1 - \P(U_I = 1, U_{NI} = 0 \mid C = 1, \bx, A = a)$. Combining these facts gives us the fifth equality. The sixth equality uses $Y \perp U_{NI} \mid X, A, U_I = 0$ to obtain $\E[Y \mid \bX, A = a, U_{NI} = 1, U_I = 0] = \E[Y \mid \bX, A = a, U_{NI} = 0, U_I = 0] = \mu_a(\bX)$, and the final equality follows by simplifying terms.
\end{proof}

\begin{proof}[Proof of Proposition 2 and Corollary \ref{corr:0}: general bounds on ATE]
Assume 1-5, and either \ref{asmpt:known-missingness} or \ref{asmpt:bounded-missingness}. For $a = 0, 1$, we have that
\begin{align*}
\E[Y(a) \mid \bx] &= \mu_a(\bx) + \pi_a^\star(\bx)[\mu_a^\star(\bx)-\mu_a(\bx)] \\
&\le \mu_a(\bx) + \delta_{au}\pi_a(\bx)[\mu_a^\star(\bx)-\mu_a(\bx)] \\
&\le \mu_a(\bx) + \delta_{au}\pi_a(\bx)[1-\mu_a(\bx)],
\end{align*}

\noindent where the first equality holds by Proposition 1, the second by the fact that $\pi_a^\star(\bx) \le \delta_{au}\pi_a(\bx)$, and the final inequality by the fact that $\mu_a^\star(\bx) \le 1$ for all $\bx$. Using the same logic it is easy to see that

\begin{align*}
\E[Y(a) \mid \bx] &= \mu_a(\bx) + \pi_a^\star(\bx)[\mu_a^\star(\bx)-\mu_a(\bx)] \\
&\ge \mu_a(\bx) + \pi^\star_a(\bx)[0-\mu_a(\bx)] \\
&\ge \mu_a(\bx) - \delta_{au}\pi_a(\bx)\mu_a(\bx).
\end{align*}

\noindent Subtracting the upper bound for $\E[Y(0) \mid \bx]$ from the lower bound for $\E[Y(1) \mid \bx]$ and integrating over the distribution of $\bX$ gives the result for $\ell_{0}^\star$, and the expression for $u_0^\star$ follows analogously. The bounds from Proposition \ref{prop:1} immediately follow setting $\delta_{u1} = \delta_{u0} = 1$.
\end{proof}

\begin{remark}
Under assumptions 1-5, and either \ref{asmpt:known-missingness} or \ref{asmpt:bounded-missingness}, the bounds $(\ell_{0}^\star,u_{0}^\star)$ (and thereby $\ell_0,u_0$) on $\Psi_0$ are sharp. For example, $\ell_{0}^\star$ holds with equality a distribution satisfying $\pi_a^\star(\bx)/\pi_a(\bx) = \delta_{au}$ and $\mu_a^\star(\bx) = 1-a$ for $a = 0, 1$ and for all $\bx$, and $u_0^\star$ holds with equality for a distribution satisfying $\pi_a^\star(\bx)/\pi_a(\bx) = \delta_{au}$ and $\mu_a^\star(\bx) = a$ for $a = 0, 1$ and for all $\bx$. 
\end{remark}

\begin{proof}[Proof of Proposition 3: bounds on ATE, monotonicity]
First assume positive monotonicity (\ref{asmpt:monotonicity-pos}), so that for all $\bx, a$, $\mu_a^\star(\bx) \ge \mu_a(\bx)$. Then

\begin{align*}
\E[Y(a) \mid \bx] &= \mu_a(\bx) + \pi_a^\star(\bx)[\mu_a^\star(\bx)-\mu_a(\bx)] \\
&\ge \mu_a(\bx)
\end{align*}

\noindent where the second inequality follows because this expression is minimized when $\mu_a^\star(\bx) = \mu_a(\bx)$. The results follow from combining these lower bounds with the upper bounds derived in the proof of Proposition 2, where we make assumption \ref{asmpt:known-missingness} or \ref{asmpt:bounded-missingness} and allow $\delta_{1u}$ and $\delta_{0u}$ to take values less than 1. To obtain $\ell_0^{\text{pos}}$, simply subtract the upper bound for $\E[Y(0) \mid \bx]$ from the lower bound for $\E[Y(1) \mid \bx]$ and integrate over the covariate distribution; and the expression for $u_0^{\text{pos}}$ follows analogously.

Next assume negative monotonicity (\ref{asmpt:monotonicity-neg}), so that for all $\bx, a$, $\mu_a^\star(\bx) \le \mu_a(\bx)$. Then

\begin{align*}
\E[Y(a) \mid \bx] &= \mu_a(\bx) + \pi_a^\star(\bx)[\mu_a^\star(\bx)-\mu_a(\bx)] \\
&\le \mu_a(\bx)
\end{align*}

\noindent where the inequality follows because the expression is maximized when $\mu_a^\star(\bx) = \mu_a(\bx)$. As with the bounds under positive monotonicity, we may combine this with the lower bounds from the proof of Proposition 2 to obtain the expressions for $\ell_0^{\text{pos}}$ and $u_0^{\text{pos}}$.
\end{proof}

\begin{remark}
Under assumptions 1-5, \ref{asmpt:known-missingness} or \ref{asmpt:bounded-missingness}, and either positive or negative monotonicity (\ref{asmpt:monotonicity-pos} or \ref{asmpt:monotonicity-neg}), the bounds on $\Psi_0$ are sharp. For example, the bound $l_0^{pos}$ holds with equality for a distribution satisfying $\pi_0^\star(\bx)/\pi_0(\bx) = \delta_{0u}$, $\mu_0^\star(\bx) = 1$, and $\mu_1^\star(\bx) = \mu_1(\bx)$ for all $\bx$; and the bound $u_0^{pos}$ holds with equality for a distribution satisfying $\pi_1^\star(\bx)/\pi_1(\bx) = \delta_{1u}$, $\mu_1^\star(\bx) = 1$, and $\mu_0^\star(\bx) = \mu_0(\bx)$ for all $\bx$. Similarly, the bound $\ell_0^{neg}$ holds with equality, letting $\pi_1^\star(\bx)/\pi_1(\bx) = \delta_{1u}$, $\mu_1^\star(\bx) = 0$, and $\mu_0(\bx) = \mu_0^\star(\bx)$ for all $\bx$, and for $u_0^{neg}$ letting $\pi_0^\star(\bx)/\pi_0(\bx) = \delta_{0u}$, $\mu_0^\star(\bx) = 0$, and $\mu_1(\bx) = \mu_1^\star(\bx)$ for all $\bx$.
\end{remark}

\begin{proof}[Proof of Proposition 4: bounds on ATE, bounded outcome risk]
By assumption, we know that

\begin{align*}
\frac{\mu_a^\star(\bx)}{\mu_a(\bx)} &\le \tau_a \\
&\implies \mu_a^\star(\bx) - \mu_a(\bx) \le \mu_a(\bx)[\tau_a - 1].
\end{align*}

\noindent and therefore,

\begin{align*}
\E[Y(a) \mid \bx] &= \mu_a(\bx) + \pi_a^\star(\bx)[\mu_a^\star(\bx)-\mu_a(\bx)] \\
& \le \mu_a(\bx) + \pi_a^\star(\bx)[\mu_a(\bx)(\tau_a - 1)] \\
& \le \mu_a(\bx) + \delta_{au}\pi_a(\bx)[\mu_a(\bx)(\tau_a - 1)].
\end{align*}

\noindent where the first inequality follows by plugging in the expression above, and the second by noting that $\tau_a - 1 \ge 0$ and by \ref{asmpt:known-missingness} or \ref{asmpt:bounded-missingness}. Similarly,

\begin{align*}
\frac{\mu_a^\star(\bx)}{\mu_a(\bx)} &\ge \tau_a^{-1} \\
&\implies \mu_a^\star(\bx) - \mu_a(\bx) \ge \mu_a(\bx)[\tau_a^{-1} - 1],
\end{align*}

\noindent and therefore,

\begin{align*}
\E[Y(a) \mid \bx] &= \mu_a(\bx) + \pi_a^\star(\bx)[\mu_a^\star(\bx)-\mu_a(\bx)] \\
& \ge \mu_a(\bx) + \pi_a^\star(\bx)[\mu_a(\bx)(\tau_a^{-1} - 1)] \\
& \ge \mu_a(\bx) + \delta_{au}\pi_a(\bx)[\mu_a(\bx)(\tau_a^{-1} - 1)],
\end{align*}

\noindent where the second inequality here follows by \ref{asmpt:bounded-risk} and by the fact that $\tau_a^{-1} - 1 \le 0$. The final results follow again by subtracting the relevant bounds on $\E[Y(a) \mid \bx]$ and integrating over the covariate distribution.
\end{proof}

\begin{remark}\label{rmk:sharp1}
Under assumptions 1-5, \ref{asmpt:known-missingness} or \ref{asmpt:bounded-missingness}, and \ref{asmpt:bounded-risk}, the bounds $\tilde \ell_0, \tilde u_0$ are sharp as long as $\tau_a\mu_a(\bx) \le 1$ for all $\bx$. For example, the upper bound is attained when $\pi^\star_a(\bx)/\pi_a(\bx) = \delta_{au}$ and $\mu_a^\star(\bx) = \tau_a \mu_a(\bx)$ for all $\bx$ and $a = 0, 1$, and similarly the lower bound is attained when $\pi_a(\bx)/\pi_a(\bx) = \delta_{au}$ and $\mu_a^\star(\bx) = \tau_a^{-1} \mu_a(\bx)$ for all $\bx$ and $a = 0, 1$. On the other hand, if $\tau_a\mu_a(\bx) \le 1$ for some $\bx$, sharp bounds may be recovered by using the expression given in Proposition \ref{prop:3a}.
\end{remark}

\begin{proof}[Proof of Proposition 5: point identification of $\Psi_0$]
This result follows immediately from applying \ref{asmpt:known-missingness} and \ref{asmpt:known-risk} to the expression derived in Proposition 1.
\end{proof}

\begin{proof}[Proof of Proposition 7: expression for $\Psi_1$]
\begin{align*}
\P(Z(a) = 1 \mid \bX = \bx) &= 1 - \P(Y(a) = 0, U_I(a) = 0 \mid \bX = \bx) \\
&= 1 - \P(Y(a) = 0 \mid U_I(a) = 0, A = a, \bx)\P(U_I(a) = 0 \mid A = a, \bx) \\
&= 1 - \P(Y = 0 \mid U_I = 0, A = a, \bx)\P(U_I = 0 \mid A = a, \bx) \\
&= 1 - \P(Y = 0 \mid C = 0, A = a, \bx)\P(U_I = 0 \mid A = a, \bx) \\
&= 1 - [1 - \mu_a(\bx)][1 - \pi^\star_a(\bx)] \\
\end{align*}

\noindent where the first equality holds by definition, the second by ignorability, the third by consistency, the fourth by non-informative missingness by $U_{NI}$, and the fifth by definition. This expression then implies that:

\begin{align*}
\E[Z(1) - Z(0) \mid \bx] &= [1 - \mu_0(\bx)][1 - \pi^\star_0(\bx)] - [1 - \mu_1(\bx)][1 - \pi^\star_1(\bx)] \\
&= 1 - \mu_0(\bx) - \pi^\star_0(\bx) + \mu_0(\bx)\pi^\star_0(\bx) - [1 - \mu_1(\bx) - \pi^\star_1(\bx) + \mu_1(\bx)\pi^\star_1(\bx)] \\
&= \mu_1(\bx) - \mu_0(\bx) + \pi_1^\star(\bx)[1 - \mu_1(\bx)] - \pi_0^\star(\bx)[1 - \mu_0(\bx)]
\end{align*}

\noindent where these equalities follow by simplifying terms. Integrating the final expression over the covariate distribution gives the result.
\end{proof}

\begin{proof}[Proof of Proposition 8: expression for $\Psi_2$]
We adopt the convention that $U_I(a_y, a_d) = 1 \implies Y(a_y, a_d) = 0$. Then,

\begin{align*}
\E[Y(1, 0) \mid \bX = \bx] &= \E[Y(1, 0) \mid U_I(1, 0) = 0, \bx]\P(U_I(1, 0) = 0 \mid \bx) \\
&= \E[Y(1, 1) \mid U_I(1, 1) = 0, \bx]\P(U_I(0, 0) = 0 \mid \bx) \\
&= \E[Y(1) \mid U_I(1) = 0, \bx]\P(U_I(0) = 0 \mid \bx) \\
&= \E[Y(1) \mid A = 1, U_I(1) = 0, \bx]\P(U_I(0) = 0 \mid A = 0, \bx) \\
&= \E[Y \mid A = 1, U_I = 0, \bx]\P(U_I = 0 \mid A = 0, \bx) \\
&= \E[Y \mid A = 1, C = 0, \bx]\P(U_I = 0 \mid A = 0, \bx) \\
&= \mu_1(\bx)(1 - \pi_0^\star(\bx))
\end{align*}

\noindent where the first line follows by the law of iterated expectations (and the fact that $\E[Y(1, 0) \mid U_I(1, 0) = 1, \bx] = 0$), the second by the dismissible component conditions, the third by definition, the fourth by ignorability, the fifth by consistency, the sixth by non-informative missingness by $U_{NI}$, and the seventh follows by definition. Identification of $\E[Y(0, 0) \mid \bx] = \mu_0(\bx)(1 - \pi_0^\star(\bx))$ follows analogous steps. Subtracting these expressions and integrating over the covariate distribution gives the result.
\end{proof}

\begin{proof}[Proof of Proposition 9: bounds on $\Psi_1$]
For simplicity, allow that $\mu_a$ and $\mu_a^\star$ denote the regression functions defined using $Z$ as the outcome instead of $Y$. Then,

\begin{align*}
\E[Z(a) \mid \bx] &= \mu_a(\bx) + \pi_a^\star(\bx)[\mu_a^\star(\bx)-\mu_a(\bx)] \\
&= \mu_a(\bx) + \pi_a^\star(\bx)[1-\mu_a(\bx)] \\
& \le \mu_a(\bx) + \delta_{au}\pi_a(\bx)[1-\mu_a(\bx)],
\end{align*}

\noindent where the second equality follows because $\mu_a^\star(\bx) = 1$. Similarly,

\begin{align*}
\E[Z(a) \mid \bx] &= \mu_a(\bx) + \pi_a^\star(\bx)[1-\mu_a(\bx)] \\
& \ge \mu_a(\bx) + \delta_{a\ell}\pi_a(\bx)[1-\mu_a(\bx)],
\end{align*}

\noindent The final results follow again by subtracting the relevant on $\E[Z(a) \mid \bx]$ and integrating over the covariate distribution.
\end{proof}

\begin{remark}
Under assumptions 1-5 and \ref{asmpt:known-missingness} or \ref{asmpt:bounded-missingness}, the bounds $\ell_1$ and $u_1$ are sharp: $\ell_1$ is attained when $\pi^\star_1(\bx)/\pi_1(\bx) = \delta_{1\ell}$ and $\pi^\star_0(\bx)/\pi_0(\bx) = \delta_{0u}$ for all $\bx$, and $u_1$ is attained when $\pi^\star_1(\bx)/\pi_1(\bx) = \delta_{1u}$ and $\pi^\star_0(\bx)/\pi_0(\bx) = \delta_{0\ell}$ for all $\bx$.
\end{remark}

\begin{proof}[Proof of Proposition 9: bounds on $\Psi_2$]

\begin{align*}
\E[Y(1,0) - Y(0, 0) \mid \bx] &= [1-\pi_0^\star(\bx)][\mu_1(\bx) - \mu_0(\bx)] \\
&= \tilde{\Psi}(\bx) - \pi_0^\star(\bx)[\mu_1(\bx) - \mu_0(\bx)] \\
&\ge \tilde{\Psi}(\bx) - \delta_{0u}\pi_0(\bx)[\mu_1(\bx) - \mu_0(\bx)]\I[\mu_1(\bx) > \mu_0(\bx)],
\end{align*}

\noindent where the first equality follows by Proposition 8, the second by rearranging terms, and the third by \ref{asmpt:known-missingness} or \ref{asmpt:bounded-missingness}. Similarly,

\begin{align*}
&\tilde{\Psi}(\bx) - \pi_0^\star(\bx)[\mu_1(\bx) - \mu_0(\bx)] \\
&\le \tilde{\Psi}(\bx) - \delta_{0u}\pi_0(\bx)[\mu_1(\bx) - \mu_0(\bx)]\I[\mu_1(\bx) \le \mu_0(\bx)].
\end{align*}
\end{proof}

\begin{remark}
Under assumptions 1-5 and \ref{asmpt:known-missingness} or \ref{asmpt:bounded-missingness}, the bounds $\ell_2$ and $u_2$ are sharp: $\ell_2$ is attained with equality when $\pi_0^\star(\bx)/\pi_0(\bx) = \delta_{0u}$ and $\mu_1(\bx) > \mu_0(\bx)$ for all $\bx$, and $u_2$ is attained with equality when $\pi_0^\star(\bx)/\pi_0(\bx) = \delta_{0u}$ and $\mu_1(\bx) \le \mu_0(\bx)$ for all $\bx$.
\end{remark}

\begin{proof}[Proof of Proposition 10: point identification of $\Psi_1, \Psi_2$]
These results immediately follows from applying \ref{asmpt:known-missingness} to the expressions derived in Propositions 7 and 8.
\end{proof}

\begin{proof}[Proof of Proposition \ref{prop:3a}]
The proof of Proposition \ref{prop:3a} follows directly from the proof of Proposition \ref{prop:3}, noting that $\frac{\mu_a^\star(\bx)}{\mu_a(\bx)} \le \min\{1/\mu_a(\bx), \tau_a\}$ implies that $\mu_a^\star(\bx) - \mu_a(\bx) \le \min\{1 - \mu_a(\bx), \mu_a(\bx)(\tau_a -1)\}$, and substituting this into the expression for the upper bounds on $\E[Y(a) \mid \bx]$ accordingly.  
\end{proof}

\subsection{Estimation}

\begin{proof}[Proof of Proposition 6]
We treat the covariates as discrete to simplify the derivations, defining $\pr(\bx) = \P(\bX = \bx)$, $m_{1a} = \sum_{\bx} \mu_a(\bx)\pr(\bx)$, $m_{2a} = \sum_{\bx} \pi_a(\bx)\pr(\bx)$. To generalize the result to cover all estimators given in Corollary \ref{corr:1}, we also define $m_{3aa'} = \sum_{\bx} \mu_a(\bx)\pi_{a'}(\bx)\pr(\bx)$. It is well-known that $\text{EIF}[\mathbb{E}[Y \mid \bX = \bx]] = \frac{\I(\bX = \bx)}{\pr(\bx)}\{Y - \mathbb{E}[Y \mid \bX = \bx]\}$. Applying this fact and the chain-rule, we obtain,

\begin{align*}
\text{EIF}[m_{1a}] &= \sum_{\bx}\text{EIF}[\mu_a(\bx)]\pr(\bx) + \mu_a(\bx)\text{EIF}[\pr(\bx)] \\
&= \frac{\I(C = 0, A = a)}{\P(C = 0, A = a \mid \bX)}[Y - \mu_a(\bX)] + \mu_a(\bX) - m_{1a}.
\end{align*}

\noindent Similarly, it is well-known that,

\begin{align*}
\text{EIF}[m_{2a}] &= \sum_{\bx}\text{EIF}[\pi_a(\bx)]\pr(\bx) + \pi_a(\bx)\text{EIF}[\pr(\bx)] \\
&= \frac{\I(A = a)}{\P(A = a \mid \bX)}[C - \pi_a(\bX)] + \pi_a(\bX) - m_{2a}.
\end{align*}

\noindent Finally,

\begin{align*}
\text{EIF}[m_{3aa'}] &= \sum_{\bx} \text{EIF}[\mu_a(\bx)]\pi_{a'}(\bx)\pr(\bx)) + \mu_a(\bx)\text{EIF}[\pi_{a'}(\bx)]\pr(\bx)) + \mu_a(\bx)\pi_{a'}(\bx)\text{EIF}[\pr(\bx))] \\
&= \frac{\I(C = 0, A = a)}{\P(C = 0, A = a \mid \bX)}[Y - \mu_a(\bX)]\pi_{a'}(\bX) \\&+ \frac{\I(A = a')}{\P(A = a' \mid \bX)}[C - \pi_{a'}(\bX)]\mu_a(\bX) + \mu_a(\bX)\pi_{a'}(\bX) - m_{3aa'}.
\end{align*}
\end{proof}

\begin{proof}[Proof of Proposition 11]
The result follows from applying the chain-rule and noting that, for a discrete $\bX$, $\eif[\E\{\Phi_\varepsilon(\mu_1(\bX)-\mu_0(\bX))\mid \bx\}] = \phi_{\varepsilon}(\mu_1(\bx)-\mu_0(\bx))\{\varphi_{1,1}(\bO)-\varphi_{1,0}(\bO) - \mu_1(\bx) + \mu_0(\bx)\}/\pr(\bx)$. 
\end{proof}

\begin{remark}
Corollary \ref{corr:1} follows immediately from the proof of Propositions 6 and 11. Because the influence function has mean-zero, so that the expressions above can be used to establish a moment condition establishes the estimator as the solution that sets the expression to zero \cite{kennedy2022semiparametric}. For example, applying this to $b_{3aa'}$, we see that

\begin{align*}
&\mathbb{E}[\eif[m_{3aa'}]] = 0 \\
&\implies \hat{m}_{3aa'} = \mathbb{P}_n\left[\frac{\I(C = 0, A = a)}{\hat{\P}(C = 0, A = a \mid \bX)}[Y - \hat{\mu}_a(\bx)]\hat{\pi}_{a'}(\bX) \right.\\ 
&+\left. \frac{\I(A = a')}{\hat{\P}(A = a' \mid \bX)}[C - \hat{\pi}_{a'}(\bX)]\hat{\mu}_a(\bX) + \hat{\mu}_a(\bX)\hat{\pi}_{a'}(\bX)\right]
\end{align*}
\end{remark}

\begin{proof}[Proof of Corollary \ref{corr:2}]
For the estimator $\P_n\hat{\varphi}_b$ of $b$, we can decompose the error into components,

\begin{align*}
\P_n\hat{\varphi}_b - \PP\varphi_b = \underbrace{(\P_n - \PP)(\hat{\varphi}_b - \varphi_b)}_{T_1} + \underbrace{(\P_n - \PP)\varphi_b}_{T_2} + \underbrace{\PP(\hat{\varphi}_b - \varphi_b)}_{T_3}
\end{align*}

\noindent (noting that $\PP\varphi_b = b$). The first term $T_1$ can be shown to be $o_p(n^{-1/2})$ as long as $\|\hat{\varphi}_b - \varphi_b\|^2 = o_p(1)$, using the law of iterated expectations, the fact that $\hat{\varphi}_b$ is fixed conditional on $D_0^n$, and Chebyshev's inequality \cite{kennedy2022semiparametric}. The second term is the centered empirical average of a fixed function; therefore, by the Central Limit Theorem, this term has distribution $\mathcal{N}(0, \text{Var}(\varphi_b) / n)$. It remains to analyze $T_3$: if we can show conditions where this term is $o_p(n^{-1/2})$, then the asymptotic distribution is characterized by $T_2$ alone. 

Let $\gamma_{0a}(\bx) = (1-\pi_a(\bx))e_a(\bx)$, and define $\hat{\gamma}_{0a}(\bx)$ analogously. To ease notation, we omit the dependence of the functions on $\bx$ (e.g. $\gamma_{0a} = \gamma_{0a}(\bx)$). Since any estimator $\hat{b}$ considered in Corollary \ref{corr:1} can be expressed as a finite linear combination of $[\hat{m}_{1a}, \hat{m}_{2a}, \hat{m}_{3aa'}]$ (for $a, a' = 0, 1$), the expressions for $R_2(\hat{b})$ follow from take the relevant combination of remainder terms, which we derive below.

First, by the law of iterated expectation, we see that,

\begin{align*}
\PP(\hat{\varphi}_{1a} - m_{1a}) &= \PP\left(\frac{\gamma_{0a}}{\hat{\gamma}_{0a}}(\mu_a - \hat{\mu}_a) + \hat{\mu}_a - \mu_a\right) \\
&= \PP\left(\frac{\gamma_{0a}}{\hat{\gamma}_{0a}}(\mu_a - \hat{\mu}_a) - \frac{\hat{\gamma}_{0a}}{\hat{\gamma}_{0a}}(\mu_a - \hat{\mu}_a)\right) \\
&= \PP\left(\frac{(\gamma_{0a} - \hat{\gamma}_{0a})(\mu_a - \hat{\mu}_a)}{\hat{\gamma}_{0a}}\right),
\end{align*}

\begin{align*}
\PP(\hat{\varphi}_{2a} - m_{2a}) &= \PP\left(\frac{e_{a}}{\hat{e}_{a}}(\pi_a - \hat{\pi}_a) + \hat{\pi}_a - \pi_a\right) \\
&= \PP\left(\frac{e_{a}}{\hat{e}_{a}}(\pi_a - \hat{\pi}_a) - \frac{\hat{e}_{a}}{\hat{e}_{a}}(\pi_a - \hat{\pi}_a)\right) \\
&= \PP\left(\frac{(e_{a} - \hat{e}_{a})(\pi_a - \hat{\pi}_a)}{\hat{e}_{a}}\right),
\end{align*}

\begin{align*}
\PP(\hat{\varphi}_{3aa'} - m_{3aa'}) &= \PP\left(\frac{\gamma_{0a}}{\hat{\gamma}_{0a}}(\mu_a - \hat{\mu}_a)\hat{\pi}_{a'} + \frac{e_{a'}}{\hat{e}_{a'}}(\pi_{a'} - \hat{\pi}_{a'})\hat{\mu}_a + \hat{\mu}_a\hat{\pi}_{a'} - \mu_a\pi_{a'}\right) \\
&= \PP\left(\frac{\gamma_{0a}}{\hat{\gamma}_{0a}}(\mu_a - \hat{\mu}_a)\hat{\pi}_{a'} - \underbrace{\frac{e_{a'}}{\hat{e}_{a'}}(\pi_{a'} - \hat{\pi}_{a'})(\mu_a - \hat{\mu}_a)}_{SO_1} \frac{e_{a'}}{\hat{e}_{a'}}\mu_a(\pi_{a'} - \hat{\pi}_{a'}) + \hat{\mu}_a\hat{\pi}_{a'} - \mu_a\pi_{a'}\right) \\
&= \PP\left(SO_1 - \underbrace{\frac{\gamma_{0a}}{\hat{\gamma}_{0a}}(\mu_a - \hat{\mu}_a)(\pi_{a'} - \hat{\pi}_{a'})}_{SO_2} \right.\\ &\left.+\frac{\gamma_{0a}}{\hat{\gamma}_{0a}}(\mu_a - \hat{\mu}_a)\pi_{a'} + \frac{e_{a'}}{\hat{e}_{a'}}\mu_a(\pi_{a'} - \hat{\pi}_{a'}) + \hat{\mu}_a\hat{\pi}_{a'} - \mu_a\pi_{a'}\right) \\
&= P\left(SO_1 + SO_2 + \underbrace{\frac{\mu_a(\pi_{a'} - \hat{\pi}_{a'})(e_{a} - \hat{e}_{a})}{\hat{e}_{a}}}_{SO_3} + \frac{\gamma_{0a}}{\hat{\gamma}_{0a}}\pi_{a'}(\mu_a - \hat{\mu}_a)  - \hat{\pi}_{a'}(\mu_a - \hat{\mu}_a)\right) \\
&= \PP\left(SO_1 + SO_2 + SO_3 + \underbrace{\pi_{a'}\frac{(\gamma_{0a} - \hat{\gamma}_{0a})(\mu_a - \hat{\mu}_a)}{\hat{\gamma}_{0a}}}_{SO_4} + \underbrace{(\mu_a - \hat{\mu}_a)(\pi_{a'} - \hat{\pi}_{a'}))}_{SO_5}\right).
\end{align*}

\noindent Additionally, notice that,
\begin{align*}
|\gamma_{0a} - \hat{\gamma}_{0a}| &= |(1 - \pi_a)e_a - (1 - \hat{\pi}_a)\hat{e}_a| \\
&= |e_a[(1 - \pi_a) - (1 - \hat{\pi}_a)] - (1 - \hat{\pi}_a)(\hat{e}_a - e_a)| \\
&\le |e_a - \hat{e}_a| + |\pi_a - \hat{\pi}_a| \\
&\implies (\gamma_{0a} - \hat{\gamma}_{0a})^2 \le (|e_a - \hat{e}_a| + |\pi_a - \hat{\pi}_a|)^2 \le 2[(e_a - \hat{e}_a)^2 + (\pi_a - \hat{\pi}_a)^2].
\end{align*}

\noindent Applying this result, Cauchy-Schwartz, and using boundedness of the propensity score, the missingness probabilities, and their estimates, we obtain:

\begin{align*}
|\PP(\hat{\varphi}_{1a} - \varphi_{1a})| &\lesssim \|\mu_a - \hat{\mu}_a\|(\|e_a - \hat{e}_a\| + \|\pi_a - \hat{\pi}_a\|) \\
|\PP(\hat{\varphi}_{2a} - \varphi_{2a})| &\lesssim \|\pi_a - \hat{\pi}_a\|\|e_{a} - \hat{e}_{a}\| \\
|\PP(\hat{\varphi}_{3aa'} - \varphi_{3aa'})| &\lesssim \|\mu_a - \hat{\mu}_a\|(\|\pi_a - \hat{\pi}_a\| + \|e_a - \hat{e}_a\| + \|\pi_{a'} - \hat \pi_{a'}\|) + \|e_{a} - \hat{e}_{a}\|\|\pi_{a'} - \hat{\pi}_{a'}\|
\end{align*}

\noindent noting that when $a = a'$, $\varphi_{3aa}=\varphi_{3a}$, and so,

\begin{align*}
|\PP(\hat{\varphi}_{3a} - \varphi_{3a})| &\lesssim \|\mu_a - \hat{\mu}_a\|(\|\pi_a - \hat{\pi}_a\| + \|e_a - \hat{e}_a\|) + \|e_{a} - \hat{e}_{a}\|\|\pi_{a} - \hat{\pi}_{a}\|
\end{align*}
\end{proof}

\begin{proof}[Proof of Theorem 2]
In an abuse of notation, we will omit the dependence of the quantities on $\varepsilon$ (so that $b_2$ is short-hand for $b_{2, \varepsilon}$). As with the proof of Theorem 1, we consider the decomposition,

\begin{align*}
\P_n\hat{\varphi}_{b_2} - \PP\varphi_{b_2} = \underbrace{(\P_n - \PP)(\hat{\varphi}_{b_2} - \varphi_{b_2})}_{T_1} + \underbrace{(\P_n - \PP)\varphi_{b_2}}_{T_2} + \underbrace{\PP(\hat{\varphi}_{b_2} - \varphi_{b_2})}_{T_3}.
\end{align*}

\noindent As explained previously, $T_1 = o_p(n^{-1/2})$ under the conditions we've assumed, can apply the CLT to $T_2$, and need to analyze $T_3$ and show the conditions when this is $o_p(n^{-1/2})$: this would imply the asymptotic distribution is characterized by $T_2$ alone. To do this we first consider $\ell_2$ and its estimator $\P_n\varphi_{\ell_2}(\bO; \hat{\eta})$. We also decompose $\varphi_{\ell_2}(\bO; \eta) = \varphi_{1,1}-\varphi_{1,0}-(\varphi_{\ell_2, 1} - \varphi_{\ell_2, 0})$, where, for example,

\begin{align*}
\varphi_{\ell_2, a}(\bO; \eta) &= \left(\frac{\I(C = 0, A = a)}{\P(C = 0, A = a \mid \bX)}\right)\{Y - \mu_a(\bX)\}\pi_0(\bX)\Phi_{\varepsilon}(\mu_1(\bX) - \mu_0(\bX)) \\
&+ \frac{\I(A = 0)}{\P(A = 0 \mid \bX)}\{C - \pi_0(\bX)\}\mu_a(\bX)\Phi_{\varepsilon}(\mu_1(\bX) - \mu_0(\bX)) \\
&+ \mu_a(\bX)\pi_0(\bX)\Phi_{\varepsilon}(\mu_1(\bX) - \mu_0(\bX))[\varphi_{1, 1}(\bO) - \varphi_{1, 0}(\bO) - \mu_1(\bX) + \mu_0(\bX)] \\
&+ \mu_a(\bX)\pi_0(\bX)\Phi_{\varepsilon}(\mu_1(\bX) - \mu_0(\bX))].
\end{align*}

\noindent Similarly, we let $\ell_2 = \tilde\Psi - (\ell_{2, 1} - \ell_{2, 0})$. To ease notation we let $\Phi = \Phi_{\varepsilon}(\mu_1(\bx) - \mu_0(\bx))$, $\hat{\Phi} = \Phi_{\varepsilon}(\hat{\mu}_1(\bx) - \hat{\mu}_0(\bx))$, and similarly with $\phi$ and $\hat{\phi}$. We again omit the dependence of the functions on $\bx$.Finally, we let $\Delta \mu = \mu_1 - \mu_0$.

We have already shown that $\PP(\hat\varphi_{1,1}(\bO)-\hat\varphi_{1,0}(\bO)-\varphi_{1,1}(\bO)+\varphi_{1,0}(\bO))$ is second-order in the nuisance estimation. We therefore now consider $\PP(\varphi_{\ell_2, 0}(\bO; \hat{\eta}) - \varphi_{\ell_2, 0}(\bO;\eta))$. Using the law of iterated expectations, we se that this is equal to,

\begin{align*}
&= \PP\left[\underbrace{\frac{\gamma_{00}}{\hat{\gamma}_{00}}\{\mu_0 - \hat{\mu}_0\}\hat{\pi}_0\hat{\Phi}}_{i} + \underbrace{\frac{e_0}{\hat{e}_0}\{\pi_0 - \hat{\pi}_0\}\hat{\mu}_0\hat{\Phi}}_{ii}
+ \underbrace{\hat{\mu}_0\hat{\pi}_0\hat{\phi}[\frac{\gamma_{01}}{\hat{\gamma}_{01}}(\mu_1 - \hat{\mu}_1) - \frac{\gamma_{00}}{\hat{\gamma}_{00}}(\mu_0 - \hat{\mu}_0)]}_{iii} 
+ \underbrace{\hat{\mu}_0\hat{\pi}_0\hat{\Phi}}_{iv} - \underbrace{\mu_0\pi_0\Phi}_{v}\right] .
\end{align*}

\noindent Consider $(ii) + (iv)$:

\begin{align*}
&= \frac{e_0}{\hat{e}_0}[\pi_0 - \hat{\pi}_0]\hat{\mu}_0\hat{\Phi} + \frac{\hat{e}_0}{\hat{e}_0}\hat{\mu}_0\hat{\pi}_0\hat{\Phi} \\
&= \frac{e_0}{\hat{e}_0}[\pi_0 - \hat{\pi}_0]\hat{\mu}_0\hat{\Phi} - \frac{\hat{e}_0}{\hat{e}_0}\hat{\mu}_0[\pi_0 - \hat{\pi}_0]\hat{\Phi} + \pi_0\hat{\mu}_0\hat{\Phi} \\
&= \underbrace{\frac{[e_0 - \hat{e}_0][\pi_0 - \hat{\pi}_0]}{\hat{e}_0}\hat{\mu}_0\hat{\Phi}}_{SO_1} + \underbrace{\pi_0\hat{\mu}_0\hat{\Phi}}_{r_1}
\end{align*}

\noindent Next consider $r_1 + (i)$:

\begin{align*}
&= \frac{\gamma_{00}}{\hat{\gamma}_{00}}(\mu_0 - \hat{\mu}_0)\hat{\pi}_0\hat{\Phi} - \frac{\hat{\gamma}_{00}}{\hat{\gamma}_{00}}\pi_0(\mu_0 - \hat{\mu}_0)\hat{\Phi} + \underbrace{\pi_0\mu_0\hat{\Phi}}_{r_2} \\
&= \underbrace{-\frac{\gamma_{00}}{\hat{\gamma}_{00}}(\mu_0 - \hat{\mu}_0)(\pi_0 - \hat{\pi}_0)\hat{\Phi}}_{SO_2} + \underbrace{\frac{(\gamma_{00} - \hat{\gamma}_{00})}{\hat{\gamma}_{00}}\pi_0(\mu_0 - \hat{\mu}_0)\hat{\Phi}}_{SO_3} + r_2
\end{align*}

\noindent Next consider $r_2 + (v)$. Let $\Delta f = f_1 - f_0$ for a function $f$ admitting arguments $1$ and $0$. Then by the Mean Value Theorem, we obtain:

\begin{align*}
\hat{\Phi} - \Phi &= \Phi_{\varepsilon}(\Delta \mu) - \Phi_{\varepsilon}(\Delta \hat{\mu}) \\
&= \Phi_{\varepsilon}(\Delta \mu^\star)[\hat{\mu}_1 - \hat{\mu}_0 - \mu_1 + \mu_0]
\end{align*}

\noindent where $\Delta \mu^\star$ lies between $\Delta \mu$ and $\Delta \hat{\mu}$. We can plug this into the expression for $r_2 + (v)$:

\begin{align*}
&= \pi_0\mu_0[\hat{\Phi} - \Phi] \\
&= \underbrace{\pi_0\mu_0\phi(\Delta \mu^\star) [\hat{\mu}_1 - \hat{\mu}_0 - \mu_1 + \mu_0]}_{r_3} 
\end{align*}

\noindent Finally, consider $r_3 + (iii)$:

\begin{align*}
&= \pi_0\mu_0\phi^\star[\hat{\mu}_1 - \hat{\mu}_0 - \mu_1 + \mu_0] + \hat{\mu}_0\hat{\pi}_0\hat{\phi}\left[\frac{\gamma_{01}}{\hat{\gamma}_{01}}(\mu_1 - \hat{\mu}_1) - \frac{\gamma_{00}}{\hat{\gamma}_{00}}(\mu_0 - \hat{\mu}_0)\right] \\
&= \underbrace{[(\mu_0 - \hat{\mu}_0) - (\mu_1 - \hat{\mu}_1)](\pi_0\mu_0\phi^\star - \hat{\pi}_0\hat{\mu}_0\hat{\phi})}_{SO_4} + \hat{\mu}_0\hat{\pi}_0\hat{\phi}\left[\frac{\gamma_{01}}{\hat{\gamma}_{01}}(\mu_1 - \hat{\mu}_1) - \frac{\gamma_{00}}{\hat{\gamma}_{00}}(\mu_0 - \hat{\mu}_0)\right] \\ &- \hat{\mu}_0\hat{\pi}_0\hat{\phi}\left[(\mu_1 - \hat{\mu}_1) - (\mu_0 - \hat{\mu}_0)\right]\\
&= SO_4 + \underbrace{\hat{\mu}_0\hat{\pi}_0\hat{\phi}\left[\frac{(\gamma_{01} - \hat{\gamma}_{01})(\mu_1 - \hat{\mu}_1)}{\hat{\gamma}_{01}} - \frac{(\gamma_{00} - \hat{\gamma}_{00})(\mu_0 - \hat{\mu}_0)}{\hat{\gamma}_{00}}\right]}_{SO_5}
\end{align*}

\noindent By Cauchy-Schwarz, the (assumed) boundedness of $\gamma_{0a}$, the (assumed) consistency of the estimated influence function, and the boundedness of $\Phi_{\varepsilon}$ for a fixed $\varepsilon$, we obtain:

\begin{align*}
|\PP[SO_1 + SO_2 + SO_3 + SO_5]| &\lesssim \sum_{a=0,1}\|\mu_a - \hat{\mu}_a\|(\|\pi_a - \hat{\pi}_a\| + \|e_a - \hat{e}_a\|) + \|\pi_0 - \hat{\pi}_0\|\|e_0 - \hat{e}_0\|
\end{align*}

\noindent Furthermore, by the Mean Value Theorem, notice that,

\begin{align*}
\phi^\star - \phi &= \Phi_{\varepsilon}(\Delta \mu^\star) - \Phi_{\varepsilon}(\Delta \mu) \\
&= \phi'_{\varepsilon}(\Delta \tilde{\mu})[\mu_1^\star - \mu_0^\star - \mu_1 + \mu_0] \\
\end{align*}

\noindent for some value of $\Delta \tilde{\mu}$ between $\Delta \mu$ and $\Delta \mu^\star$. Moreover, notice that $|\Delta\mu -\Delta\mu^\star| \le |\Delta \mu - \Delta \hat{\mu}|$, and recall that $[a - \hat{a}][bc - \hat{b}\hat{c}] = [a - \hat{a}](c[b - \hat{b}] + \hat{b}[c - \hat{c}])$ and that $\phi'_{\varepsilon}$ is bounded for a fixed $\epsilon$. Therefore,

\begin{align*}
|\PP[SO_4]| \lesssim \sum_{a=0,1}\|\mu_a - \hat{\mu}_a\|^2 + \|\mu_0 - \hat{\mu}_0\|(\|\mu_1 - \hat{\mu}_1\| + \|\pi_0 - \hat{\pi}_0\|) + \|\mu_1 - \hat{\mu}_1\|\|\pi_0 - \hat{\pi}_0\|
\end{align*}

Putting everything together we obtain,

\begin{align*}
|\PP(\hat{\varphi}_{\ell_2, 0} - \varphi_{\ell_2,0})| &\lesssim \sum_{a=0,1}\|\mu_a - \hat{\mu}_a\|(\|\pi_a - \hat{\pi}_a\| + \|e_a - \hat{e}_a\| + \|\mu_a - \hat{\mu}_a\|) \\
&+ \|\pi_0 - \hat{\pi}_0\|(\|e_0 - \hat{e}_0\| + \|\mu_1 - \hat{\mu}_1\|) + \|\mu_0 - \hat{\mu}_0\|\|\mu_1 - \hat{\mu}_1\|.
\end{align*}

\noindent Similar derivations give that an analogous expression for  $\PP(\hat{\varphi}_{\ell_2, 1} - \varphi_{\ell_2, 1})$. Finally, notice that $\hat{\varphi}_{u_2, 1}$ and $\hat{\varphi}_{u_2,0}$ have equivalent decompositions up to sign changes, so that these same bounds hold for these quantities, respectively, completing the proof.
\end{proof}

\begin{proof}[Proof of Theorem 3]
We largely follow the same steps in the proof above. However, here we consider $\PP(\hat{\varphi}_{\ell_{2,0}} - \varphi_{\ell_{2,0}})$ (with these terms defined analogously), which is easy to see and is equivalent to,

\begin{align*}
&= \PP\left[\underbrace{\frac{\gamma_{00}}{\hat{\gamma}_{00}}\{\mu_0 - \hat{\mu}_0\}\hat{\pi}_0\hat{\I}}_{i} + \underbrace{\frac{e_0}{\hat{e}_0}\{\pi_0 - \hat{\pi}_0\}\hat{\mu}_0\hat{\I}}_{ii} 
+ \underbrace{\hat{\mu}_0\hat{\pi}_0\hat{\I}}_{iv} - \underbrace{\mu_0\pi_0\I}_{v}\right], 
\end{align*}

\noindent where $\hat\I = \hat\I(\hat{\mu}_1(\bX) > \hat\mu_0(\bX))$. We can follow the same steps in the proof above to obtain that $(i)+(ii)+(iv)$ takes the form:

\begin{align*}
&= \underbrace{\frac{[e_0 - \hat{e}_0][\pi_0 - \hat{\pi}_0]}{\hat{e}_0}\hat{\mu}_0\hat{\I}}_{SO_{1,0}}  \\
&\underbrace{-\frac{\gamma_{00}}{\hat{\gamma}_{00}}(\mu_0 - \hat{\mu}_0)(\pi_0 - \hat{\pi}_0)\hat{\I}}_{SO_{2,0}} + \underbrace{\frac{1}{\hat{\gamma}_{00}}\pi_0(\mu_0 - \hat{\mu}_0)(\gamma_{00} - \hat{\gamma}_{00})\hat{\I}}_{SO_{3,0}} + \pi_0\mu_0(\hat{\I}-\I).
\end{align*}

\noindent We can follow the same steps to decompose $\ell_{2,1}$ to obtain,

\begin{align*}
\PP(\hat\varphi_{\ell_2,1}-\varphi_{\ell_2,1}) &= \PP\left[\underbrace{\frac{[e_0 - \hat{e}_0][\pi_0 - \hat{\pi}_0]}{\hat{e}_0}\hat{\mu}_1\hat{\I}}_{SO_{1,1}}  \right.\\
&\left.\underbrace{-\frac{\gamma_{01}}{\hat{\gamma}_{01}}(\mu_1 - \hat{\mu}_1)(\pi_0 - \hat{\pi}_0)\hat{\I}}_{SO_{2,1}} + \underbrace{\frac{1}{\hat{\gamma}_{01}}\pi_0(\mu_1 - \hat{\mu}_1)(\gamma_{01} - \hat{\gamma}_{01})\hat{\I}}_{SO_{3,1}} + \pi_0\mu_1(\hat{\I}-\I)\right].
\end{align*}

\noindent Putting both expressions together we are left with the final term,

\begin{align*}
\PP(\pi_0(\mu_1-\mu_0)(\hat{\I}-\I)).
\end{align*}

\noindent By Lemma 2 from \cite{bonvini2022sensitivity}, we know that $|\hat\I-\I| = \I\bigl(|\Delta\mu|<|\Delta\hat\mu-\Delta\mu|\bigr)$ for $\Delta\mu = \mu_1-\mu_0$. Therefore,

\begin{align*}
&|\PP(\pi_0(\mu_1-\mu_0)(\hat{\I}-\I))| \\
&\le |\PP(\pi_0|\mu_1-\mu_0|\I(|\Delta\mu|<|\Delta\hat\mu-\Delta\mu|))| \\
&\le |\PP(|\mu_1-\mu_0|\I(|\Delta\mu|<|\Delta\hat\mu-\Delta\mu|))| \\
&\lesssim \bigl(\|\hat\mu_1-\mu_1\|_\infty+\|\hat\mu_0-\mu_0\|_\infty\bigr)^{1+\alpha},
\end{align*}

\noindent where the final inequality we used the margin condition, the triangle inequality, and taking $t = \|\hat\mu_1-\mu_1\|_\infty+\|\hat\mu_0-\mu_0\|_\infty$. The exact same logic can be used to bound $\PP(\hat \varphi_{u_2}-\varphi_{u_2})$, thus completing the proof.
\end{proof}

\begin{proof}[Proof of Theorem \ref{thm:5}]
    We follow the same decomposition as in the proof of Theorem \ref{thm:3}. For either $b\in\{\tilde \ell_0^\star,\tilde u_0^\star\}$,
\begin{align*}
\mathbb{P}_n\hat\varphi_b-\PP\varphi_b = \underbrace{(\mathbb{P}_n-\PP)(\hat\varphi_b-\varphi_b)}_{T_1}
+ \underbrace{(\mathbb{P}_n-\PP)\varphi_b}_{T_2} + \underbrace{\PP(\hat\varphi_b-\varphi_b)}_{T_3}.
\end{align*}
As before, $T_1=o_p(n^{-1/2})$ under the assumed $L_2(\P)$ consistency of the estimated influence
functions and sample splitting (or cross-fitting), and $T_2$ is asymptotically normal by the Central Limit Theorem. It remains to bound $T_3$. We begin with the upper bound. Let
\[
\I_1(\bX)=\I\{\tau_1\mu_1(\bX)>1\},\qquad
\hat \I_1(\bX)=\I\{\tau_1\hat\mu_1(\bX)>1\}.
\]
Then
\begin{align*}
\tilde u_0 &=
\PP\Big[
\varphi_{1,1}-\varphi_{1,0}
+\delta_{1u}\big\{(\varphi_{2,1}-\varphi_{3,1})\I_1+(\tau_1-1)\varphi_{3,1}(1-\I_1)\big\}
-\delta_{0u}(\tau_0^{-1}-1)\varphi_{3,0}
\Big].
\end{align*}
Replacing nuisance functions by their estimates gives the analogous expression with $\hat \I_1$.
Subtracting the two expressions, the smooth terms are exactly linear combinations of the same
remainder terms already analyzed in Corollary \ref{corr:2}, yielding the bound on these quantities,
\[
\sum_{a=0,1}\|\mu_a-\hat\mu_a\|
\bigl(\|\pi_a-\hat\pi_a\|+\|e_a-\hat e_a\|\bigr)
+\sum_{a=0,1}\|e_a-\hat e_a\|\|\pi_a-\hat\pi_a\|.
\]
It remains to analyze the non-smooth term coming from replacing $\I_1$ by $\hat \I_1$. At the
population level this term is
\begin{align*}
\PP\Big(
\delta_{1u}\big[(\pi_1-\pi_1\mu_1)-(\tau_1-1)\pi_1\mu_1\big](\hat \I_1-\I_1)
\Big)
=
\PP\Big(
\delta_{1u}\pi_1(1-\tau_1\mu_1)(\hat \I_1-\I_1)
\Big).
\end{align*}
By the same argument as in the proof of Theorem \ref{thm:3},
\[
\hat \I_1-\I_1 = \I\Big(\big|\tau_1\mu_1(\bX)-1\big| < \tau_1\big|\hat\mu_1(\bX)-\mu_1(\bX)\big|\Big),
\]
so that
\begin{align*}
&\left|
\PP\Big(
\delta_{1u}\pi_1(1-\tau_1\mu_1)(\hat \I_1-\I_1)
\Big)
\right| \\
&\lesssim
\|\hat\mu_1-\mu_1\|_\infty
\,\P\!\left(
\big|\tau_1\mu_1(\bX)-1\big|
<
\tau_1\|\hat\mu_1-\mu_1\|_\infty
\right) \\
&\lesssim
\bigl(\|\hat\mu_1-\mu_1\|_\infty\bigr)^{1+\alpha_1},
\end{align*}
where the final inequality uses the stated margin condition and taking $t = \|\hat\mu_1-\mu_1\|_\infty$. The proof for the lower bound is identical. Combining the smooth and non-smooth pieces yields the stated bound on $R_2(\hat{\tilde b}_0)$, and the asymptotic normality result follows.
\end{proof}

\subsection{Uniform inference}
\setcounter{theorem}{3}

\begin{proof}[Proof sketch of Theorem \ref{thm:4}]
    This proof closely follows \cite{mauro2020instrumental}, which in turn follows \cite{kennedy2019nonparametric}, noting that the latter contains many of the technical details. Define $\|\cdot\|_\Delta := \sup_{\Theta^\star\in\Delta} |\cdot|.$. First, by variance consistency and Slutsky's Theorem, we have that

    \begin{align*}
        \left\|\frac{\sqrt{n}(\hat b(\Theta^\star) - b(\Theta^\star))}{\hat \sigma(\Theta^\star)} - \frac{\sqrt{n}(\P_n - \PP)(\varphi(\Theta^\star))}{\sigma(\Theta^\star)}\right\|_{\Delta} = \left\|\frac{\sqrt{n}(\hat b(\Theta^\star) - b(\Theta^\star))}{\sigma(\Theta^\star)} - \frac{\sqrt{n}(\P_n - \PP)(\varphi(\Theta^\star))}{\sigma(\Theta^\star)}\right\|_{\Delta} + o_p(1).
    \end{align*}
    Then,
        \begin{align*}
        \frac{\sqrt{n}(\hat b(\Theta^\star) - b(\Theta^\star))}{\sigma(\Theta^\star)} - \frac{\sqrt{n}(\P_n - \PP)(\varphi(\Theta^\star))}{\sigma(\Theta^\star)} = \frac{\sqrt{n}}{\sigma(\Theta^\star)}\left\{\underbrace{(\P_n - \PP)(\hat\varphi(\Theta^\star) - \varphi(\Theta^\star))}_{B_{1,\Theta^\star}} + \underbrace{\PP(\hat\varphi(\Theta^\star)-\varphi(\Theta^\star))}_{B_{2,\Theta^\star}}\right\}
    \end{align*}
    Standard empirical process arguments as in \cite{kennedy2019nonparametric} imply that $\sup_{\Theta^\star\in\Delta}|\frac{\sqrt{n}}{\sigma(\Theta^\star)}B_{1,\Theta^\star}|= o_p(1)$. Moreover, we have by assumption that $\sup_{\Theta^\star\in\Delta}|\frac{\sqrt{n}}{\sigma(\Theta^\star)}B_{2,\Theta^\star}| = o_p(1)$. This completes the proof.  
\end{proof}

\newpage

\section{Joint sensitivity analyses}\label{app:adtl}

In this section we sketch an approach to conduct a joint sensitivity analysis with respect to both the MAR and no unmeasured confounding assumptions. We consider $\Psi_0, \Psi_1$, and $\Psi_2$ separately in each section below for the case of a binary outcome $Y$. Throughout we let the set of nuisance functions be $\eta = \{\mu_1(\bX), \mu_0(\bX), \pi_1(\bX), \pi_0(\bX), e(\bX)\}$, defined in the main text as $\mu_a(\bx) = \E[Y \mid A = a, C = 0, \bX = \bx]$, $\pi_a(\bx) = \P(C = 1 \mid A = a, \bX = \bx)$, and $e(\bx) = \P(A = 1 | \bX = \bx)$. Let $\omega = \P(A = 1).$ At times it will be convenient to use the alternative notation $e_a(\bx) = \P(A = a | \bX = \bx)$. Finally, let $\lambda_a(\bx) = \E\{Y(a) | \bX = \bx, A = 1-a\}$ and define $\mubar_a(\bx) = \E[Y \mid A = a, \bX = \bx] = \mu_a(\bx)\{1 - \pi^*_a(\bx)\} + \mu^*_a(\bx)\pi^*_a(\bx) = \mu_a(\bx) + \pi^*_a(\bx)\{\mu_a^\star(\bx) - \mu_a(\bx)\}$ (a result we previously showed in the proof of Proposition 1).

\subsection{General bounds on $\Psi_0$}

We begin by considering $\Psi_0$. We first present general bounds on $\Psi_0$.

\begin{proposition}
Under Assumptions 1 (non-informative missingness by $U_{NI}$), 3 (consistency), and 4 (positivity), we have 
\begin{align}
    \ellbar_0 \leq \Psi_0 \leq \ubar_0
    \label{eqn:identity2}
\end{align} for
\begin{align*}
\ellbar_0 &= - \omega + \E\left(\mu_1(\bX)\{1 - \pi_1(\bX)\}e(\bX) - [\mu_0(\bX) + \pi_0(\bX)\{1 - \mu_0(\bX)\}]\{1 - e(\bX)\}\right)\\
\ubar_0 &= \omega + \E\left(\left[\mu_1(\bX) + \pi_1(\bX)\{1 - \mu_1(\bX)\}\right]e(\bX) - \mu_0(\bX)\{1 - \pi_0(\bX)\}\{1 - e(\bX)\}\right)
\end{align*}
\end{proposition}

\begin{proof}
First, notice that

\begin{align*}
\P\{Y(a) = 1\} &= \E\left[\mubar_{a}(\bX) + e_{1-a}(\bX)\{\lambda_{a}(\bX) - \mubar_{a}(\bX)\}\right] \\
&\le \E\left[\mubar_{a}(\bX) + e_{1-a}(\bX)\{1 - \mubar_{a}(\bX)\}\right] \\
&= \E\left\{\mubar_a(\bX)e_a(\bX) + e_{1-a}(\bX)\right\} \\
&\le \E\left([\mu_a(\bX) + \pi_a^\star(\bX)\{1 - \mu_a(\bX)\}]e_a(\bX) + e_{1-a}(\bX)\right) \\
&\le \E\left([\mu_a(\bX) + \pi_a(\bX)\{1 - \mu_a(\bX)\}]e_a(\bX) + e_{1-a}(\bX)\right)
\end{align*}

\noindent Next, notice that

\begin{align*}
\P\{Y(a) = 1\} &= \E\left[\mubar_{a}(\bX) + e_{1-a}(\bX)\{\lambda_{a}(\bX) - \mubar_{a}(\bX)\}\right] \\
&\ge \E\left\{\mubar_{a}(\bX) - e_{1-a}(\bX)\mubar_{a}(\bX)\right\} \\
&= \E\left\{\mubar_a(\bX)e_a(\bX)\right\} \\
&= \E\left(\left[\mu_a(\bX) + \pi^*_a(\bX)\{\mu_a^\star(\bx) - \mu_a(\bx)\}\right]e_a(\bX)\right) \\
&\ge \E\left[\mu_a(\bX)\{1 - \pi^*_a(\bX)\}e_a(\bX)\right] \\
&\ge \E\left[\mu_a(\bX)\{1 - \pi_a(\bX)\}e_a(\bX)\right].
\end{align*}
Combining these two expressions yields the result.
\end{proof}

\begin{remark}
We can also invoke Assumptions \ref{asmpt:known-missingness}-\ref{asmpt:known-risk} and analogous assumptions with respect to the counterfactual quantities $\mu_a(\bx)$ and $\lambda_a(\bx)$ to further narrow these bounds. For example, consider the following monotonicity assumptions,

\setcounter{assumption}{9}
\begin{intassumption}[Counterfactual monotonicity, positive]\label{asmpt:monotonicity-pos-a}
    $\P\{\mu_1(\bX) \geq \mu_{0}(\bX)\} = 1, \quad \P\{\lambda_1(\bX) \geq \lambda_0(\bX)\} = 1 \qquad a = 0, 1$.
\end{intassumption}
\begin{intassumption}[Counterfactual monotonicity, negative]\label{asmpt:monotonicity-neg-a}
    $\P\{\mu_1(\bX) < \mu_{0}(\bX)\} = 1, \quad \P\{\lambda_1(\bX) < \lambda_0(\bX)\} = 1 \qquad a = 0, 1$.
\end{intassumption}
\noindent Positive counterfactual outcome monotonicity assumes that people with greater counterfactual risk of the outcome are always more likely to select into treatment. Similarly, we could invoke known and bounded counterfactual outcome risk assumptions impose assumptions on the ratio $\lambda_a(\bx) / \mu_a(\bx)$ similar to assumptions 6.c and 6.d. In either case, these assumptions allow us to bound or rewrite \eqref{eqn:identity2} in terms of the quantities $\mubar_a(\bX), e(\bX)$. For example, assume 9.a. Then it is easy to show that,

\begin{align*}
\Psi_0 &\ge \tilde{\Psi} + \E[\pi_1(\bX)[1 - \mu_1(\bX)]],   \\
\Psi_0 &\le \E[\mu_1(\bX)e(\bX) - \{\mu_0(\bX) + \pi_0(\bX)(1-\mu_0(\bX))\}\{1-e(\bX)\}] - \omega.
\end{align*}

\noindent We could similarly invoke other assumptions to possibly narrow these bounds, or rewrite them as a function of additional sensitivity parameters.
\end{remark}

In contrast to the results in the main paper, the expressions for bounds when relaxing the no unmeasured confounding assumption are often also a function of the propensity score $e(\bX)$. This leads to a more complicated expression for the corresponding one-step estimators. Yet all bounds remain linear combinations of the elements of $\eta$. We therefore give expressions for the uncentered influence functions of additional relevant quantities below in Proposition \ref{prop:new}.

\setcounter{proposition}{11}

\begin{proposition}\label{prop:new}
The uncentered influence functions for the functionals $\E\{\mu_a(\bX)e_a(\bX)\}$, $\E\{\pi_a(\bX)e_a(\bX)\}$, $\E\{\mu_a(\bX)\pi_{a}(\bX)e_a(\bX)\}$, and $\E[e_a(\bX)]$ are given by $\varphi_{4, a}(\bO; \eta), \varphi_{5, a}(\bO; \eta)$, $\varphi_{6, a}(\bO; \eta)$, $\varphi_{7, a}(\bO; \eta)$, respectively, where
\begin{align*}
\varphi_{4, a}(\bO; \eta) &= \varphi_{1, a}(\bO; \eta)e_a(\bX) + \varphi_{7, 1}(\bO; \eta)\mu_a(\bX) \\
\varphi_{5, a}(\bO; \eta) &= \varphi_{2, a}(\bO; \eta)e_a(\bX) + \varphi_{7, 1}(\bO; \eta)\pi_a(\bX) \\
\varphi_{6, a}(\bO; \eta) &= \varphi_{3, a}(\bO; \eta)e_a(\bX) + \varphi_{7, 1}(\bO; \eta)\mu_a(\bX)\pi_a(\bX) \\
\varphi_{7, 1}(\bO; \eta) &= \I(A_i = a)
\end{align*}
\end{proposition}

This result follows from a simple application of the chain-rule \cite{kennedy2022semiparametric}; we omit the full derivation for brevity. One-step estimators for linear combinations of any of the functionals noted in Propositions 6 and \ref{prop:new} can be constructed using these results. Extending the conditions and results of Theorems 1 and 2 to show the asymptotic normality of the resulting estimator is straightforward, though we also omit this for brevity.

\subsection{Sensitivity analyses for $\Psi_1$}

\begin{proposition}\label{prop:new2}
Under Assumptions 1, 3, 4, and 6.d, we have $\bar{\ell}_1 \leq \Psi_1 \leq \bar{u}_1$, where

\begin{align*}
\bar{\ell}_1 &= \E\{\P(A = 1 \mid \bX)[\delta_{1\ell}\pi_1(\bX)(1-\mu_1(\bx)) + \mu_1(\bX)]\} \\
&- \E\{1 - (1 - \mu_0(\bX))(1 - \delta_{0u}\pi_0(\bX))(1 + \P(A = 1 \mid \bX))\},\\
\ubar_1 &= \E\left(1 - \{1 - \mu_1(\bX)\}\{1 - \delta_{1u}\pi_1(\bX)\}\left[1 + \{1-e(\bX)\}\right]\right) \\
&- \E\{\P(A = 0 \mid \bX)[\delta_{0\ell}\pi_0(\bX)(1-\mu_0(\bX)) + \mu_0(\bX)]\}.
\end{align*}
\end{proposition}

\begin{proof}[Proof of Proposition \ref{prop:new2}]
First, notice that
\begin{align*}
\P(Z(a) = 1 \mid \bx) &= \P(Z(a) \mid \bx, A = a) \\ &+ \P(A = a' \mid \bx)(\P(Z(a) = 1 \mid \bx, A = a') -\P(Z(a) = 1 \mid \bx, A = a)) \\
&\le \P(Z(a) \mid \bx, A = a) + e_{1-a}(\bx)(1 - \P(Z(a) = 1 \mid \bx, A = a)) \\
&= 1 - \P(Y(a) = 0, U_I(a) = 0 \mid \bx, A = a)(1 + e_{1-a}(\bx))) \\
&= 1 - [1 - \mu_a(\bx)][1 - \pi^*_a(\bx)](1 + e_{1-a}(\bx)) \\
&\le 1 - [1 - \mu_a(\bx)][1 - \delta_{au}\pi_a(\bx)](1 + e_{1-a}(\bx))
\end{align*}

\noindent Next, notice that
\begin{align*}
\P(Z(a) = 1 \mid \bx) &= \P(Z(a) = 1 \mid \bx, A = a) \\&+ e_{1-a}(\bx)(\P(Z(a) = 1 \mid \bx, A = a') -\P(Z(a) = 1 \mid \bx, A = a)) \\
&\ge \P(Z(a) = 1 \mid \bx, A = a)e_a(\bx) \\
&= e_a(\bx)[1 - (1 - \mu_a(\bx))(1 - \pi_a^\star(\bx))] \\
&\ge e_a(\bx)[\delta_{a\ell}\pi_a(\bx)(1 - \mu_a(\bx)) + \mu_a(\bx)]
\end{align*}
\end{proof}

\begin{remark}
We can again narrow these bounds by invoking additional assumptions that mirror assumptions \ref{asmpt:known-missingness}-\ref{asmpt:known-risk}. For example, consider the monotonicity assumption,

\begin{assumption}[Joint counterfactual monotonicity, positive]\label{asmpt:monotonicity-pos-b}
    $\P\{\P(Z(a) = 1 \mid \bX, A = 1) \geq \P(Z(a) = 1 \mid \bX, A = 0)\} = 1, \qquad a = 0, 1$.
\end{assumption}

This assumption simply states that the potential risk of the outcome or informative missingness is greater in the treated group than the control group. Combined with assumption \ref{asmpt:known-missingness} or \ref{asmpt:bounded-missingness}, it is easy to show that

\begin{align*}
\Psi_1 &\le  \tilde{\Psi} + \E\{\delta_{1u}\pi_1(\bX)(1 - \mu_1(\bX)) - \delta_{0\ell}\pi_0(\bX)(1 - \mu_0(\bX))\}\\
\Psi_1 &\ge \E\{e_1(\bX)[\mu_1(\bX) + \delta_{1\ell}\pi_1(\bX)(1 - \mu_1(\bX))] - \omega \\
&- e_0(\bX)[\mu_0(\bX) + \delta_{0u}\pi_0(\bX)(1 - \mu_0(\bX))]\}, 
\end{align*}

\noindent noting that $\delta_{au} = 1, \delta_{al} = 0, \quad a = 0, 1$ holds under no assumption (i.e. without \ref{asmpt:known-missingness} or \ref{asmpt:bounded-missingness}).

\end{remark}

\subsection{Sensitivity analyses for $\Psi_2$}

A similar approach could be taken to bound $\Psi_2$. However, as shown below, this becomes more challenging in part because the expression depends on the quantity $\P(A = 1 \mid U_I(0) = 0, \bX)$. Additional assumptions are therefore needed to bound $\Psi_2$ as a function of the observed data distribution. Providing a more complete characterization of bounds on $\Psi_2$ while relaxing no unmeasured confounding would be an interesting avenue for future research.

\begin{align*}
\P(Y(a,0) = 1 \mid \bx) &= \P(Y(a, 0) = 1 \mid U_I(a, 0) = 0, \bx)\P(U_I(a, 0) = 0 \mid \bx) \\
 &= \P(Y(a, a) = 1 \mid U_I(a, a) = 0, \bx)\P(U_I(0, 0) = 0 \mid \bx) \\
 &= \P(Y(a) = 1 \mid U_I(a) = 0, \bx)\P(U_I(0) = 0 \mid \bx) \\
 &= \{\mu_a(\bX) + \P(A = 1-a \mid \bx, U_I(a) = 0)[\P(Y(a) = 1 \mid U_I(a) = 0, \bx, A = 1-a) \\
 &- \mu_a(\bX)]\}\{1 - \P(U_I(0) = 1 \mid \bx)\}
\end{align*}
\newpage

\section{Simulation specifications}\label{app:dgp}

To run the simulations in Section \ref{sec:simulations}, we generate data under a potential outcomes framework with covariates, treatment assignment, and two types of censoring: non-informative censoring and informative censoring. All code for the simulations is available at \url{https://github.com/mrubinst757/mixture-missingness}.

\subsection{Covariates and treatment}
For each individual $i=1,\dots,N$, we generate a covariate
\[
X_i \sim \text{Unif}(-3,3),
\]
and assign treatment according to
\[
A_i \sim \text{Bernoulli}(e(X_i)), \qquad e(X) = \text{expit}(X).
\]

\subsection{Censoring mechanisms.}
Let $C(a)$ denote the censoring indicator under treatment $a \in \{0,1\}$. We first define the baseline censoring probability
\[
\gamma_{a}(X) = \mathbb{P}(C(a)=1 \mid X),
\]
with functional form
\[
\gamma_{1}(X) =
\begin{cases}
\text{expit}(X), & X < 0,\\
\text{expit}(0.8X), & 0 \le X < 1,\\
\text{expit}(0.2 + 0.6X), & X \ge 1,
\end{cases}
\]
and
\[
\gamma_{0}(X) =
\begin{cases}
\text{expit}(0.6X), & X < 0,\\
\text{expit}(0.5X), & 0 \le X < 1,\\
\text{expit}(0.1 + 0.4X), & X \ge 1.
\end{cases}
\]
We partition the censoring into non-informative and informative components using sensitivity parameters $\delta_a \in [0,1]$:
\[
\mathbb{P}(U_I(a)=1 \mid X) = \delta_a \gamma_{a}(X),
\]
\[
\mathbb{P}(U_{NI}(a)=1 \mid X) = (1-\delta_a)\gamma_{a}(X),
\]
where $U_{NI}(a)$ denotes non-informative censoring and $U_I(a)$ denotes informative censoring. These indicators are generated sequentially to ensure they cannot both occur:
\[
U_{NI}(a) \sim \text{Bernoulli}((1-\delta_a)\gamma_{c,a}(X)),
\]
\[
U_I(a) \sim \text{Bernoulli}\!\left(
\frac{\delta_a \gamma_{a}(X)}{1-(1-\delta_a)\gamma_{a}(X)}
\right)
\quad\text{if } U_{NI}(a)=0,
\]
and $U_I(a)=0$ otherwise. The overall censoring indicator is
\[
C(a) = \mathbb{I}\{U_{NI}(a)=1 \ \text{or}\ U_I(a)=1\}.
\]

\subsection{Outcome model}
Let
\[
\mu_a(X,u) = \mathbb{P}(Y(a)=1 \mid X, U_I(a)=u)
\]
denote the conditional outcome probability given the informative censoring indicator. We impose a sensitivity parameter $\tau_a$ governing the association between the outcome and informative censoring:
\[
\frac{\mu_a(X,1)}{\mu_a(X,0)} = \frac{\mu_a^\star(X)}{\mu_a(X)}=\tau_a.
\]
where the first equality holds by randomization (or ignorability) of $A$. To ensure probabilities lie in $[0,1]$, we generate
\[
\mu_a(X,1) = \text{expit}(\beta X), \qquad
\mu_a(X,0) = \mu_a(X,1)/\tau_a,
\]
when $\tau_a \ge 1$, and
\[
\mu_a(X,0) = \text{expit}(\beta X), \qquad
\mu_a(X,1) = \tau_a \mu_a(X,0),
\]
when $\tau_a < 1$. We then generate the potential outcomes as
\[
Y(a) \sim \text{Bernoulli}(\mu_a(X,U_2(a))).
\]
Finally, we apply consistency to obtain the observed outcome,
\[
Y = A Y(1) + (1-A)Y(0),
\]
and the observed censoring indicator,
\[
C = A C(1) + (1-A)C(0).
\]
\subsection{Parameterization}
We then set $N=2,000,000$, the outcome model coefficient $\beta = 1$, and sensitivity parameters that obey assumptions (\ref{asmpt:known-missingness}) and (\ref{asmpt:known-risk}), with $\delta_1 =0.5$, $\delta_0 = 0.2$, $\tau_1 = 0.5$, and $\tau_0 = 2$. We then conduct the simulations as described in Section \ref{sec:simulations}.

\end{document}

%% file: GrandMacros.tex
\def\bc{{\bf c}}

\def\bo{{\bf o}}

\def\bx{{\bf x}}

\def\bO{{\bf O}}

\def\bX{{\bf X}}

\def\thick#1{\hbox{\rlap{$#1$}\kern0.25pt\rlap{$#1$}\kern0.25pt$#1$}}

\def\btheta{\boldsymbol{\theta}}

\def\smbalpha{\boldsymbol{{\scriptstyle{\alpha}}}}

\def\ehat{{\widehat e}}

\def\etahat{{\widehat\eta}}
\def\thetahat{{\widehat\theta}}

\def\muhat{{\widehat\mu}}

\def\pihat{{\widehat\pi}}

\def\Psihat{{\widehat\Psi}}

\def\psitilde{{\widetilde\psi}}

\def\Psitilde{{\widetilde\Psi}}

\def\smbalpha{\widehat{\smbalpha}}

\def\hbar{\bar{ h}}

\def\ubar{\bar{ u}}

\def\Osc{{\cal O}}

\def\transpose{{\sf \scriptscriptstyle{T}}}

\def\E{\mbox{E}}
\def\pr{P}

\def\trans{^{\transpose}}
\def\inv{^{-1}}

\def\mybox#1{\vskip1mm \begin{center}
        \hspace{.0\textwidth}\vbox{\hrule\hbox{\vrule\kern6pt
\parbox{.9\textwidth}{\kern6pt#1\vskip6pt}\kern6pt\vrule}\hrule}
        \end{center} \vskip-5mm}
\def\lboxit#1{\vbox{\hrule\hbox{\vrule\kern6pt
      \vbox{\kern6pt#1\vskip6pt}\kern6pt\vrule}\hrule}}

\def\thickboxit#1{\vbox{{\hrule height 1mm}\hbox{{\vrule width 1mm}\kern6pt
          \vbox{\kern6pt#1\kern6pt}\kern6pt{\vrule width 1mm}}
               {\hrule height 1mm}}}

\def\fat#1{\hbox{\rlap{$#1$}\kern0.25pt\rlap{$#1$}\kern0.25pt$#1$}}

%% file: Macro.tex
\newcommand{\I}{{\sf I}}